%% file: main.tex
\documentclass{article}
\usepackage[utf8]{inputenc}
\usepackage{amsmath,amssymb,bbm,amsthm,appendix,mathtools}
\usepackage[normalem]{ulem}
\usepackage[numbers,sort]{natbib}
\usepackage{graphicx}
\usepackage[scr=rsfs]{mathalpha}
\usepackage{cancel}
\usepackage[margin=1.2in]{geometry}
\usepackage[dvipsnames]{xcolor}
\usepackage{footmisc}

\usepackage{color}

\newcommand{\loc}{\mathrm{loc}}
\newcommand{\Sp}{\mathbb{S}}
\newcommand{\id}{\mathrm{id}}

\newcommand{\Mw}{M_{\mathrm{wns}}}
\newcommand{\gw}{g_{\mathrm{wns}}}
\newcommand{\rd}{\partial}

\newcommand{\R}{\mathbb{R}}
\newcommand{\ssubset}{\subset\joinrel\subset}
\usepackage{todonotes}

\usepackage{marginnote}

\newtheorem{definition}{Definition}[section]
\newtheorem{example}[definition]{Example}
\newtheorem{prop}[definition]{Proposition}
\newtheorem{remark}[definition]{Remark}
\newtheorem{lemma}[definition]{Lemma}
\newtheorem{thm}[definition]{Theorem}

\newtheorem{cor}[definition]{Corollary}
\usepackage{graphicx} % Required for inserting images
\mathtoolsset{showonlyrefs,showmanualtags}
\title{On the uniqueness of continuous spacetime extensions in 1+1-dimensions with applications to weak null singularities}

\author{\bigskip Peter Cameron\footnote{Department of Mathematics, Imperial College London, South Kensington Campus, London SW7 2AZ, United Kingdom and the Heilbronn Institute for
Mathematical Research, Bristol, UK. Email: p.cameron24@imperial.ac.uk} \hspace{1pt} \& Jan Sbierski\footnote{School of Mathematics and Maxwell Institute for Mathematical Science, University of Edinburgh, James Clerk Maxwell Building, Peter Guthrie Tait Road, Edinburgh EH9 3FD, United Kingdom. Email: jan.sbierski@ed.ac.uk}
}
\date{\today}

\begin{document}

\maketitle

\begin{abstract}
   Motivated by weak null singularities in black hole interiors, we study 1+1-dimensional Lorentzian manifolds $(M,g)$ which admit a continuous spacetime extension across a null boundary $v=0$, where $v<0$ is a null coordinate. We study the degree to which such  extensions are unique \emph{up to the boundary}. Firstly, we find that in general not even the $C^0$-structure of the extension is uniquely determined by the assumption that the metric extends continuously. However, we exhibit an interesting local-global relation regarding the $C^0$-structure which in particular entails its rigidity for ``strongly spherically symmetric'' continuous extensions across the Cauchy horizon of the Reissner-Nordstr\"{o}m spacetime. Secondly, we construct continuous extensions which have the same $C^0$-structure, but do \emph{not} have equivalent $C^1$-structures. This construction also carries over to weak null singularities in 3+1-dimensions. Understanding the uniqueness properties of continuous spacetime extensions to the boundary is of importance for the study of low-regularity inextendibility problems \cite{Sbie24}.
\end{abstract}

\tableofcontents

\section{Introduction}\label{Introduction}

A weak null singularity is a singular null boundary of a spacetime $(M,g)$ solving the Einstein equations such that the metric is continuous up to the boundary but the Christoffel symbols fail to be square integrable. These are expected to form in the interior of generic rotating black holes. For a detailed background to weak null singularities we refer the reader to \cite{Luk18,DafLuk17,Sbie24}.  
 Here, we focus on a local piece $(\Mw,\gw)$ of the spacetime near the weak null singularity, where $\Mw = (-1,1) \times (-1,0) \times \Sp^2$ with coordinates $(u, v)$ on the first two factors and  $(\theta^1, \theta^2) = \theta^A$, $A=1,2$, an arbitrary set of smooth coordinates on $\Sp^2$. The metric is given in the double null form
\begin{equation} \label{EqGDN}
\gw = -\Omega^2(du \otimes dv + dv \otimes du) + \gamma_{AB} (d\theta^A - b^A d v) \otimes (d\theta^B - b^B dv)\;,
\end{equation}
where $\Omega$ is a smooth strictly positive function on $\Mw$, $\gamma(u,v)$ is a smooth Riemannian metric on $\Sp^2$ which depends smoothly on $u$ and $v$, and $b(u,v)$ is a smooth vector field on $\Sp^2$ which also depends smoothly on $u$ and $v$. The time orientation is fixed by stipulating that $\rd_u$ is future-directed null. The ``weakness'' of the singularity is signalled by the fact that the metric extends continuously to the singularity at $\{v = 0\}$ with respect to the $(u,v, \theta^1, \theta^2)$-differentiable structure. This means $\Omega$, $\gamma_{AB}$ and $b^A$ extend continuously to $\overline{\Mw} :=(-1,1) \times (-1,0] \times \Sp^2 \supseteq \Mw$ such that $\Omega$ is strictly positive and, for fixed $u$ and $v$, $\gamma_{AB}$ is a Riemannian metric on $\Sp^2$ and $b^A$ is a vector field on $\Sp^2$. Note that $\{v=0\}$ is a null hypersurface. The existence of this continuous extension to $\{v=0\}$, together with upper bounds on various Christoffel symbols, has been established in \cite[Theorem 4.24]{DafLuk17}. Moreover, it is conjectured and expected (cf.\ \cite{Chris09, DafLuk17, Luk18, Sbie23, Gur24} and references therein) that the metric $\gw$ indeed becomes singular for $v \to 0$ and that the Lorentzian manifold $(\Mw, \gw)$ is inextendible `across' $\{v = 0\}$ as a Lorentzian manifold with a continuous metric and locally square integrable Christoffel symbols. Here, it is helpful to recall the following definitions: consider a Lorentzian manifold $(M,g)$, a regularity class $\Gamma$ (e.g.\ $\Gamma = C^0, C^{0,1}_{\loc}, C^\infty$), and assume that the Lorentzian metric $g$ is at least $\Gamma$-regular. In this paper, all manifolds themselves are assumed to be smooth. A \textbf{$\Gamma$-extension} of a Lorentzian manifold $(M,g)$ consists of an isometric embedding $\iota : M \hookrightarrow \tilde{M}$ of $(M,g)$ into a Lorentzian manifold $(\tilde{M}, \tilde{g})$ of the same dimension as $M$ where $\tilde{g}$ is $\Gamma$-regular and such that $\partial \iota(M) \subset \tilde{M}$ is non-empty. If a $\Gamma$-extension of $(M,g)$ exists, we say that $(M,g)$ is \textbf{$\Gamma$-extendible}; otherwise we say $(M,g)$ is \textbf{$\Gamma$-inextendible}. Returning to the particular case of the weak null singularity, a \textbf{$\Gamma$-extension} of $(\Mw, \gw)$ \textbf{across} $\{v = 0\}$ is a $\Gamma$-extension $\iota : \Mw \hookrightarrow \tilde{M}$ of $(\Mw, \gw)$ such that there exists a future-directed causal curve $\tau : [-1,0) \to \Mw$ with $\lim_{s \to 0} \tau_{v}(s) = 0$, $\lim_{s \to 0} \tau_u(s) <1$ and such that $\lim_{s \to 0}(\iota \circ \tau)(s) \in \rd \iota(\Mw) \subset \tilde{M}$ exists. 

We are interested in the local uniqueness of \emph{continuous} extensions $\iota : \Mw \hookrightarrow \tilde{M}$ of $(\Mw,\gw)$ across $\{v = 0\}$. We ask the following question: suppose $\lim_{s \to 0} \tau(s) = (u_*, 0, \omega_*) \in \overline{\Mw}$, does there exist $\epsilon>0$ and a small neighbourhood $B_{\Sp^2}$ of $\omega_*$ in $\Sp^2$ such that $\iota\vert_{(u_* - \epsilon, u_* + \epsilon) \times (-\epsilon,0) \times B_{\Sp^2}} : (u_* - \epsilon, u_* + \epsilon) \times (-\epsilon,0) \times B_{\Sp^2} \to \tilde{M}$ extends  to $$\overline{\iota} : (u_* - \epsilon, u_* + \epsilon) \times (-\epsilon,0] \times B_{\Sp^2} \to \tilde{M}$$
as a $C^1$-diffeomorphism (or at least as a homeomorphism) onto its image? If the answer is positive, then we say that $\iota : \Mw \hookrightarrow \tilde{M}$ is \textbf{locally $C^1$-($C^0$-)equivalent to $\overline{\Mw}$ at $(u_*, 0, \omega_*) \in \overline{\Mw}$} and it is easy to see that this notion is independent of the curve $\tau$ that approaches $(u_*,0, \omega_*)$. If the answer is negative, we say that $\iota : \Mw \hookrightarrow \tilde{M}$ is \textbf{locally $C^1$-($C^0$-)inequivalent to $\overline{\Mw}$ at $(u_*, 0, \omega_*) \in \overline{\Mw}$}.

Paraphrasing the above question, we have already found one continuous extension of $(\Mw,\gw)$ across $\{v=0\}$ by extending the topological and differentiable structure of $\Mw = (-1,1) \times (-1,0) \times \Sp^2\ni (u,v, \theta^A)$ to  $\overline{\Mw} = (-1,1) \times (-1,0] \times \Sp^2\ni (u,v, \theta^A)$. Can one extend the topological and/or differentiable structure of $\Mw$ differently and still obtain a continuous extension across $\{v=0\}$?

The relevance of this local uniqueness question is twofold. Firstly, it arises naturally in low-regularity inextendibility proofs. For example, in the proof of the $C^{0,1}_{\loc}$-inextendibility of spherically symmetric weak null singularities \cite{Sbie22a}, which is based on the blow-up of a local holonomy, it is crucial to control the velocity vectors of a sequence of curves in the local coordinates of the $C^{0,1}_{\loc}$-extension. That is, one needs a priori control on the $C^1$-differentiable structure of the $C^{0,1}_{\loc}$-extension. This would be provided by the statement above that $\iota\vert_{(u_* - \epsilon, u_* + \epsilon) \times (-\epsilon,0) \times B_{\Sp^2}}$ extends as a $C^1$-diffeomorphism to $\{v=0\}$.\footnote{In the paper \cite{Sbie22a} a slightly weaker statement is proved, which is, however, still sufficient for the application.} Furthermore, recent work by the second author \cite{Sbie24}, combined with \cite{sbierski2022uniqueness}, shows that for hypothetical $C^{0,1}_{\loc}$-extensions of weak null singularities \emph{without any symmetry assumptions}, $\iota\vert_{(u_* - \epsilon, u_* + \epsilon) \times (-\epsilon,0) \times B_{\Sp^2}}$ indeed extends as a $C^1$-diffeomorphism to $\{v=0\}$, thus exactly providing this a priori control of the $C^1$-differentiable structure. If curvature is assumed to blow-up in a particular frame -- note that a frame is an object at the level of the $C^1$-differentiable structure -- then one can show that this is in contradiction to the metric extending in $C^{0,1}_{\loc}$. For showing that the hypothetical $C^{0,1}_{\loc}$-extension is locally $C^1$-equivalent to $\overline{\Mw}$, one crucially uses, in addition to bounds on various connection coefficients, the local Lipschitz regularity of the extension.
Hence, the local uniqueness question posed in this paper asks whether this a priori control on the $C^1$-differentiable structure can still be obtained for a $C^0$-extension. Understanding this question informs and constrains approaches to proving inextendibility of weak null singularities in regularity classes below $C^{0,1}_{\loc}$ (and above $C^0$). Ideally, one would like to prove inextendibility with a continuous metric and locally square integrable Christoffel symbols.

The second motivation for studying the local uniqueness of continuous extensions to weak null singularities stems from the old question in general relativity which asks what kind of structure one can associate to a spacetime singularity. The problem with this of course is that the singularity itself is usually not part of the spacetime. This prompted the development of various ideal boundary constructions \cite{gboundary, cboundary, bboundary, aboundary} (or \cite{boundaryreview} for a review) which assign various structures (causality, topology, etc.) to the singularity. The case of weak singularities is, however, fundamentally different in the sense that the spacetime, that is, the manifold with continuous metric, can be defined \emph{at the singularity}. This naturally provides a causal, topological, and even differentiable structure. However, it has not yet been investigated whether these structures are uniquely determined by the property that the singularity admits one continuous extension, i.e., whether they are independent of the continuous extension chosen.

\subsection{Previous results}

The first systematic study of the uniqueness of at least $C^{1,1}_{\loc}$-extensions to $C^{2,1}_{\loc}$ boundaries was, to the best of the authors' knowledge, carried out in \cite{Chrus10}. It was shown that if two such extensions terminate the same null geodesic (i.e. if they are \textbf{anchored} by the same null geodesic), then locally they have the same $C^{2,1}_{\loc}$ differentiable structure. In \cite{sbierski2022uniqueness} it was shown that if two $C^{0,1}_{\loc}$-extensions $\iota_i : M \hookrightarrow \tilde{M}_i$, $i=1,2$, of a \emph{globally hyperbolic} Lorentzian manifold $(M,g)$ (without any further regularity assumptions on the  boundary) are anchored by the same causal curve $\tau :[-1,0) \to M$, then locally they have the same $C^{1,1}_{\loc}$-differentiable structure. This can be phrased mathematically as follows: if $\lim_{s \to 0} (\iota_i \circ \tau)(s) =:\tilde{p}_i \in \rd \iota_i(M) \subset \tilde{M}_i$, then there exist small neighbourhoods $\tilde{W}_i \subset \tilde{M}_i$ of $\tilde{p}_i$ which are divided into two parts, $(\tilde{W}_i)_<$ and $(\tilde{W}_i)_>$, by a Lipschitz graph $(\tilde{W}_i)_=$ which lies in $\rd \iota_i(M)$, contains $\tilde{p}_i$ and is such that $(\tilde{W}_i)_<$ lies in $\iota_i(M)$. The statement is then that the identification map $\id := \iota_2 \circ \iota_1^{-1}\vert_{\iota_1(M)} : \iota_1(M) \to \iota_2(M)$, restricted to $(\tilde{W}_1)_<$, maps $(\tilde{W}_1)_<$ diffeomorphically to $(\tilde{W}_2)_<$ and extends as a $C^{1,1}_{\loc}$-diffeomorphism to $(\tilde{W}_1)_\leq := (\tilde{W}_1)_< \cup (\tilde{W}_1)_=$. 

The local uniqueness question for continuous extensions of weak null singularities $(\Mw, \gw)$ in this paper is of exactly the same form: one may trivially extend the continuous boundary extension $(\overline{\Mw}, \gw)$ to a continuous extension $\iota_1 : \Mw \hookrightarrow (-1,1) \times (-1,1) \times \Sp^2 =: \tilde{M}_1$ in the sense defined above by extending the metric to be constant in the $(u, v, \theta^A)$-coordinates for $v \geq 0$. Let now $\iota_2 : \Mw \hookrightarrow \tilde{M}_2$ be another $C^0$-extension across $\{v=0\}$. Then, considering $\Mw$ as a subset of $\tilde{M}_1$, the identification map is $\id = \iota_2$ and we can choose $\tilde{W}_1$ such that $(\tilde{W}_1)_< = (u_* - \epsilon, u_* +  \epsilon) \times (-\epsilon, 0) \times B_{\Sp^2}$ and $(\tilde{W}_1)_\leq = (u_* - \epsilon, u_* +  \epsilon) \times (-\epsilon, 0] \times B_{\Sp^2}$. The difference between this paper and the previous work mentioned above is that we have fixed a particular reference $C^0$-extension $\overline{\Mw}$.\footnote{To be precise, $(\overline{\Mw}, \gw)$ is not a $C^0$-extension of $(\Mw, \gw)$ in the sense defined earlier, since $(\overline{\Mw}, \gw)$ is a manifold \emph{with} boundary, whereas an extension is defined to be an isometric embedding into a manifold \emph{without} boundary. Hence it is more accurate to call $(\overline{\Mw}, \gw)$ a \emph{continuous boundary extension of $(\Mw, \gw)$.} However, as noted above, a proper reference $C^0$-extension $\iota_1 : \Mw \hookrightarrow \tilde{M}_1$ may be constructed from this.}  

Furthermore, it is shown in \cite{sbierski2022uniqueness} that local uniqueness fails in general for anchored extensions which are only H\"older continuous. This is done by constructing a counterexample in which the identification map extends as a homeomorphism but not as a $C^1$-diffeomorphism.  
However, the non-uniqueness mechanism of the counterexample does not transfer to the case of weak null singularities. The counterexample in \cite{sbierski2022uniqueness} is 1+1-dimensional. In a double null gauge it is of the form $M= \{(u,v) \in \R^2 \; | \; u+ v > 0\}$, $g = - \Omega^2(u+v)(du \otimes dv + dv \otimes du)$ with $\Omega^2(u+v) \to 0$ for $u+v \to 0$, and uses the fact that the boundary, which is null (!) in the extension, can be tangentially approached by either the left or right going null geodesics in the original spacetime. This is qualitatively different from the weak null singularities we are interested in, where the boundary in the reference $C^0$-extension is the limit of the family of null hypersurfaces $\{v = \mathrm{const}\}$. Hence this example does not necessarily inform us about what  might be expected in the case of weak null singularities.

\subsection{A 1+1-dimensional toy model}\label{toymodel}

The local uniqueness problem for continuous extensions in 3+1-dimensions, as set out earlier, is a very intricate problem about which literally nothing is known to the best knowledge of the authors.
As a first step towards the resolution of this problem we begin in this paper by studying the following instructive 1+1-dimensional toy problem: we consider $M_t := (-1,1) \times (-1,0)$ with coordinates $(u,v)$ and metric $g_t = -\Omega^2 (du \otimes dv + dv \otimes du)$, where $\Omega$ is a smooth strictly positive function on $M_t$ which extends continuously as a strictly positive function to $\overline{M_t} := (-1,1) \times (-1,0] \supseteq M_t$. A time orientation is fixed by stipulating that $\rd_u$ is future-directed null. Now, given a continuous extension $\iota : M_t \hookrightarrow \tilde{M}$ of $(M_t,g_t)$ and a future-directed causal $C^1$-curve $\tau : [-1,0) \to M_t$ with $\lim_{s \to 0} \tau_u(s)=: u_* <1$, $\lim_{s \to 0} \tau_v(s) = 0$ and such that $\lim_{s \to 0}(\iota \circ \tau)(s) \in \rd \iota(M_t) \subset \tilde{M}$ exists\footnote{As before, a continuous extension $\iota : M_t \hookrightarrow \tilde{M}$ with these properties is called a \textbf{$C^0$-extension across $\{v=0\}$.}}, we ask whether there is an $\epsilon>0$ such that $\iota\vert_{(u_* - \epsilon, u_* + \epsilon) \times (-\epsilon, 0)} \to \tilde{M}$ extends to 
\begin{equation}\label{eqn:iotaextension}
\begin{split}
      \overline{\iota} : (u_* - \epsilon, u_* + \epsilon) \times (-\epsilon,0] \to \tilde{M}
\end{split}
\end{equation}
as a $C^1$-diffeomorphism/homeomorphism onto its image. Again, if the answer is positive, then we say that $\iota : M_t \hookrightarrow \tilde{M}$ is \textbf{locally $C^1$-($C^0$-)equivalent to $\overline{M_t}$ at $(u_*, 0) \in \overline{M_t}$} and it is easy to see that this notion is independent of the curve $\tau$ that approaches $(u_*,0)$. If the answer is negative, we say that $\iota : M_t \hookrightarrow \tilde{M}$ is \textbf{locally $C^1$-($C^0$-)inequivalent to $\overline{M_t}$ at $(u_*, 0) \in \overline{M_t}$}. This follows the definitions made in 3+1-dimensions. 

%\begin{example}[The reference extension of a conformally flat manifold]\label{example:referenceextensionconformal}
 %   Let $M$ be the 2-dimensional manifold with global coordinates $(u,v)\in(-1,1)\times(-1,0)$ and let $g=-2\Omega(u,v)(du\otimes dv+dv\otimes du)$ where $\Omega:(-1,1)\times (-1,0)\rightarrow\mathbbm{R}_{>0}$ is smooth and extends continuously to $v=0$ as a strictly positive function (which we also denote $\Omega$). 

  %  Define the continuous function
   % \begin{equation}
 %   \begin{split}
 %       \Omega'(u',v')&:(-1,1)\times (-1,1)\rightarrow\mathbbm{R}_{>0}\\
%        (u',v')&\mapsto  \begin{cases} 
 %     \Omega(u',v')  &-1 < v' < 0\\
%      \lim\limits_{v\rightarrow0}\Omega(u',v)  &0 \leq v' < 1 
 %  \end{cases}
  %  \end{split}
   % \end{equation}

   % Let $M'$ be the 2-dimensional manifold with global coordinates $(u',v')\in(-1,1)\times (-1,1)$ and let $\tilde{g}_{wns}=-2\Omega'(u',v')(du'\otimes dv'+dv'\otimes du')$. Then $(M,g)$ can be extended continuously across $v=0$ via the embedding
%\begin{equation}
%\begin{split}
%\iota_{ref}:(M,g)&\hookrightarrow(M',\tilde{g}_{wns})\\
%(u,v)&\mapsto(u',v').
%\end{split}
%\end{equation}
%where we have omitted the coordinate maps for brevity.
%\end{example} 

\subsection{Main results}\label{SecMainResults}

The main results of our investigation are contained in Section \ref{The $C^0$ Structure of the Extension in 1+1-Dimensions}, which discusses the $C^0$-structure of extensions, and Section \ref{SecC1Structure}, which looks at the $C^1$-structure. Recall that we consider the 1+1-dimensional model spacetime $M_t = (-1,1) \times (-1,0)$ with smooth metric $g_t = -\Omega^2(du \otimes dv + dv \otimes du)$ under the assumption that $\Omega$ extends as a continuous strictly positive function to $\overline{M_t} = (-1,1) \times (-1,0]$, which defines our reference extension. In  Example \ref{example:cornerextension}, under the additional assumption that $\Omega$ extends continuously as a strictly positive function to $[-1,1] \times (-1,0]$, we construct a whole family of $C^0$-extensions of $(M_t, g_t)$ across $\{v=0\}$ which terminate all future-directed null curves of constant $u$ at the same point! We refer to such extensions as \textbf{corner extensions}. They are clearly not locally $C^0$-equivalent to the reference extension $\overline{M_t}$. A natural question is now whether it is also possible for a $C^0$-extension of $(M_t, g_t)$ across $\{v=0\}$ to terminate only a subset of constant $u$ curves at the same point and to terminate the remaining constant $u$ curves at different boundary points. Proposition \ref{prop:nomixedcornerextension} shows that, under our assumptions on $\Omega$, this is not possible: if two nearby curves of constant $u$ terminate at the same point, then all other curves of constant $u$ must also terminate at this point.\footnote{See Proposition \ref{prop:nomixedcornerextension} for the precise statement. It is important to assume that the two curves of constant $u$ approach the same boundary point from the same null direction in the extension -- and not, say, one from the future and one from the past with respect to a small neighbourhood of the boundary point in question. Such extensions may be constructed, but they do not resemble a corner extension.} In particular, this \emph{local} assumption  on nearby curves of constant $u$ fixes the \emph{global} structure of the $C^0$-extension. We then conclude in Theorem \ref{thm:C0structure} that a $C^0$-extension of $(M_t,g_t)$ across $\{v=0\}$ is either locally $C^0$-equivalent to the reference extension $\overline{M_t}$ or has the global $C^0$-structure of a corner extension. This classifies the possible local $C^0$-structures of such extensions.

Next, we focus on the corner extensions and ask whether they are also relevant for $C^0$-extensions to weak null singularities in 3+1-dimensions. More specifically, we study a class of $C^0$-extensions across the Reissner-Nordstr\"om Cauchy horizon which preserve spherical symmetry. Considering the quotient spacetime induced by the spherical symmetry of subextremal Reissner-Nordstr\"om (or more generally by a spherically symmetric weak null singularity), we obtain a 1+1-dimensional Lorentzian manifold. By considering extensions of this manifold separately, and extending the $\Sp^2$ identically, one obtains a \textbf{strongly spherically symmetric extension} of the original 3+1-dimensional spacetime (see Section \ref{Implications for strongly spherically symmetric extensions across the Reissner-Nordstrom Cauchy horizon}). Working in Kruskal coordinates $(U,V)$, a finite $U$-chunk of this Lorentzian manifold is of the form $(M_t,g_t)$ and thus we may indeed construct a corner extension as in Example \ref{example:cornerextension} (see  Section \ref{cornerextensions}). Moreover, since the area radius $r$ in exact subextremal Reissner-Nordstr\"om is constant on the Cauchy horizon, this corner extension can even be constructed in 3+1-dimensions, see Example \ref{example:RNlocal}.\footnote{For generic spherically symmetric weak null singularities $r$ is strictly monotonically decreasing in $u$ and thus this local construction is not possible.} However, this construction is only possible for $C^0$-extensions of a \emph{subset of the Reissner-Nordstr\"om spacetime of finite $U$-coordinate}, see Corollary \ref{cor:eventhorizon}. If we consider the quotient spacetime globally, then it is shown in  Proposition \ref{prop:cornervolumecondition} and Corollary \ref{cor:eventhorizon} that, due to the infinite spacetime volume between the event horizon and Cauchy horizon, no such corner extension can exist. It thus follows that every strongly spherically symmetric $C^0$-extension across the Cauchy horizon is locally $C^0$-equivalent to the reference extension. To summarise, the possibility of the $C^0$-structure of a corner extension as in Theorem \ref{thm:C0structure} disappears for \emph{global} reasons if one considers strongly spherically symmetric $C^0$-extensions of the global spacetime -- thus leaving only the familiar local $C^0$-structure of the reference extension.

Having understood the rigidity of the local $C^0$-structure of $C^0$-extensions of $(M_t, g_t)$ across $\{v=0\}$, in Section \ref{SecC1Structure} we investigate  whether there is any rigidity for the $C^1$-structure. Example \ref{ExC1Cor} shows there are different corner extensions with non-equivalent $C^1$-structures. In Example \ref{example:C^1inequivalent2dimensions} we construct a $C^0$-extension which has the same $C^0$-structure as the reference extension $\overline{M_t}$ but inequivalent $C^1$-structure. This shows there is no rigidity at the level of the $C^1$-structure. Moreover, the construction from Example \ref{example:C^1inequivalent2dimensions} also generalises to weak null singularities in 3+1-dimensions without any symmetry assumptions. In particular, in Example \ref{ExWNS} we construct a $C^0$-extension of $(M_{wns},g_{wns})$ across $\{v=0\}$ which is $C^0$-equivalent to the reference extension $\overline{\Mw}$ constructed by Dafermos and Luk in \cite{DafLuk17} but $C^1$-inequivalent!

\subsection{Outline of paper}\label{Overview of results}

In Section \ref{SecDefBackground} we lay down some further definitions, recall some basic results on near Minkowski neighbourhoods and causal homotopies, and slightly strengthen them in 1+1-dimensions. Furthermore, as an auxiliary result needed for the next section, we prove in Proposition \ref{prop:C1nullcurve} that, in 1+1-dimensions, a locally Lipschitz causal curve which is achronal with respect to smooth timelike curves can be reparameterised as  a $C^1$ null curve. The main result of Section \ref{SecBdry} is Proposition \ref{prop:boundarystructure}.
We show that $\epsilon>0$ can be chosen so that $\overline{\iota}$ defined in \eqref{eqn:iotaextension} is continuous. This can be thought of as the statement that the reference extension has the ``richest'' possible $C^0$-structure in the sense that it adds the most boundary points. Another important property established in Proposition \ref{prop:boundarystructure} is that if the boundary corresponding to $\{v=0\}$ is not a single point (as in the case of a corner extension), then this part of the boundary is a $C^1$ null curve which is transverse to the curves of constant $u$. Starting from this a priori description of the boundary, we then investigate the local $C^0$- and $C^1$-structures of $C^0$-extensions in Sections \ref{The $C^0$ Structure of the Extension in 1+1-Dimensions} and \ref{SecC1Structure} as discussed above.

\subsection*{Acknowledgements}

Both authors were supported by the Royal Society University Research Fellowship URF\textbackslash R1\textbackslash 211216. The first author was also supported by the Additional Funding Programme for
Mathematical Sciences, delivered by EPSRC (EP/V521917/1) and the
Heilbronn Institute for Mathematical Research. Moreover, we are grateful to the two anonymous referees for their helpful suggestions.

\section{Definitions and Background Results}\label{SecDefBackground}

Throughout this paper we will rely on a number of definitions and technical results which we collect in this section. Following \cite{FutureNotAlwaysOpen} (see in particular Lemma 2.7), given a time-oriented Lorentzian manifold $(M,g)$, we define 
\begin{equation}
     \begin{split}
            I^\pm(p,M)&:=\{q\in M:\exists \text{ future/past-directed \textbf{piecewise smooth} timelike curve from $p$ to $q$}\}\\
            J^\pm(p,M)&:=\{q\in M:\exists \text{ future/past-directed \textbf{locally Lipschitz} causal curve from $p$ to $q$}\}.
        \end{split}
\end{equation}
The definition of a (future/past-directed) locally Lipschitz causal curve follows \cite[Definition 1.3]{ChruscielGrant}, see also \cite[Definition 2.3]{FutureNotAlwaysOpen}.  Unless otherwise stated, all causal curves are assumed to be locally Lipschitz and all timelike curves are assumed to be piecewise smooth. This definition ensures that, for $g\in C^0$, $I^\pm(p,M)$ is open \cite[Lemma 2.7]{FutureNotAlwaysOpen} and the accumulation curve result \cite[Theorem 1.6]{ChruscielGrant} holds for causal curves. 

It will often be convenient to reparameterise causal curves by different smooth time functions, in which case the following Lemma will be useful.
\begin{lemma}\label{lemma:reparamdifferentiability}
    Let $(M,g)$ be a $(d+1)$-dimensional time-oriented Lorentzian manifold with a continuous metric $g$, let $\gamma : [0,1] \to M$ be a future (or past) directed locally Lipschitz causal curve, and let $t$ be a smooth time function on $M$. Then $\gamma$ can be reparameterised by $t$ and we denote this reparameterisation by $\overset{(1)}{\gamma}$. If $\gamma$ is differentiable at $s_0 \in [0,1]$ with non-vanishing causal derivative at this point, then $\overset{(1)}{\gamma}$ is differentiable at $t\big(\gamma(s_0)\big)$. Furthermore, if $\gamma$ is $C^1$ near $s_0$ with non-vanishing causal derivative at this point then $\overset{(1)}{\gamma}$ is $C^1$ near $t\big(\gamma(s_0)\big)$.
\end{lemma}

\begin{proof}
    Since $\gamma$ is locally Lipschitz and $t$ is smooth, we have that $s \overset{\sigma}{\mapsto} t\big(\gamma(s)\big)$ is locally Lipschitz and thus for $0\leq s_1 < s_2 \leq 1$ we have $$\sigma(s_2) - \sigma(s_1) = \int_{s_1}^{s_2} \sigma'(s) \,ds = \int_{s_1}^{s_2} dt\big(\gamma'(s)\big) \, ds >0 \;,$$
    where in the last inequality we have assumed that $\gamma$ is future-directed -- if it were past-directed, we would get the reverse inequality. In either case it follows that $s \overset{\sigma}{\mapsto} t\big(\gamma(s)\big)$ is strictly monotonic, which implies that it can be inverted and thus $\gamma$ can be reparameterised by $t$. Moreover, since $t$ is a smooth time function and $\gamma$ is differentiable at $s_0$ with non-vanishing causal derivative at this point, it follows that $\sigma$ is differentiable at $s_0$ with non-vanishing derivative there. Hence the inverse function $\sigma^{-1}$ is differentiable at $t_0 := t\big(\gamma(s_0)\big)$. Thus, $\overset{(1)}{\gamma}(t) := \gamma\big(\sigma^{-1}(t)\big)$ is differentiable at $t_0$. {If $\gamma$ is $C^1$ near $s_0$ then the same argument shows that $\sigma$ is $C^1$ near $s_0$ and hence $\sigma^{-1}$ is $C^1$ near $t_0$. It follows that $\overset{(1)}{\gamma}$ is $C^1$ near $t_0$.}
\end{proof}

As in \cite{ChruscielGrant}, we will use certain smooth metrics to control the lightcones of $g$ (see for example Lemma \ref{lemma:nearmink} below). For $\zeta>-1$ define
\begin{equation}\label{modifiedminkowski}
\begin{split}
        \hat{g}_{\zeta}&=-(1+\zeta)^2dx_0^2+dx_1^2+\ldots + dx_d^2.
\end{split}
\end{equation}
We write $g'\prec g''$ for two metrics $g',g''$ if the following holds:
\begin{equation}
    g'(T,T)\leq0,T\neq0\implies g''(T,T)<0.
\end{equation}

\begin{lemma}\label{lemma:nearmink}\cite[Lemma 2.4]{SbierskiSchw}
 Let $(M,g)$ be a $(d+1)$-dimensional Lorentzian manifold with continuous metric $g$ and let $p\in M$. Then for every $\delta>0$ we can find $\epsilon_0,\epsilon_1>0$, an open neighbourhood $U$ of $p$, and a coordinate chart $\varphi:U\rightarrow(-\epsilon_0,\epsilon_0)\times(-\epsilon_1,\epsilon_1)^{d}$ such that
    \begin{enumerate}
        \item $\varphi(p)=(0,0,...,0)$
 %       \item $(\varphi\circ\gamma)(s)=(0,s)$
        \item $g_{\mu\nu}(p)=m_{\mu\nu}$
        \item $|g_{\mu\nu}(p')-m_{\mu\nu}|<\delta$\text{ for all } $p'\in U$
    \end{enumerate}
where $m_{\mu\nu}=\text{diag}(-1,1,\ldots,1)$ is the Minkowski metric on $\mathbbm{R}^{1,d}$. Moreover, for any $\zeta>0$, choosing $\delta >0$ sufficiently small ensures that\footnote{A straightforward calculation shows that $\delta=\frac{(1+\zeta)^2-1}{(1+(1+\zeta)^2)(d+1)}$ suffices, so in particular $\delta_0=\frac{3}{5(d+1)}$.\label{deltaalphafootnote}}
    \begin{equation}\label{eqn:lightconebound}
    \begin{split}
            \hat{g}_{-\frac{\zeta}{1+\zeta}}&\prec g\prec\hat{g}_\zeta\\
    i.e. \quad -\frac{1}{(1+\zeta)^2}dx_0^2+dx_1^2+\ldots + dx_d^2&\prec g\prec -(1+\zeta)^2dx_0^2+dx_1^2+\ldots + dx_d^2 \;.
    \end{split}
\end{equation}
It will be convenient to assume that $\delta<\delta_0$, where $\delta_0$ is chosen sufficiently small that \eqref{eqn:lightconebound} holds with $\zeta=1$. 
\end{lemma}
We refer to $(U,\varphi)$ as a \textbf{near Minkowski neighbourhood} centred at $p$ and to $\varphi$ as \textbf{near Minkowski coordinates}. 
  We will denote near Minkowski coordinates by $(x_0,\underline{x})$, where $x_0\in(-\epsilon_0,\epsilon_0)$ and $\underline{x}=(x_1,\ldots,x_d)\in(-\epsilon_1,\epsilon_1)^d$. Given a curve $\gamma(s)$, we will denote its components in these coordinates by $(\gamma_0(s),\underline{\gamma}(s))=(\gamma_0(s),\gamma_1(s),\ldots,\gamma_d(s)):=\big(x_0(\gamma(s)),\underline{x}(\gamma(s))\big)$. 
  If $(M,g)$ in Lemma \ref{lemma:nearmink} is time-oriented, then without loss of generality we can assume that $\frac{\partial}{\partial x_0}$ is future-directed. If $(M,g)$ is not time-oriented (or not time-orientable), we define a time orientation on $(U,g)$ by stipulating that $\frac{\partial}{\partial x_0}$ is future-directed.
We have added the condition that $\delta<\delta_0$ since this allows us to prove the following corollary, which gives control over causal curves in $\tilde{U}$ and will be used frequently.

%Since we will be interested in extending Lorentzian manifolds across null boundaries, it will often be more convenient to work in coordinates adapted to this. By defining $u:=x_0-x_1$ and $v:=x_0+x_1$, one obtains a coordinate chart satisfying the above conditions with $m_{\mu\nu}$ replaced by 
%    \begin{equation}
%       m'_{\mu\nu}=
%      \begin{pmatrix}
%    0     &  -\frac{1}{2}  \\
%    -\frac{1}{2} & 0  
%\end{pmatrix}
%\end{equation}
%We refer to such coordinates as \textit{near Minkowski double null coordinates}.

\begin{cor}\label{cor:reparam}
        Let $(U,\varphi=(x_0,\underline{x}))$ be a near Minkowski neighbourhood in $(M,g)$ as in Lemma \ref{lemma:nearmink}
        %with $\delta\in(0,\delta_0)$ 
        and let $\gamma:[a,b]\rightarrow U$ be a locally Lipschitz future-directed causal curve. Then
        
      (a) $x_0$ is a smooth time function on $U$. %(after possibly replacing $x_0$ by $-x_0$). 
      
            (b) $\gamma$ can be parameterised by $x_0$, which is strictly monotonic along this curve. Denoting this reparameterisation by $\overset{(1)}{\gamma}$, we have 
        \begin{equation}\label{eqn:gradientboundcausal}
            \left\lVert\underline{\overset{(1)}{\gamma}}(x'_0)-\underline{\overset{(1)}{\gamma}}(x_0)\right\rVert_{\mathbbm{R}^d}< (1+\zeta)|x'_0-x_0|.
        \end{equation}
        
          %  (c) Suppose $\check{x}_0$ is another smooth time function on $\tilde{U}$ and denote the corresponding parameterisation of $\gamma$ by  $\overset{(2)}{\gamma}$. Suppose $\overset{(1)}{\gamma}(X_0)=\overset{(2)}{\gamma}(\check{X}_0)$. Then $\overset{(1)}{\gamma}$ is differentiable at $x_0=X_0$ if and only if $\overset{(2)}{\gamma}$ is differentiable at $\check{x}_0=\check{X}_0$.
            
            (c) In 1+1-dimensions if $\gamma$ is $C^1$ and null then
            
            (i) $\gamma$ is achronal with respect to smooth timelike curves in $U$, and
            
            (ii) for any $\zeta\in(0,1)$ we can choose $U$ sufficiently small so that, for $x_0\neq x'_0$, we have
    \begin{equation}\label{eqn:gradientbound1+1}
        \frac{1}{1+\zeta}<\frac{|x'_0-x_0|}{\left|\overset{(1)}{\gamma}_1(x'_0)-\overset{(1)}{\gamma}_1(x_0)\right|}<1+\zeta \;
    \end{equation}
   in $U$ and 
   \begin{equation}\label{eqn:2dnullderivativeformula}
       \frac{d\overset{(1)}{\gamma}_1}{dx_0}=\frac{1}{g_{11}(x_0)}\left(-g_{01}(x_0)\pm\sqrt{g_{01}(x_0)^2-g_{00}(x_0)g_{11}(x_0)}\right)
   \end{equation}
  with either the positive or negative root everywhere in $U$, where we have written $g_{\mu\nu}(x_0)$ to denote $g\left(\overset{(1)}{\gamma}(x_0)\right)_{\mu\nu}$ (recall from Lemma \ref{lemma:reparamdifferentiability} that $\overset{(1)}{\gamma}(x_0)$ is also $C^1$).
\end{cor}

\begin{proof} (a)  It follows from \eqref{eqn:lightconebound} that  the level sets of $x_0$ are spacelike, thus $g^{-1}(dx_0, dx_0) <0$. Recall that $\partial_0$ is future-directed timelike in $U$. Since $1=dx_0(\partial_0)=g(\nabla x_0,\partial_0)$ it follows that $\nabla x_0$ is past-directed timelike, so $x_0$ is a time function on $U$.

(b)  Since $\gamma$ is locally Lipschitz, it follows that it is differentiable almost everywhere, and from \eqref{eqn:lightconebound} we have
\begin{align}
     -(1+\zeta)^2\left(\frac{d\gamma_0}{ds}\right)^2+\sum_{j=1}^d\left(\frac{d\gamma_j}{ds}\right)^2&<0\\
    \implies \Big|\Big|\frac{d\underline{\gamma}}{ds}\Big|\Big|_{\R^d}&<(1+\zeta)\frac{d\gamma_0}{ds}\label{eqn:derivbound2}\\
     \implies ||\underline{\gamma}(s')-\underline{\gamma}(s)||_{\R^d}&<(1+\zeta) |\gamma_0(s')-\gamma_0(s)|\label{eqn:bound}
\end{align}
where the final line follows from the fundamental theorem of calculus and holds for any $s,s'\in[a,b]$. 

By (a) and Lemma \ref{lemma:reparamdifferentiability}, we can parameterise $\gamma$ by $x_0$, which is strictly increasing along this curve. This reparameterisation is given by $\overset{(1)}{\gamma}(x_0)=\gamma(s^{\overset{(1)}{\gamma}}(x_0))$, where $s^{\overset{(1)}{\gamma}}(x_0)$ is the inverse of $x_0(\gamma(s))$. Writing $s=s^{\overset{(1)}{\gamma}}(x_0)$ in \eqref{eqn:bound} we obtain the required result.

%(c) Suppose $\overset{(1)}{\gamma}(x_0)$ is differentiable at $x_0=X_0$. As in (b), we deduce that $\check{x}_0(\overset{(1)}{\gamma}(x_0))$ is strictly monotonic and differentiable at $x_0=X_0$. Hence by we can define its inverse $x_0^{\overset{(1)}{\gamma}}(\check{x}_0)$ which is differentiable at $\check{x}_0=\check{X}_0$.
%, the function $\check{x}_0(\overset{(1)}{\gamma}(x_0))$ is strictly increasing and hence we can define its inverse $x_0^{\overset{(2)}{\gamma}}(\check{x}_0)$. Moreover,$\check{x}_0(\overset{(1)}{\gamma}(x_0))$ is differentiable at $x_0=X_0$ and hence, by , $s^{\overset{(2)}{\gamma}}$ is differentiable at $\check{x}_0= =\check{X}_0$. 
%It follows that $\overset{(2)}{\gamma}(\check{x}_0):=\overset{(1)}{\gamma}(x_0^{\overset{(1)}{\gamma}}(\check{x}_0))$ is differentiable at $\check{x}_0=\check{X}_0$. other direction?

(c)(i) In 1+1-dimensions, if $\gamma$ is $C^1$ then Lemma \ref{lemma:reparamdifferentiability} implies that $\overset{(1)}{\gamma}_1(x_0)$ is also $C^1$. From \eqref{eqn:lightconebound} we have 
%$\frac{d\overset{(1)}{\gamma}_1}{dx_0}\neq0$ for all $s\in[a,b]$ and hence $\overset{(1)}{\gamma}_1(x_0)$ is either strictly increasing or strictly decreasing. Again by \eqref{eqn:lightconebound} we have 
\begin{equation}\label{eqn:derivativebound1+1}
    \frac{1}{1+\zeta}<\left|\frac{d\overset{(1)}{\gamma}_1}{dx_0}\right|<1+\zeta
\end{equation} so in particular $\overset{(1)}{\gamma}_1(x_0)$ is either strictly increasing or strictly decreasing. It follows that we can invert this function and parameterise $\gamma$ instead by $x_1$. We denote this reparameterised curve by $\overset{(2)}{\gamma} : [ \alpha, \beta] \to U$ where, in the case of $\overset{(1)}{\gamma}_1$ being strictly increasing, we have $\alpha = \gamma_1(a)$ and $\beta = \gamma_1(b)$, while if $\overset{(1)}{\gamma}_1$ is strictly decreasing then the values of $\alpha$ and $\beta$ are swapped. We extend $\overset{(2)}{\gamma}$ to a curve $\mu:(-\epsilon_1,\epsilon_1)\rightarrow U$, also parameterised by $x_1$, by (see Figure \ref{fig:achronality})
\begin{equation}
   \mu_0(x_1)=  \begin{cases}
       \overset{(2)}{\gamma}_0(\alpha)&\text{ for }x_1<\alpha\\
       \overset{(2)}{\gamma}_0(x_1)&\text{ for }\alpha\leq x_1\leq \beta\\
       \overset{(2)}{\gamma}_0(\beta)&\text{ for }x_1>\beta.
    \end{cases}
\end{equation}
We say a point $(x_0,x_1)\in U$ lies above or below $\mu$ if $x_0>\mu_0(x_1)$ or $x_0<\mu_0(x_1)$ respectively. 
%We say a curve in $U$ lies above or below $\mu$ if all points in its image lie above or below $\mu$ respectively.
Throughout the rest of the proof, $\sigma:[0,1]\rightarrow U$ will refer to a future-directed smooth timelike curve.

\textbf{Claim}: Suppose $\sigma(0)\in \text{Im}(\mu)$. Then $\sigma(s)$ lies above $\mu$ for $s>0$ sufficiently small. 

\begin{figure}[h]
    \centering
    \includegraphics[scale=0.25]{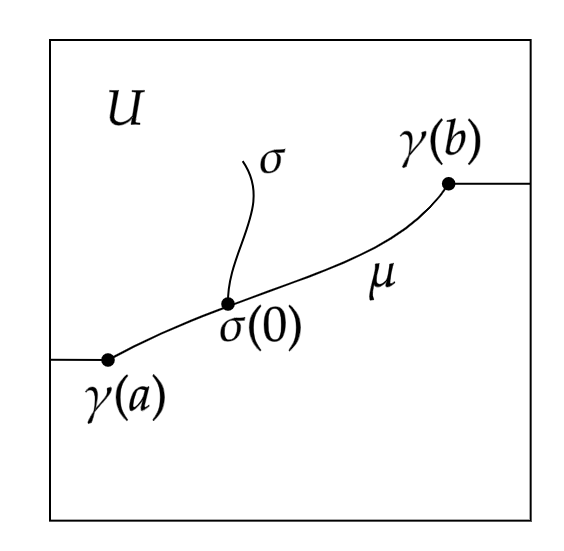}
    \caption{We extend $\gamma$ as a curve of constant $x_0$ so that the range of $x_1$ is now $(-\epsilon_1,\epsilon_1)$. This figure illustrates the claim in the case where $\frac{d\overset{(1)}{\gamma}_1}{dx_0}>0$ and $\sigma(0)\in \text{Im}(\gamma)$.}
    \label{fig:achronality}
\end{figure}

\textbf{Proof of claim}: Suppose $\frac{d\overset{(1)}{\gamma}_1}{dx_0}>0$ (the case where $\overset{(1)}{\gamma}_1(x_0)$ is strictly decreasing follows similarly) and recall that $\sigma_0(s)$ is strictly increasing. Since $x_0$ is constant along $\mu$ for $x_1>\beta$ or $x_1<\alpha$, the claim follows immediately in the cases where $\sigma_1(0)\geq\beta$ or $\sigma_1(0)<\alpha$. 

Now suppose $\sigma(0)\in\text{Im}\overset{(2)}{\gamma}\vert_{[\alpha,\beta)}$. It follows from part (b) that we can parameterise $\sigma$ by $x_0$. We denote this reparameterisation by $\overset{(1)}{\sigma}$. If $\frac{d\overset{(1)}{\sigma}_1}{dx_0}<0$ at $x_0=\sigma_0(0)$ then the claim follows immediately. Otherwise, since $\sigma$ is timelike we have $\frac{d\overset{(1)}{\gamma}_1}{dx_0}>\frac{d\overset{(1)}{\sigma}_1}{dx_0}\geq0$ at $x_0=\sigma_0(0)$, from which the claim also follows. 
%Finally, suppose $\sigma_0(0)=\gamma_0(a)$. If $\frac{d\overset{(1)}{\sigma}_1}{dx_0}\vert_{x_0=\gamma_0(a)}\leq0$ then the claim follows from the fact that $x_0$ is increasing along $\sigma$. If $\frac{d\overset{(1)}{\sigma}_1}{dx_0}\vert_{x_0=\gamma_0(a)}>0$ then the claim follows from the fact that $\frac{d\overset{(1)}{\sigma}_1}{dx_0}>\frac{d\overset{(1)}{\gamma_1}}{dx_0}$ at $x_0=\gamma_0(a)$. 

By time reversal, one can also show that if instead $\sigma(1)\in\text{Im}(\mu)$ then $\sigma(s)$ lies below $\mu$ for $s<1$ sufficiently close to 1. 

We now prove part (c)(i) of Corollary \ref{cor:reparam}. Suppose for a contradiction that $\sigma(0),\sigma(1)\in\text{Im}(\gamma)$. Then by the claim above (and its time reversed analogue), we must have $\sigma(s)$ lying above $\mu$ for $s>0$ sufficiently small and below $\mu$ for $s<1$ sufficiently close to 1. Then by the intermediate value theorem, there exists $s_1\in(0,1)$ such that $\sigma(s_1)\in\text{Im}(\mu)$. Since $\sigma$ and $\mu$ are continuous, we may choose the maximal such $s_1$, so $\sigma(s)$ lies below $\mu$ for all $s\in(s_1,1)$, which contradicts the claim above.

(ii) By considering separately the cases where $\overset{(1)}{\gamma}_1(x_0)$ is strictly increasing and strictly decreasing, it follows from \eqref{eqn:derivativebound1+1} that for $x_0\neq x_0'$ we have
\begin{equation}
\begin{split}
        \overset{(1)}{\gamma}_1(x_0')-\overset{(1)}{\gamma}_1(x_0)&=\int_{x_0}^{x'_0}\frac{d\overset{(1)}{\gamma}_1(\rho)}{d\rho}d\rho\\
       \implies \frac{1}{1+\zeta}\left|x_0'-x_0\right|&<\left|\overset{(1)}{\gamma}_1(x_0')-\overset{(1)}{\gamma}_1(x_0)\right|<(1+\zeta)\left|x_0'-x_0\right|\\
       \implies\frac{1}{1+\zeta}&<\frac{\left|x_0'-x_0\right|}{\left|\overset{(1)}{\gamma}_1(x_0')-\overset{(1)}{\gamma}_1(x_0)\right|}<1+\zeta.
\end{split}
\end{equation}
Since $\gamma$ is null, $\overset{(1)}{\gamma}_1(x_0)$ must satisfy \eqref{eqn:2dnullderivativeformula}. Since $\delta<\delta_0<1$, taking the positive root in \eqref{eqn:2dnullderivativeformula} gives $\frac{d\overset{(1)}{\gamma}_1}{dx_0}>0$ while taking the negative root gives $\frac{d\overset{(1)}{\gamma}_1}{dx_0}<0$. Since we have $\frac{1}{1+\zeta}<\left|\frac{d\overset{(1)}{\gamma}_1}{dx_0}\right|$, it follows that $\frac{d\overset{(1)}{\gamma}_1}{dx_0}\neq0$ and hence $\overset{(1)}{\gamma}_1(x_0)$ satisfies \eqref{eqn:2dnullderivativeformula} with either the positive or negative root everywhere in $U$.
%A similar result holds in the case $x_1(x_0)$ is strictly decreasing. Combining these we obtain the final result \eqref{eqn:gradientbound}.
\end{proof}

%\begin{lemma}
 %   Let $t_1$, $t_2$ be smooth time functions and $\gamma(s)$ be a locally Lipschitz - do you want to say somewhere that you always mean this when you say causal? causal curve. Recall from  that $t_i(\gamma(s))$ is invertible, so we can define $s^\gamma_i(t_i)$ such that $s^\gamma_i(t_i(\gamma(s)))s$. Let $\gamma^{(i)}(t_i):=\gamma(s^\gamma_i(t_i))$ for $i=1,2$. If $\gamma(s_1^\gamma(T_1))=\gamma(s_2^\gamma(T_2))$ for some $T_1,T_2$ then $\gamma^{(1)}(t_1)$ is differentiable at $t_1=T_1$ if and only if $\gamma^{(2)}$ is differentiable at $t_2=T_2$.  
%\end{lemma}
%}

A useful notion which is used extensively in \cite{GalLin16}, \cite{SbierskiSchwarzschild2}, and \cite{sbierski2022uniqueness} is that of the future boundary. 

\begin{definition}\cite[Definition 2.1]{GalLin16}\label{defn:futureboundary}
    Let $(M,g)$ be a time-oriented $(d+1)$-dimensional Lorentzian manifold with $g\in C^0$ and let $\iota:M\hookrightarrow\tilde{M}$ be a $C^0$-extension of $M$. The future boundary of $M$ is the set $\partial^+\iota(M)$ consisting of all points $\tilde{p}\in\tilde{M}$ such that there exists a smooth timelike curve $\tilde{\gamma}:[-1,0]\rightarrow\tilde{M}$ with $Im\left(\tilde{\gamma}\vert_{[-1,0)}\right)\subset\iota(M)$, $\tilde{\gamma}(0)=\tilde{p}\in\partial\iota(M)$, and $\iota^{-1}\circ\tilde{\gamma}\vert_{[-1,0)}$ is future-directed in $M$.
\end{definition}
It is clear that $\partial^+\iota(M)\subset\partial\iota(M)$. Throughout this paper we will abuse notation and write e.g. $\tilde{\gamma}\vert_{[-1,0)}$ for $Im(\tilde{\gamma}\vert_{[-1,0)})$ when referring to images of curves. Under the additional assumption that $(M,g)$ is globally hyperbolic, one can obtain the following result for points in the future boundary $\partial^+\iota(M)$.

\begin{prop}\cite[Proposition 2.3]{sbierski2022uniqueness}\label{prop:futureboundarychart}
    Let $\iota:M\hookrightarrow\Tilde{M}$ be a $C^0$-extension of a time-oriented globally hyperbolic Lorentzian manifold $(M,g)$ with $g\in C^0$ and with Cauchy hypersurface $\Sigma$. Let $\Tilde{p}\in\partial^+\iota(M)$. Then, for every $\delta>0$, in addition to choosing a corresponding near Minkowski neighbourhood $\tilde{\varphi}:\tilde{U}\rightarrow(-\epsilon_0,\epsilon_0)\times(-\epsilon_1,\epsilon_1)^{d}$, $\epsilon_0,\epsilon_1>0$, centred at $\tilde{p}$, there exists a Lipschitz continuous function $f:(-\epsilon_1,\epsilon_1)^d\rightarrow(-\epsilon_0,\epsilon_0)$ with the following property: 
    \begin{equation}\label{eqn:belowgraphoff}
      \tilde{U}_<:= \{(x_0,\underline{x})\in(-\epsilon_0,\epsilon_0)\times(-\epsilon_1,\epsilon_1)^{d}:x_0<f(\underline{x})\}\subset \tilde{\varphi}(\iota(I^+(\Sigma,M))\cap\tilde{U})
    \end{equation}
    and
     \begin{equation}\label{eqn:graph}
        \{(x_0,\underline{x})\in(-\epsilon_0,\epsilon_0)\times(-\epsilon_1,\epsilon_1)^{d}:x_0=f(\underline{x})\}\subset \tilde{\varphi}(\partial^+\iota(M)\cap\tilde{U}).
    \end{equation}
Moreover, the set on the left hand side of \eqref{eqn:graph}, i.e. the graph of $f$, is achronal (with respect to smooth timelike curves) in $\tilde{U}$.
\end{prop}
We refer to $(\tilde{U},\tilde{\varphi})$ as a \textbf{future boundary chart} centred at $\Tilde{p}$. It will also be convenient to define 
\begin{equation}
    \tilde{U}_{\leq}:= \{\tilde{\varphi}^{-1}(x_0,\underline{x})\in\tilde{U}:x_0\leq f(\underline{x})\}.
\end{equation}

Restricting now to 1+1-dimensions, we prove two results for Lorentzian manifolds which will be used in Proposition \ref{prop:boundarystructure} to understand the structure of the boundary $\partial\iota(M_t)$, where $M_t$ is as in Section \ref{toymodel}.

\begin{prop}\label{prop:C1nullcurve}
    Let $(M,g)$ be a 1+1-dimensional time-oriented Lorentzian manifold with $g\in C^0$. 
   Let $\gamma:[0,1]\rightarrow M$ be a future-directed locally Lipschitz causal curve such that no two points on $\gamma$ can be connected by a smooth timelike curve in $M$. Then $\gamma$ can be reparameterised as a $C^1$ null curve.
\end{prop}

\begin{proof}
Let $s_0\in[0,1]$ and $(U,\varphi)$ be a near Minkowski neighbourhood centred on $p:=\gamma(s_0)$ such that $\hat{g}_{-\frac{1}{4}}\prec g\prec\hat{g}_{\frac{1}{4}}$ on $U$. 

For $s\in[0,\epsilon')$, where $\epsilon':=\min\{\epsilon_0,\epsilon_1\}$, define the ball
\begin{equation}
\begin{split}
    B(s)&:=\left\{q\in U:\sqrt{(x_0(q))^2+(x_1(q))^2}< s\right\}
\end{split}
\end{equation}
and
\begin{equation}\label{eqn:omega}
\begin{split}
     w:[0,\epsilon')\rightarrow&\left[0,\frac{1}{4}\right]\\
   s\mapsto
      &\inf\left\{\zeta\in\left[0,\frac{1}{4}\right]:\hat{g}_{- \zeta} \prec g \prec \hat{g}_{\zeta}\text{ on }B(s)\right\}.
\end{split}
\end{equation}
Note that we have bounded the deviation of the lightcones of $g$ from the Minkowski lightcones slightly differently than in \eqref{eqn:lightconebound}.  The function $ w(s)$ measures the maximal deviation in $B(s)$. Since $g$ is continuous and $(x_0,x_1)$ are near Minkowski coordinates, it follows that $ w$ is continuous, non-decreasing and $ w(0)=0$. By reducing $\epsilon_0$ if necessary, it follows from Corollary \ref{cor:reparam} (b) that $\gamma$ can be parameterised by $x_0\in(-\epsilon_0,\epsilon_0)$ such that this reparameterised curve, denoted $\overset{(1)}{\gamma},$ is  Lipschitz in $U$. It follows that  $\overset{(1)}{\gamma}_1(x_0)$ is differentiable almost everywhere on $(-\epsilon_0,\epsilon_0)$. In particular, since $\gamma$ is causal, at points in $B(\epsilon')$ where $\overset{(1)}{\gamma}_1(x_0)$ is differentiable, we have
\begin{equation}\label{eqn:causalgradientbounds}
\begin{split}
    \left|\frac{d\overset{(1)}{\gamma}_1}{dx_0}(x_0)\right|\leq1+  w\left(\sqrt{x_0^2+\overset{(1)}{\gamma}_1(x_0)^2}\right)
    \end{split}
\end{equation}
For $|s|<\frac{\epsilon'}{\sqrt{5}}$, define $\mathcal{W}(s):=\int^{|s|}_0 w(\sqrt{5}s')ds'$ and we now restrict our attention to $B(\frac{\epsilon'}{\sqrt{5}}) \subset U$ such that for $q = (x_0,x_1) \in B(\frac{\epsilon'}{\sqrt{5}}) $ we have that $\mathcal{W}(x_0)$ is defined. We separate $B(\frac{\epsilon'}{\sqrt{5}})$ into four disjoint regions as follows (see Figure \ref{fig:omegaseparation}). 
\begin{itemize}
    \item $A:=\{q\in B(\frac{\epsilon'}{\sqrt{5}}):|x_1|<|x_0|-\mathcal{W}(x_0)\}$
    \item $B:=\{q\in B(\frac{\epsilon'}{\sqrt{5}}):|x_1|>|x_0|+\mathcal{W}(x_0)\}$
    \item $C_-:=\{q\in B(\frac{\epsilon'}{\sqrt{5}}):|x_0+x_1|\leq \mathcal{W}(x_0)\}\setminus\{(0,0)\}$
    \item $C_+:=\{q\in B(\frac{\epsilon'}{\sqrt{5}}):|x_0-x_1|\leq \mathcal{W}(x_0)\}$
\end{itemize}
\begin{figure}[h]
        \centering
        \includegraphics[scale=0.3]{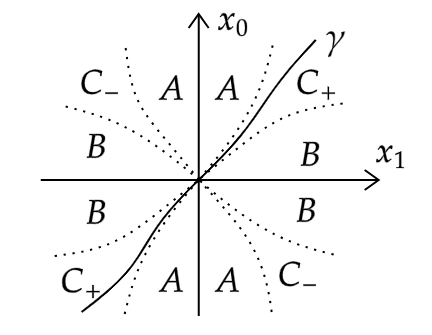}
        \caption{Figure showing the split of $B(\frac{\epsilon'}{\sqrt{5}})$ into the regions $A$, $B$, $C_-$ and $C_+$ as described in the proof of Proposition \ref{prop:C1nullcurve}, where we show that $\gamma\subset C_+$.}
        \label{fig:omegaseparation}
    \end{figure}

\textbf{Claim}: Replacing $x_1$ with $-x_1$ if necessary, we have $\gamma\subset C_+$.

\textbf{Proof of claim}: We first show that $\gamma\cap B=\emptyset$. 
%By \eqref{eqn:causalgradientbounds}, we have $\left|\frac{dx_1}{dx_0}\right|\leq1+ w(s)$ almost everywhere on $\gamma$, where $s^2=(x_0)^2+(x_1)^2$. 
By Corollary \ref{cor:reparam} (b), we have $|\overset{(1)}{\gamma}_1(x_0)|< 2|x_0|$. Combining this with \eqref{eqn:causalgradientbounds} and using the fundamental theorem of calculus, we have
\begin{equation}
    \begin{split}
    \overset{(1)}{\gamma}_1(x_0)&=\int_0^{x_0}\frac{d\overset{(1)}{\gamma}_1(x_0')}{dx_0'}dx_0'\\
       \implies |\overset{(1)}{\gamma}_1(x_0)|&\leq\int_0^{|x_0|}1+ w\left(\sqrt{x_0'^2+\overset{(1)}{\gamma}_1(x_0')^2}\right)dx_0'\\
       &\leq \int_0^{|x_0|}1+ w\left(\sqrt{5}x_0'\right)dx_0'\\
        &= |x_0|+ \mathcal{W}(x_0)
    \end{split}
\end{equation}
and hence $\gamma\cap B=\emptyset$. 

Now suppose $\overset{(1)}{\gamma}(x_0^*)\in A$ for some $x_0^* \in (-\epsilon_0,\epsilon_0)$ (note that this requires $x^*_0\neq0$) and define $x_1^*:=\overset{(1)}{\gamma}_1(x_0^*)$. We assume that $x_0^*>0$ and $x_1^*\geq 0$, so 
\begin{equation}\label{EqArrow}
x_1^*<x_0^*- \mathcal{W}(x_0^*)
\end{equation}
(the other cases follow similarly).
We will show that there exists a $C^1$ timelike curve $\sigma$ from $\overset{(1)}{\gamma}(0)$ to $\overset{(1)}{\gamma}(x_0^*)$. To do this, we seek a $C^1$ function $\sigma_1:[0,x_0^*]\rightarrow[0,x_1^*]$ such that $\sigma_1(0)=0$, $\sigma_1(x_0^*)=x_1^*$ and $|\sigma_1'(x_0)|<1- w(\sqrt{5}x_0)$. Such a function would satisfy $|\sigma_1(x_0)|\leq x_0$ and hence
\begin{equation}
    \begin{split}
        \left|\sigma_1'(x_0)\right|&< 1- w(\sqrt{5}x_0) \leq   1- w\left(\sqrt{(x_0)^2+(\sigma_1(x_0))^2}\right),
    \end{split}
\end{equation}
where we have used the monotonicity of $ w$, and from which it then follows that defining $\sigma(x_0):=(x_0,\sigma_1(x_0))$ would give the required curve. We claim that $\sigma_1(x_0):=- \mathcal{W}(x_0)+\frac{x_0}{x_0^*}(x_1^*+ \mathcal{W}(x_0^*))$ satisfies the conditions. We have $\sigma_1(0)=0$, $\sigma(x_0^*)=x_1^*$ and 
\begin{equation}
     \begin{split}
        \sigma_1'(x_0)&=- w(\sqrt{5}x_0)+\frac{1}{x_0^*}(x_1^*+ \mathcal{W}(x_0^*))\\
        \overset{\eqref{EqArrow}}{\implies} 1- w(\sqrt{5}x_0)>\sigma_1'(x_0)&\geq - w(\sqrt{5}x_0)\\
        &> -1+ w(\sqrt{5}x_0)\;,
          \end{split}
\end{equation}
where we have also used $0 \leq  w \leq \frac{1}{4}$.
Moreover, $\sigma_1$ is $C^1$ since $ w$ is continuous.

Then, since the timelike future of a point is the same whether defined with respect to $C^1$ or smooth curves \cite[Lemma 2.7]{ChruscielGrant}, it follows that there exists a future-directed smooth timelike curve from $\overset{(1)}{\gamma}(0)$ to $\overset{(1)}{\gamma}(x_0^*)$, which contradicts the assumption of the proposition. We conclude that $\gamma\cap A=\emptyset$ and hence $\gamma\subset C_-\cup C_+$.
%Since $\gamma(x_0^*)\in A$ we have 
%\begin{equation}
%    \begin{split}
%        x_0^*-x_0(q)&> W(x_0^*)- W(x_0(q))\\
%        &=\int^{x_0^*}_{x_0(q)} w(s)ds\\
%       &\geq(x_0^*-x_0(q))\inf_{s\in S} w(s)\\
%       \implies (x_0^*-x_0(q))(1-\inf_{s\in S} w(s)&>0.
%    \end{split}
%\end{equation} 
%where $S:=[\min\{x_0(q),x_0^*\},\max\{x_0(q),x_0^*\}]$.

%Since $ w$ is continuous and takes values in $[0,1)$, it follows from the compactness of $S$ that $1-\inf_{s\in S} w(s)>0$ and hence $x_0^*>x_0(q)$.  
In order to show that $\gamma\cap C_-=\emptyset$ (after possibly replacing $x_1$ with $-x_1$), we first show that $\overset{(1)}{\gamma}_1(x_0)$ is a strictly monotonic function on $(-\epsilon_0,\epsilon_0)$. Suppose this were not the case, so for example there exist $a<b<c$ such that $\overset{(1)}{\gamma}(a),\overset{(1)}{\gamma}_1(c)\leq \overset{(1)}{\gamma}_1(b)$ (the case where $\overset{(1)}{\gamma}_1(b)<\overset{(1)}{\gamma}_1(a),\overset{(1)}{\gamma}_1(c)$ follows similarly). If $\overset{(1)}{\gamma}_1(a)=\overset{(1)}{\gamma}_1(b)$ then the straight line from $(a,\overset{(1)}{\gamma}_1(a))$ to $(b,\overset{(1)}{\gamma}_1(a))$ is a smooth, future-directed timelike curve connecting distinct points on $\overset{(1)}{\gamma}$. This is again a contradiction. We get a similar contradiction if $\overset{(1)}{\gamma}_1(c)=\overset{(1)}{\gamma}_1(b)$. We may therefore assume that $\overset{(1)}{\gamma}_1(a),\overset{(1)}{\gamma}_1(c)< \overset{(1)}{\gamma}_1(b)$. Let $\alpha=\frac{1}{2}(\overset{(1)}{\gamma}_1(b)+\max\{\overset{(1)}{\gamma}_1(a),\overset{(1)}{\gamma}_1(c)\})\in(\overset{(1)}{\gamma}_1(a),\overset{(1)}{\gamma}_1(b))\cap(\overset{(1)}{\gamma}_1(c),\overset{(1)}{\gamma}_1(b))$. By the intermediate value theorem, there exist $a'\in(a,b)$ and $b'\in(b,c)$ such that $\overset{(1)}{\gamma}_1(a')=\overset{(1)}{\gamma}_1(c')=\alpha$. Once again we obtain a contradiction by considering the straight line from $(a',\overset{(1)}{\gamma}_1(a'))$ to $(c',\overset{(1)}{\gamma}_1(a'))$.  Without loss of generality (replacing $x_1$ with $-x_1$ if necessary) we can therefore assume that $\overset{(1)}{\gamma}_1(x_0)$ is strictly increasing. 

Recalling $\overset{(1)}{\gamma}_1(0) = 0$, it  now follows that, along $\overset{(1)}{\gamma}$ in $B(\frac{\epsilon'}{\sqrt{5}})$, we have
\begin{equation}
\begin{split}
    |x_0+\overset{(1)}{\gamma}_1(x_0)|&\geq|x_0|\\
&\geq\int_0^{|x_0|} w(\sqrt{5}s)ds\\
&= \mathcal{W}(x_0)
\end{split}
\end{equation}
since $ w(s)\in\left[0,\frac{1}{4}\right]$. Furthermore, equality in the above holds if and only if $x_0=0$. We conclude that $\gamma\cap C_-=\emptyset$ and hence $\gamma\subset C_+$.

Finally we show that $\overset{(1)}{\gamma}(x_0)$ is $C^1$ and null. Since $\gamma\subset C_+$, we have
\begin{equation}
    \begin{split}
    \left|\frac{\overset{(1)}{\gamma}_1(x_0)}{x_0}-1\right|&=\frac{|\overset{(1)}{\gamma}_1(x_0)-x_0|}{|x_0|}\\
        &\leq\frac{ \mathcal{W}(x_0)}{|x_0|}\\
        &\leq w(\sqrt{5}x_0)\\
        &\rightarrow0\text{ as }x_0\rightarrow0.
    \end{split}
\end{equation}
Hence $\overset{(1)}{\gamma}$ is differentiable at $(0,0)$ with $\frac{d\overset{(1)}{\gamma}_1}{dx_0}=1$ (so in particular $\gamma$ is null at this point). Since this point on $\gamma$ was arbitrary, we conclude that for any $s\in[0,1]$ we can choose near Minkowski coordinates $(\hat{x}_0,\hat{x}_1)$ centred at $\gamma(s)$ such that the reparameterisation of $\gamma$ by $\hat{x}_0$ is differentiable and null at this new point. Returning to the original near Minkowski coordinates $(x_0,x_1)$, since $x_0$ is a smooth time function on $U$
it follows from Lemma \ref{lemma:reparamdifferentiability} that 
$\overset{(1)}{\gamma}(x_0)$
is differentiable and null everywhere. Hence by Corollary \ref{cor:reparam} (c)(ii) we have
\begin{equation}\label{eqn:nullgradient2}
    \left.\frac{d\overset{(1)}{\gamma}_1}{dx_0}\right\vert_{x_0=x_0^*}=\frac{1}{g_{11}(x_0^*)}\left(-g_{01}(x_0^*)+\sqrt{g_{01}(x_0^*)^2-g_{00}(x_0^*)g_{11}(x_0^*)}\right)
\end{equation}
   where we have chosen the positive root since $\overset{(1)}{\gamma}(x_0)$ was assumed to be strictly increasing. 
    Since $g$ is continuous, it follows that $\frac{d\overset{(1)}{\gamma}}{dx_0}$ is continuous.
\end{proof}

\begin{remark}
In \cite{LorentzMeetsLipschitz} it was shown that if the metric is $C^{0,1}$ then causal maximisers admit a $C^{1,1}$ parameterisation, while if the metric is $C^{0,\alpha}$ ($\alpha\in(0,1)$) then they admit only a $C^{1,\frac{\alpha}{4}}$ parameterisation. Proposition \ref{prop:C1nullcurve} shows that, in \underline{two spacetime dimensions}, if the metric is $C^0$ then null maximisers admit a $C^1$ parameterisation. 
\end{remark}

\begin{definition} Let  $(M,g)$ be a time-oriented Lorentzian manifold with $g \in C^0$.    For $p,q\in M$ we define the timelike distance 
    \begin{equation}
        d_g(p,q):=\begin{cases}\sup\{L(\gamma):\gamma\text{ is a  future-directed piecewise}
        \\
        \hspace{1.8cm}\text{smooth timelike curve from $p$ to $q$}\} & \text{if }q\in I^+(p,M)\\
       0 & \text{otherwise.}            
       \end{cases}
   \end{equation} 
\end{definition}

Here, $L(\gamma) = \int \sqrt{-g(\dot{\gamma}(s), \dot{\gamma}(s))} \, ds$ denotes the standard Lorentzian length of the timelike curve $\gamma$.
The following lemma will be used to prove the achronality of the boundary in Proposition \ref{prop:boundarystructure} (b).
\begin{lemma}\label{lemma:semicontinuity}
    The timelike distance function, $d_g$, is lower semicontinuous. 
\end{lemma}

\begin{proof}
    Let $\epsilon>0$ and suppose $(p_n)_{n=1}^\infty$, $(q_n)_{n=1}^\infty$ are sequences in $M$ such that $p_n\rightarrow p\in M$ and $q_n\rightarrow q\in M$ as $n\rightarrow\infty$. We will show that 
    \begin{equation}
        d(p_n,q_n)\geq d(p,q)-\epsilon
    \end{equation}
for all $n$ sufficiently large.

    If $q\notin I^+(p,M)$ then $d(p,q)=0$ and there is nothing to show. Now suppose $q\in I^+(p,M)$ and let $\gamma:[0,1]\rightarrow M$ be a future-directed timelike curve from $p$ to $q$ such that $L(\gamma)\geq d(p,q)-\frac{\epsilon}{2}$. 
    %Note that if $\hat{\gamma}$ is a future-directed timelike curve contained in a near Minkowski neighbourhood then it follows from \eqref{eqn:lightconebound} that $L(\hat{\gamma})\leq \sqrt{2}\epsilon_0$. Hence we can choose $U_p$ and $U_q$ to be sufficiently small near Minkowski neighbourhoods of $p$ and $q$ respectively so that future-directed timelike curves in these neighbourhoods have length at most $\frac{\epsilon}{8}$.

    Next choose $\delta>0$ sufficiently small so that $L(\gamma\vert_{[\delta,1-\delta]})\geq L(\gamma)-\frac{\epsilon}{2}$ and define $p':=\gamma(\delta)$, $q':=\gamma(1-\delta)$. 
Recall that $I^-(p',M)$ and $I^+(q',M)$ are open \cite[Lemma 2.7]{FutureNotAlwaysOpen}. It follows that for $n$ sufficiently large 
%we have $p_n\in I^-(p',M)$ and $q_n\in I^+(q',M)$. Hence, again for $n$ sufficiently large, 
there exist future-directed timelike curves $\gamma^{(1,n)}$ from $p_n$ to $p'$ and
%a future-directed timelike curve 
$\gamma^{(2,n)}$ 
%in $U_q$ 
from $q'$ to $q_n$. 
%In particular, there exists a neighbourhood $N$ of $q$ such that $N\subset I^+(p,M)$. Let $q'\in N\cap I^-(q,M)\subset I^+(p,M)\cap I^-(q,M)$. Since $I^\pm(q',M)$ are open, it follows that for $n$ sufficiently large we have $p_n\in I^-(q',M)$ and $q_n\in I^+(q',M)$ and hence $q_n\in I^+(p_n,M)$ (recall that timelike futures are defined with respect to piecewise smooth curves). It follows that, for $n$ sufficiently large, there exists a future-directed timelike curve $\gamma_n$ from $p_n$ to $q_n$ with $L(\gamma_n)\geq L(\gamma)-\frac{\epsilon}{2}$need more explanation here and hence
We conclude that
    \begin{equation}
        \begin{split}
            d(p_n,q_n)&\geq L(\gamma^{(1,n)})+L(\gamma\vert_{[\delta,1-\delta]})+L(\gamma^{(2,n)})\\
            &\geq L(\gamma)-\frac{\epsilon}{2}\\
            &\geq d(p,q)-\epsilon.
        \end{split}
    \end{equation}
\end{proof}

%\begin{thm}\label{thm:C1nullboundary}
%Let $\iota$ be a $C^0$ extension of $(M,g)$ from Example \ref{example:referenceextensionconformal} into some globally hyperbolic Lorentzian manifold $(\tilde{M},\tilde{g})$. Suppose $\gamma_{u_*}$ can be future extended in $\tilde{M}$ and hence by Proposition \ref{prop:boundarystructure}, $\gamma_u$ can also be extended through a point $\tilde{p}_u\in\partial\iota(M)$ for all $u\in[u_*-\epsilon,u_*+\epsilon]$.  If $p_{u_*-\epsilon}\neq\tilde{p}_{u_*-\epsilon}$ then this accumulation curve is a $C^1$ null curve.
%\end{thm}

%\begin{cor}\label{cor:transverseintersection}
%Let $\iota$ be a $C^0$ extension of $(M,g)$ from Example \ref{example:referenceextensionconformal} into $(\tilde{M},\tilde{g})$ and assume that $\partial\iota(M)$ is not a single point (and hence by Theorem \ref{thm:C1nullboundary} is a $C^1$ null curve). Suppose also that $\iota$ future extends the curve, $\gamma_u$, of constant $u$ in $M$. Let $\tilde{\gamma}^u:(-1,0]\rightarrow\Tilde{M}$ denote this extension, so $\tilde{\gamma}^u\vert_{(-1,0)}=\iota(\gamma_u)$ and $\tilde{\gamma}^u(0)\in\partial\iota(M)$. Then $\tilde{\gamma}^u$ is a $C^1$ curve up to $\partial\iota(M)$ which intersects $\partial\iota(M)$ transversely. 
%\end{cor}
The following lemma concerns the 1+1-dimensional spacetime $(M_t,g_t)$ from Section \ref{toymodel} and will also be used in the proof of Proposition \ref{prop:boundarystructure}.
\begin{lemma}\label{lemma:causalhomotopyexistence}
    Let $\gamma:[0,1]\rightarrow M_t$ be a locally Lipschitz future-directed causal curve in $(M_t,g_t)$ and suppose $p\in J^+(\gamma(0),M_t)\cap J^-(\gamma(1),M_t)$. Then there exists a causal homotopy of $\gamma$ with fixed endpoints (see \cite[Lemma 2.12]{Sbie22a}), denoted $\Gamma:[0,1]\times[0,1]\rightarrow M_t$, such that $p=\Gamma(s;\lambda)$ for some $(s,\lambda)\in[0,1]\times[0,1]$.
\end{lemma}
\begin{proof}
    Since $\gamma$ is a future-directed causal curve then, since $\gamma$ is locally Lipschitz, it is differentiable almost everywhere with $\gamma_u'(s),\gamma_v'(s)\geq0$. Since $p\in J^+(\gamma(0),M_t)\cap J^-(\gamma(1),M_t)$, it follows that $\gamma_u(1)\geq u(p)\geq\gamma_u(0)$ and $\gamma_v(1)\geq v(p)\geq\gamma_v(0)$. By the intermediate value theorem, there exists $\lambda_*\in [0,1]$ such that $\gamma_v(\lambda_*)=v(p)$. We first consider the case $u(p)\leq \gamma_u(\lambda_*)$. 
    %Again by the intermediate value theorem, there exists $\lambda_*\in[0,s_*]$ such that $\gamma_u(\lambda_*)=u(p)$. 
    Define (see Figure \ref{fig:homotopylemma})
    \begin{equation}
        \begin{split}
            \Gamma(s;\lambda):=\begin{cases}
                \big(\gamma_u(0),\gamma_v(2s)\big) & \text{for } s\in\big[0,\frac{\lambda}{2}\big]\\
                \big(\gamma_u(2s -\lambda),\gamma_v(\lambda)\big) & \text{for } s\in\big[\frac{\lambda}{2},\lambda\big] \\
                \gamma(s) & \text{for } s\in[\lambda,1].
            \end{cases}
        \end{split}
    \end{equation}
This is a causal homotopy of $\gamma$ with fixed endpoints, where $s$ is the curve parameter and $\lambda$ the homotopy parameter. Since $\gamma_u(0) \leq u(p) \leq \gamma_u(\lambda_*)$ it follows from the intermediate value theorem that there is $s_* \in [\frac{\lambda_*}{2}, \lambda_*]$ such that 
$\Gamma(s_*;\lambda_*)=p$. The case $u(p)>\gamma_u(s_*)$ follows similarly.
\begin{figure}
    \centering
    \includegraphics[scale=0.25]{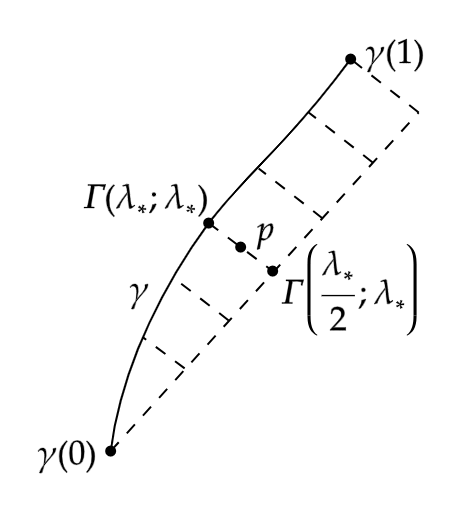}
    \caption{Construction (in the case $u(p)\leq\gamma_u(\lambda_*)$) of $\Gamma(s;\lambda)$, a causal homotopy of $\gamma$ with fixed endpoints such that $\Gamma(s_*;\lambda_*)=p$ for some $\lambda_*\in[0,1]$ and some $s_*\in[\frac{\lambda_*}{2},\lambda_*]$. $\Gamma(s;\lambda)$ is shown as a dotted line for $s\in[0,\lambda)$. For $s\in[\lambda,1]$, points $\Gamma(s;\lambda)$ lie on $\gamma$.}
    \label{fig:homotopylemma}
\end{figure}
\end{proof}

\section{First Structural Properties of Extensions in 1+1-Dimensions} \label{SecBdry}

The aim of this section is to start the investigation of the local structure of the boundary $\partial\iota(M_t)$, where $\iota : M_t \hookrightarrow \tilde{M}$ is a $C^0$-extension across $\{v = 0\}$ of the two dimensional Lorentzian manifold $(M_t,g_t)$ introduced in Section \ref{toymodel}. In Proposition \ref{prop:boundarystructure} we will show that, locally, $\iota$ extends continuously to the boundary of $\overline{M}_t$. Moreover, we make a special choice of near Minkowski coordinates centred at a point on the boundary $\partial \iota(M_t)$ which may be considered as introducing a `normal form' for such continuous extensions. Various general properties of such $C^0$-extensions with respect to these coordinates will be collated, which will be used in the following sections.    We define $\gamma^{u_0}$ to be the curve $u=u_0$ in $M_t$ (parameterised by $v$) and $\lambda^{v_0}$ to be the curve $v=v_0$ (parameterised by $u$). We also define $\tilde{\gamma}^{u_0}:=\iota\circ\gamma^{u_0}$ and $\tilde{\lambda}^{v_0}:=\iota\circ\lambda^{v_0}$. Furthermore, we recall the following definition:

%We will show that the set of $u$ for which $\tilde{p}_u$ exists is an open subset of $(-1,1)$ and that $\tilde{p}_u$ varies continuously with $u$ and therefore traces out curves in $\tilde{M}$. We will show that these curves are either single points or are $C^1$ null curves which are the accumulation (in the sense defined below) of curves $\tilde{\lambda}^v\vert_{(u_1,u_2)}$ as $v\rightarrow0$. As we will see later (Proposition \ref{prop:boundarystructure}), in cases where $\iota(M)$ consists of more than a single point we will be able to think of it as the ``accumulation'' in the following sense of the curves $\tilde{\lambda}^v$ as $v\rightarrow0$.

%In the case where the boundary does not consist of only a single point, we will show that extensions of curves of constant $u$ intersect it transversely and that $\lim_{v\rightarrow0}\iota(\gamma_u(v))\neq\lim_{v\rightarrow0}\iota(\gamma_{u'}(v))$ for $u\neq u'$ (assuming these points both exist).

\begin{definition}\cite[Definition 2.6.3]{ElementsofCausality}\label{defn:accumulationcurve}
    Let $\mu_n:I\rightarrow M$, $n\in\mathbbm{N}$, be a sequence of continuous curves in $(M,g)$, where $g\in C^0$. We shall say that $\mu:I\rightarrow M$ is an accumulation curve of the $\mu_n$'s, or that the $\mu_n$'s accumulate at $\mu$, if there exists a subsequence of $(\mu_{n})_{n=1}^\infty$ that converges to $\mu$ uniformly (with respect to local coordinates) on compact subsets of $I$.
    \end{definition}

\begin{prop}\label{prop:boundarystructure}
Suppose $\iota:M_t\hookrightarrow\tilde{M}$ is a $C^0$-extension of $(M_t,g_t)$ to $(\tilde{M},\tilde{g})$. Suppose $\tau : [-1,0) \to M_t$ is a $C^1$ future-directed causal curve such that  $\lim_{s \to 0} \tau_u(s)=: u_* <1$, $\lim_{s \to 0} \tau_v(s) = 0$ and $\tilde{p}_{u_*}:=\lim_{s \to 0}(\iota \circ \tau)(s) \in \rd \iota(M_t) \subset \tilde{M}$ exists. 

(a)(i) There exists a near Minkowski neighbourhood, $(\tilde{U},\tilde{\varphi}=(\tilde{x}_0,\tilde{x}_1))$ centred at $\tilde{p}_{u_*}$ and $\epsilon>0$, $v_*<0$ such that $\iota:[u_*-\epsilon,u_*+\epsilon]\times[v_*,0)\hookrightarrow \tilde{U}$ extends continuously to $\iota:[u_*-\epsilon,u_*+\epsilon]\times[v_*,0]\hookrightarrow\tilde{U}$ and \begin{equation}\label{eqn:monotonicityconventions}
    \frac{\partial\tilde{x}_1}{\partial u}<0<\frac{\partial\tilde{x}_0}{\partial u},\frac{\partial\tilde{x}_0}{\partial v},\frac{\partial\tilde{x}_1}{\partial v}
\end{equation}
holds on the connected component of $\iota(M_t) \cap \tilde{U}$ that contains $\iota\big([u_* - \epsilon, u_* + \epsilon] \times [v_*,0)\big)$.\footnote{Here, and in the following, by slight abuse of notation we have denoted the restriction of $\iota$ to $[u_*-\epsilon,u_*+\epsilon]\times[v_*,0)$ as well its continuous extension to $[u_*-\epsilon,u_*+\epsilon]\times[v_*,0]$  by the same symbol.}

(ii) For all $u\in[u_*-\epsilon,u_*+\epsilon]$, parameterising the extension of $\tilde{\gamma}^u$ up to $\partial\iota(M_t)$ by $\tilde{x}_0$ gives a $C^1$-regular null curve   with $\frac{d\overset{(1)}{\tilde{\gamma}^u}_1}{d\tilde{x}_0}>0$ (including at the boundary point). 

(iii) Define the (a priori only continuous) curve $\tilde{\lambda}^0:[u_*-\epsilon,u_*+\epsilon]\rightarrow\tilde{U}$ by  $\tilde{\lambda}^0(u)=\iota(u,0)$.\footnote{Note that the image of this curve may consist of a single point - see Example \ref{example:cornerextension}.} Then the interior of the region bound by the curves $\tilde{\gamma}^{u_*-\epsilon}\vert_{[v_*,0)}$, $\tilde{\gamma}^{u_*+\epsilon}\vert_{[v_*,0)}$, $\tilde{\lambda}^{v_*}\vert_{[u_*-\epsilon,u_*+\epsilon]}$ and $\tilde{\lambda}^{0}\vert_{[u_*-\epsilon,u_*+\epsilon]}$ is equal to $\iota\big{(}(u_*-\epsilon,u_*+\epsilon)\times(v_*,0)\big{)}$.

(b) Suppose $\tilde{\lambda}^0(u_*-\epsilon)\neq\tilde{\lambda}^0(u_*)$ or $\tilde{\lambda}^0(u_*+\epsilon)\neq\tilde{\lambda}^0(u_*)$. Then the image of $\tilde{\lambda}^0\vert_{[u_*-\epsilon,u_*+\epsilon]}$ can be parameterised by $\tilde{x}_0$ to give a $C^1$ null curve $\overset{(1)}{\tilde{\lambda}^0}(\tilde{x}_0)$ with $\frac{d\overset{(1)}{\tilde{\lambda}^0}_1}{d\tilde{x}_0}<0$ which intersects the extension of $\tilde{\gamma}^u$ transversely for all $u\in[u_*-\epsilon,u_*+\epsilon]$.\footnote{For emphasis: we do not claim here that there is a surjective function  $h : I \to [u_* - \epsilon, u_* + \epsilon]$ such that $\left(\tilde{x}_0\circ\tilde{\lambda}^0 \circ h\right)(s) = s$, although this will follow eventually from Theorem \ref{thm:C0structure} where we will show that $\tilde{\lambda}^0:[u_*-\epsilon,u_*+\epsilon]\rightarrow\tilde{U}$ is injective. Here we claim that the image of $\tilde{\lambda}^0\vert_{[u_*-\epsilon,u_*+\epsilon]}$ equals the image of $\overset{(1)}{\tilde{\lambda}^0}$.}
\end{prop}

Let us remark that there are $C^0$-extensions across $\{v=0\}$ for which the assumption $\tilde{\lambda}^0(u_*-\epsilon)\neq\tilde{\lambda}^0(u_*)$ made in part (b) is not satisfied -- see Example \ref{example:cornerextension}. Note also that in part (b) we are not assuming $\tilde{\lambda}^0(u)\neq\tilde{\lambda}^0(u_*)$ for $u\in(u_*,u_*+\epsilon]$ (although this will follow from Proposition \ref{prop:nomixedcornerextension}).

\begin{proof} (a)(i) Let $(\Tilde{U},\tilde{\varphi}=(\tilde{x}_0,\tilde{x}_1))$ be a near Minkowski neighbourhood centred at $\Tilde{p}_{u_*}$. Choose $s_0 \in (-1,0)$ sufficiently close to $0$ such that $\tilde{\tau}([s_0,0))\subset\tilde{U}$, where $\tilde{\tau}:=\iota\circ\tau$. If necessary, after replacing $\tilde{x}_0$ with $-\tilde{x}_0$ we can arrange that the time orientation on $\Tilde{U}$ agrees with the time orientation on $\tilde{\tau}|_{[s_0,0)}$ induced from $M_t$. In other words, $\partial_{\tilde{x}_0}$ as well as   $\iota_*\partial_u=\frac{\partial\tilde{x}_0}{\partial u}\partial_{\tilde{x}_0}+\frac{\partial\tilde{x}_1}{\partial u}\partial_{\tilde{x}_1}$ and $\iota_*\partial_v=\frac{\partial\tilde{x}_0}{\partial v}\partial_{\tilde{x}_0}+\frac{\partial\tilde{x}_1}{\partial v}\partial_{\tilde{x}_1}$ (which are non-zero and null since $\iota$ is an isometry) are future-directed along $\tilde{\tau}|_{[s_0,0)} \subset \tilde{U}$. Also note that the light cone bound \eqref{eqn:lightconebound} implies $\frac{\partial\tilde{x}_0}{\partial u},\frac{\partial\tilde{x}_0}{\partial v}, \frac{\partial\tilde{x}_1}{\partial u}, \frac{\partial\tilde{x}_1}{\partial v}\neq0$ on $\tilde{U}\cap\iota(M_t)$.
This gives directly $\frac{\partial\tilde{x}_0}{\partial u},\frac{\partial\tilde{x}_0}{\partial v} >0$ on $\tilde{\tau}|_{[s_0,0)} \subset \tilde{U}$. Furthermore, after replacing $\tilde{x}_1$ with $-\tilde{x}_1$ if necessary and using again \eqref{eqn:lightconebound}, we can assume  that
\begin{equation}\label{eqn:monotonicity}
    \frac{\partial\tilde{x}_1}{\partial u}<0<\frac{\partial\tilde{x}_0}{\partial u},\frac{\partial\tilde{x}_0}{\partial v},\frac{\partial\tilde{x}_1}{\partial v}
\end{equation} 
on $\tilde{\tau}|_{[s_0,0)} \subset \tilde{U}$.
Since $\frac{\partial\tilde{x}_0}{\partial u},\frac{\partial\tilde{x}_0}{\partial v}, \frac{\partial\tilde{x}_1}{\partial u}, \frac{\partial\tilde{x}_1}{\partial v}\neq0$ in $\tilde{U}\cap\iota(M_t)$, it follows that \eqref{eqn:monotonicity} holds on the connected component of $\iota(M_t)\cap\tilde{U}$ containing $\tilde{\tau}\big([s_0,0)\big)$.

\textbf{Claim:} There exist $v_*<0<\epsilon$ such that\footnote{We will write $A\ssubset B$ to denote that $A$ is compactly contained in $B$.} $\iota([u_*-\epsilon,u_*+\epsilon]\times[v_*,0))\ssubset\tilde{U}$.

%We now show there exist $v_*<0<\epsilon$ such that $\iota(u,0):=\lim_{v\rightarrow0}\tilde{\gamma}^u(v)\in\tilde{U}$ exists for all $u\in[u_*-\epsilon,u_*+\epsilon]$ and $\iota([u_*-\epsilon,u_*+\epsilon]\times[v_*,0])\subset\tilde{U}$. To do this, we will show that for $v_*<0<\epsilon$ both sufficiently small and any $u\in[u_*-\epsilon,u_*+\epsilon]$, the curve $\tilde{\gamma}^u\vert_{[v_*,0)}$ lies in some set which is compactly contained\footnote{We will write $A\ssubset B$ to denote that $A$ is compactly contained in $B$.} in $\tilde{U}$. This compactness, combined with the monotonicity properties \eqref{eqn:monotonicityconventions}, implies that $\lim_{v\rightarrow0}\tilde{\gamma}^u(v)$ exists and lies in $\tilde{U}$.is this enough? or say a bit more? I've changed some of the notation below. I think it makes things clearer?

We begin with $u\leq u_*$. By reducing $|s_0|$ if necessary, it follows from \eqref{eqn:gradientboundcausal} that 
\begin{equation} \label{EqCpct}
J^+(\tilde{\tau}(s_0),\tilde{U})\cap J^-(\tilde{p}_{u_*},\tilde{U})\ssubset\tilde{U}\;.
\end{equation}
A priori we may have $u(\tau(s_0)) = u_*$, i.e., the segment $\tau|_{[s_0, 0)}$ may be a curve of constant $u = u_*$. If this is the case, we move to a point on $\tau|_{[s_0, 0)}$ even closer to its future end point and
then modify $\tau$ to the past of this point by adding a short timelike segment to the past. In this way we can, without loss of generality, assume $u(\tau(s_0)) < u_*$ in addition to \eqref{EqCpct}. This will be needed further below.

For any $s' \in (s_0,0)$ the restriction $\tau|_{[s_0,s']}$ is a $C^1$-causal curve in $M_t$ such that, by Lemma \ref{lemma:causalhomotopyexistence}, any point in $J^+\big(\tau(s_0), M_t\big) \cap J^-\big(\tau(s'), M_t \big)= [u(\tau(s_0)), u(\tau(s'))] \times [v(\tau(s_0)),v(\tau(s'))]$ lies on a causal curve in $M_t$ which is causally homotopic to $\tau|_{[s_0,s']}$ with fixed endpoints. It then follows from\footnote{In the cited Lemma, the curve $\gamma$, whose place here is taken by $\tau|_{[s_0,s']}$, is assumed to be timelike. However, this is nowhere needed in the proof, which also holds if $\gamma$ is causal.} \cite[Lemma 2.12]{Sbie22a} that 
\begin{equation}\label{eqn:homotopic}
    \iota \Big([u(\tau(s_0)), u(\tau(s'))] \times [v(\tau(s_0)),v(\tau(s'))]\Big) \subset J^+(\tilde{\tau}(s_0),\tilde{U})\cap J^-(\tilde{p}_{u_*},\tilde{U})\ssubset\tilde{U}\end{equation}
for all $s' \in (s_0,0)$ (see Figure \ref{fig:homotopy}). Letting $s' \to 0$ gives $\iota \Big([u(\tau(s_0)), u_*)\times [v(\tau(s_0)),0)\Big) \ssubset \tilde{U}$. Let  $\epsilon :=u_*-u(\tau(s_0))>0$ and $v_0 := v(\tau(s_0))<0$. By compactness and the continuity of $\iota$, we obtain 
$\iota \Big([u_* - \epsilon, u_*]\times [v_0,0)\Big) \ssubset \tilde{U}$.

\begin{figure}[h]
    \centering
    \begin{minipage}{0.48\textwidth}
            \centering
    \includegraphics[scale=0.25]{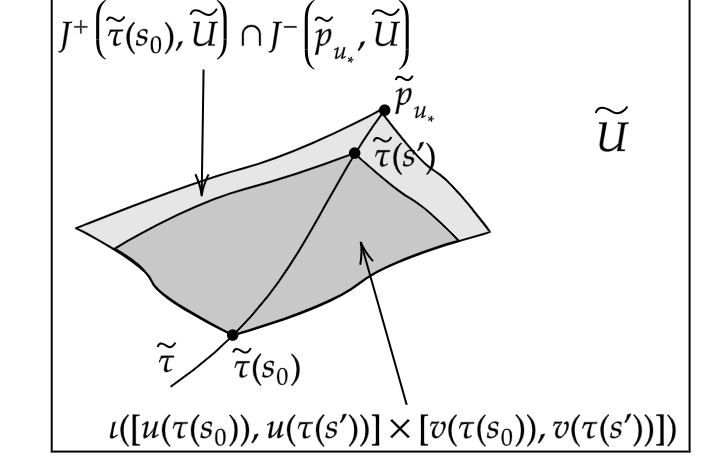}
    \caption{The homotopy argument used to show $\iota\left([u_*-\epsilon,u_*]\times[v_0,0)\right)\ssubset\tilde{U}$.}
    \label{fig:homotopy}
    \end{minipage}%
     \minipage{0.04\textwidth}
    \includegraphics[scale=0.05]{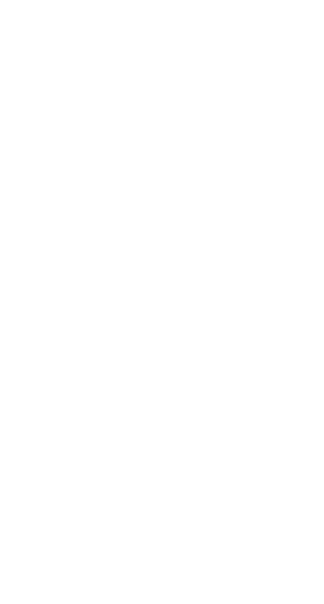}
 \endminipage
    \begin{minipage}{0.48\textwidth}
        \centering
\includegraphics[scale=0.28]{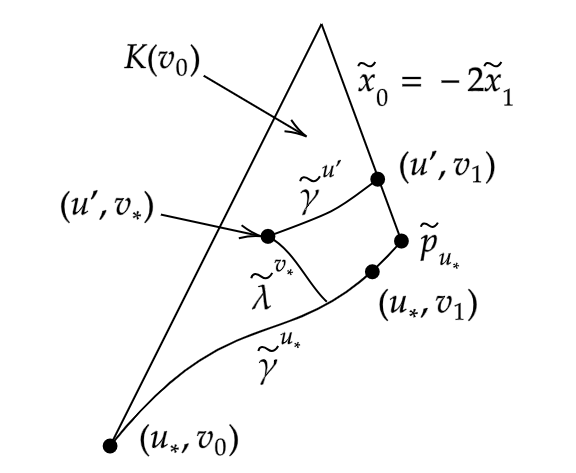}
    \caption{The contradiction argument used to show that 
    $\tilde{\gamma}^{u'}([v_*,0)) \subset K(v_0)$
    for $u'\in(u_*,u_*+\epsilon]$.}
    \label{fig:Kv0}
    \end{minipage}
\end{figure}

Next we consider $u>u_*$. Reducing $|v_0|$ if necessary, define $K(v_0)$ to be the compact subset of $\tilde{U}$ enclosed by $\tilde{\gamma}^{u_*}\vert_{[v_0,0)}$; the line $\tilde{x}_0=-2\tilde{x}_1$; and the line through $(u_*,v_0)$ with gradient $2$ (see Figure \ref{fig:Kv0}). 
Recall here our convention regarding near Minkowski neighbourhoods that \eqref{eqn:lightconebound} holds with $\zeta = 1$. In 1+1-dimensions this implies that right going null curves in near Minkowski coordinates must have slope between $\frac{1}{2}$ and $2$ and left going ones between $-2$ and $-\frac{1}{2}$. This motivates the choice of the straight lines in Figure \ref{fig:Kv0}.
Let $v_*\in(v_0,0)$. Reducing $\epsilon>0$ if necessary, it follows from \eqref{eqn:monotonicity} that $\tilde{\lambda}^{v_*}\vert_{(u_*,u_*+\epsilon]}\subset \text{int}(K(v_0))$. Let $u'\in(u_*,u_*+\epsilon]$. Suppose there exists $v_1\in(v_*,0)$ such that $\tilde{\gamma}^{u'}\vert_{[v_*,v_1)}\subset \text{int}(K(v_0))$ and $\iota\big((u',v_1)\big)\in\partial K(v_0)$. From \eqref{eqn:gradientbound1+1} and \eqref{eqn:monotonicity}, the point $\iota\big((u',v_1)\big)$ cannot lie on the boundary line of $K(v_0)$ with gradient $2$. The injectivity of $\iota$ means it also cannot lie on $\tilde{\gamma}^{u_*}\vert_{[v_0,0)}$. Hence this point must lie on the line $\tilde{x}_0=-2\tilde{x}_1$. Consider the curve $\tilde{\lambda}^{v_1}\vert_{[u_*,u']}$. By \eqref{eqn:monotonicity}, this curve must enter the interior of $K(v_0)$ and, since $\iota\big((u',v_1)\big)\in\partial K(v_0)$, it must also exit. From \eqref{eqn:gradientbound1+1} and \eqref{eqn:monotonicity}, in order to do so, $\tilde{\lambda}^{v_1}\vert_{[u_*,u')}$ must intersect $\tilde{\gamma}^{u'}$ or $\tilde{\lambda}^{v_*}$, which contradicts the injectivity of $\iota$. We conclude that $\tilde{\gamma}^{u'}\vert_{[v_*,0)}\subset K(v_0)\ssubset\tilde{U}$ for all $u' \in(u_*, u_* + \epsilon]$. 

Since $v_0 < v_* <0$, this establishes the claim. Since \eqref{eqn:monotonicity} holds in the connected component of $\iota(M_t) \cap \tilde{U}$ containing $\tilde{\tau}\big([s_0,0)\big)$, this also proves the claim about \eqref{eqn:monotonicityconventions} in (a)(i). 

We now extend $\iota$ to $[u_* - \epsilon, u_* + \epsilon] \times [v_*,0]$ by setting $\iota(u,0):=\lim_{v\rightarrow0}\iota(u,v)\in\tilde{U}$. This limit exists by compactness and the monotonicity properties \eqref{eqn:monotonicityconventions}.

%As above, combining this with the monotonicity properties \eqref{eqn:monotonicityconventions} allows us to define $\iota(u,0):=\lim_{v\rightarrow0}\iota(u,v)\in\tilde{U}$ for any $u\in(u_*,u_*+\epsilon]$, and we also have $\iota([v_*,0]\times[u_*,u_*+\epsilon])\subset\tilde{U}$.

%Since $|v_*|<|v_1|$, we have therefore extended $\iota$ to $\iota:[u_*-\epsilon,u_*+\epsilon]\times[v_*,0]\hookrightarrow\tilde{U}$I could replace $v_1$ with $v_*$ to reduce the number of labels but I think this would be confusing since in the first argument we have $v_1<v_0$  whereas in the second we have $v_*>v_0$. 
We now show this extension is continuous. Let $u'\in[u_*-\epsilon,u_*+\epsilon]$ and let $\tilde{N}$ be any neighbourhood of $\tilde{p}_{u'}:=\iota(u',0)$. Repeating the arguments above -- with a $C^1$ future-directed causal curve $\tau' : [-1,0) \to M$ with $\lim_{s \to 0} \tau'_u(s) = u'$ and $\lim_{s \to 0} \tau'_v(s) = 0$ instead of $\tau$ -- if $\tilde{U}'\subset \tilde{N}$ is a near Minkowski neighbourhood centred on $\iota(u',0)$, we can choose $v_*'<0<\epsilon'$ such that $\iota([u'-\epsilon',u'+\epsilon']\times[v_*',0])\subset\tilde{U}'\subset{N}$.  We conclude that the extension of $\iota$ defined above is continuous. 

(a)(ii) Let $u\in[u_*-\epsilon,u_*+\epsilon]$. Lemma \ref{lemma:reparamdifferentiability} implies that reparameterising $\tilde{\gamma}^u\vert_{[v_*,0)}$ by $\tilde{x}_0$ gives a $C^1$ null curve, which we denote $\overset{(1)}{\tilde{\gamma}^u}(\tilde{x}_0)$. 
It follows from \eqref{eqn:monotonicityconventions} that $\overset{(1)}{\tilde{\gamma}_1^u}(\tilde{x}_0)$ is strictly increasing and hence its derivative satisfies \eqref{eqn:2dnullderivativeformula} with the positive root.
Let ${\alpha_u}:=\tilde{x}_0(u,0)$. By part (a)(i), $\overset{(1)}{\tilde{\gamma}_1^u}$  extends continuously to $[\tilde{x}_0\big(\tilde{\gamma}^u(v_*)\big), \alpha_u]$ and we denote this extended curve by the same symbol.
 
 Let $\epsilon'=\alpha_u-\tilde{x}_0\big(\tilde{\gamma}^u(v_*)\big)$ and define the function 
\begin{equation}
    \begin{split}
        f_{u}:[{\alpha_u}-\epsilon',{\alpha_u}]&\rightarrow\mathbbm{R}_{>0}\\
        \tilde{x}_0&\mapsto \frac{1}{\tilde{g}_{11}(\tilde{x}_0)}\left(-\tilde{g}_{01}(\tilde{x}_0)+\sqrt{\tilde{g}_{01}(\tilde{x}_0)^2-\tilde{g}_{00}(\tilde{x}_0)\tilde{g}_{11}(\tilde{x}_0)}\right),
    \end{split}
\end{equation}
where we write $\tilde{g}_{\mu \nu}(\tilde{x}_0)$ to denote $ \tilde{g}_{\mu \nu}\big(\overset{(1)}{\tilde{\gamma}^u}(\tilde{x}_0)\big)$. The continuity of $\overset{(1)}{\tilde{\gamma}^u}$ on $[\tilde{x}_0\big(\tilde{\gamma}^u(v_*)\big), \alpha_u]$ combined with the continuity of $\tilde{g}_{\mu\nu}$ ensures that $f_u$ is continuous.
It follows that
\begin{equation}
    \begin{split}
        \lim_{\tilde{x}_0\nearrow{\alpha_u}}\frac{\overset{(1)}{\tilde{\gamma}^u}_1(\tilde{x}_0)-\overset{(1)}{\tilde{\gamma}^u}_1({\alpha_u})}{\tilde{x}_0-{\alpha_u}}&=\lim_{\tilde{x}_{0}\nearrow{\alpha_u}}\frac{\int_{\alpha_u}^{\tilde{x}_{0}}\frac{d\overset{(1)}{\tilde{\gamma}^u}_1}{ds}ds}{\tilde{x}_0-{\alpha_u}}\\
        &=\lim_{\tilde{x}_0\nearrow{\alpha_u}}\frac{\int_{\alpha_u}^{\tilde{x}_0}f_{u}(s)ds}{\tilde{x}_0-{\alpha_u}}\\
        &=f_{u}({\alpha_u})\\
        &>0
    \end{split}
\end{equation}
where we have used the continuity of $f_{u}$ in the penultimate line. We conclude that $\overset{(1)}{\tilde{\gamma}^u}$ extends to $\partial\iota(M_t)$  as a $C^1$ null curve and $\frac{d\overset{(1)}{\tilde{\gamma}^u_1}}{d\tilde{x}_0}\vert_{\tilde{x}_0=\alpha_u}>0$. 

(a)(iii) Let $\tilde{A}$ denote the interior of the region bound by $\tilde{\gamma}^{u_*-\epsilon}\vert_{[v_*,0)}$, $\tilde{\gamma}^{u_*+\epsilon}\vert_{[v_*,0)}$, $\tilde{\lambda}^{v_*}\vert_{[u_*-\epsilon,u_*+\epsilon]}$ and $\tilde{\lambda}^{0}\vert_{[u_*-\epsilon,u_*+\epsilon]}$. Note that for $v\in[v_*,0)$, \eqref{eqn:monotonicityconventions} implies that $\tilde{x}_0(u,v)$ is strictly increasing in $u$ and $\tilde{x}_1(u,v)$ is strictly decreasing in $u$. Thus
\begin{equation}\label{eqn:monotonicityatboundary}\begin{split}
        &\bullet\tilde{x}_0(u,0) \text{ is non-decreasing in }u \text{ and }\\ &\bullet\tilde{x}_1(u,0)\text{ is non-increasing in }u.
\end{split}
\end{equation}

We first show that $\tilde{A}\subset \iota\big((u_* - \epsilon, u_* + \epsilon) \times (v_*,0)\big)$. This argument is depicted in Figure \ref{fig:interior}. Let $\tilde{q}\in\tilde{A}$
%be a point in this interior 
%which is not contained in $\iota\big((u_* - \epsilon, u_*) \times (v_*,0)\big)$. C
and consider the past-directed timelike curve in $\tilde{U}$ of constant $\tilde{x}_1$, starting at $\tilde{q}$. By the monotonicity properties of $\tilde{x}_0(u,0)$ and $\tilde{x}_1(u,0)$, as well as \eqref{eqn:gradientbound1+1} and \eqref{eqn:monotonicityconventions}, this curve must intersect $\tilde{\lambda}^{v_*}|_{[u_* - \epsilon, u_* + \epsilon]}$ or $\tilde{\gamma}^{u_* - \epsilon}|_{[v_*,0)}$. We denote by $\tilde{\sigma} : [0,1] \to \tilde{U}$ the future-directed timelike curve of constant $\tilde{x}_1$ starting at this intersection point and ending at $\tilde{q}$. Consider the largest $s_0 \in [0,1]$ such that $\tilde{\sigma}|_{[0,s_0)}$ lies in $\iota(M_t)$. By the openness of $\iota(M_t)$ we have $s_0 >0$. Moreover, \eqref{eqn:monotonicityconventions} implies that the timelike curve $\iota^{-1} \circ \tilde{\sigma}|_{[0,s_0)}$ in $M_t$ must enter $(u_* - \epsilon, u_* + \epsilon) \times (v_*,0)$ but cannot leave  this region again since otherwise $\tilde{\sigma}$ would intersect either $\tilde{\gamma}^{u_* + \epsilon}|_{[v_*,0)}$ or $\tilde{\lambda}^0|_{[u_* - \epsilon, u_* + \epsilon]}$. Since $\tilde{\sigma}(s_0)$ does not lie on $\tilde{\gamma}^{u_* + \epsilon}|_{[v_*,0)}$ or $\tilde{\lambda}^0|_{[u_* - \epsilon, u_* + \epsilon]}$, it follows that $\tilde{\sigma}(s_0)\in\iota\big((u_* - \epsilon, u_* + \epsilon) \times (v_*,0)\big)$ and hence $s_0=1$.
%and and we must also have $\tilde{\sigma}(s_0) \in \iota \big((u_* - \epsilon, u_*) \times (v_*,0)\big)$. This set is open, so the maximality of $s_0$ implies $s_0 = 1$. This contradicts the assumption that 
We conclude that $\tilde{q}\in\iota\big((u_* - \epsilon, u_* + \epsilon) \times (v_*,0)\big)$.

We now show that $\iota\big((u_* - \epsilon, u_* + \epsilon) \times (v_*,0)\big)\subset \tilde{A}$. Let $(u,v)\in(u_*-\epsilon,u_*+\epsilon)\times(v_*,0)$ and consider a continuous curve $\mu:[0,1]\rightarrow(u_*-\epsilon,u_*+\epsilon)\times(v_*,0)$ from $(u,v)$ to $\iota^{-1}(\tilde{q})$, where $\tilde{q}\in \tilde{A} \subset\iota\big((u_* - \epsilon, u_* + \epsilon) \times (v_*,0)\big)$ is arbitrary. The continuity of $\iota$ implies that $\iota\circ\mu$ is a continuous curve in $\iota\big((u_* - \epsilon, u_*+\epsilon) \times (v_*,0)\big)$ from $\iota(u,v)$ to $\tilde{q}\in\tilde{A}$. Hence $\iota\circ\mu$ cannot intersect $\tilde{\lambda}^0\vert_{[u_*-\epsilon,u_*+\epsilon]}$. Since $\iota$ is injective on $(u_* - \epsilon, u_*+\epsilon) \times (v_*,0)$, it follows that $\iota\circ\mu$ also cannot intersect $\tilde{\gamma}^{u_*-\epsilon}\vert_{[v_*,0)}$, $\tilde{\gamma}^{u_*+\epsilon}\vert_{[v_*,0)}$ or $\tilde{\lambda}^{v_*}\vert_{[u_*-\epsilon,u_*+\epsilon]}$. We conclude that $\iota(u,v)\in\tilde{A}$.

(b)
By part (a)(i) we have 
\begin{equation} \label{EqConvU}
    \lim_{v \to 0} \tilde{\lambda}^v\vert_{[u_* - \epsilon, u_* + \epsilon]}(u) = \tilde{\lambda}^0\vert_{[u_* - \epsilon, u_* + \epsilon]}(u) \qquad \textnormal{ for all } u \in [u_* - \epsilon, u_* + \epsilon] \;.     
\end{equation}
 In the $\tilde{x}_0$, $\tilde{x}_1$ coordinates, \eqref{eqn:monotonicityconventions} implies that this convergence is monotone, and hence by Dini's theorem it is also uniform.  In the following, to ease notation we may write $(u,v)$ or $(u,0)$ instead of $\iota\left((u,v) \right)$ or $\iota\left((u,0) \right)$ when labelling points in $\tilde{U}$ .

We will apply the theory of accumulation curves (contained for example in \cite{ElementsofCausality}, \cite{ChruscielGrant}, \cite{AspectsofC0causaltheory}) to the curves $\tilde{\lambda}^v\vert_{[u_*-\epsilon,u_*+\epsilon]}$. To do this, we first need to reparameterise these curves (over a common interval) so that they are equicontinuous.
Let $v\in[v_*,0)$. By Corollary \ref{cor:reparam} (b), $\tilde{\lambda}^v\vert_{[u_*-\epsilon,u_*+\epsilon]}$ can be reparameterised by $\tilde{x}_0$. We denote this reparameterised curve by $\overset{(1)}{\tilde{\lambda}^v}$.
Let $\beta_-:=\tilde{x}_0(u_*-\epsilon,v_*)<0\leq\beta_+:=\tilde{x}_0(u_* + \epsilon, 0)$ and, for fixed $v\in[v_*,0)$, define 
\begin{equation}
    \begin{split}
        s(t;v):[\beta_-,\beta_+]&\rightarrow [\tilde{x}_0(u_*-\epsilon,v),\tilde{x}_0(u_*+\epsilon,v)]\\
        t&\mapsto \frac{\beta_+-t}{\beta_+-\beta_-}\tilde{x}_0(u_*-\epsilon,v)+\frac{t-\beta_-}{\beta_+-\beta_-}\tilde{x}_0(u_*+\epsilon,v).
    \end{split}    
\end{equation}
We can then reparameterise $\tilde{\lambda}^v\vert_{[u_*-\epsilon,u_*+\epsilon]}$ as follows 
\begin{equation}
    \begin{split}
        \overset{(2)}{\tilde{\lambda}^v}:[\beta_-,\beta_+]&\rightarrow\tilde{U}\\
        t&\mapsto\overset{(1)}{\tilde{\lambda}^v}(s(t;v)).
    \end{split}
\end{equation}
This gives the parameterisation over a common interval. Before we proceed, we claim that
\begin{equation} \label{EqAddEstimatePara}
  \tilde{x}_0(u_* + \epsilon,0)  - \tilde{x}_0(u_* - \epsilon,0) >0\;.
\end{equation}
Note that \eqref{eqn:monotonicityatboundary} implies $\tilde{x}_0(u_* + \epsilon,0) \geq \tilde{x}_0(u_* - \epsilon,0)$. Suppose for a contradiction that equality holds. Recall that by assumption and \eqref{eqn:monotonicityatboundary} we have $\tilde{\lambda}^0(u_* - \epsilon) \neq \tilde{\lambda}^0(u_* + \epsilon)$, so by \eqref{eqn:monotonicityatboundary} it follows that $\tilde{x}_1(u_*+\epsilon,0)<\tilde{x}_1(u_*-\epsilon,0)$. 
%It follows from continuity that we can choose $v'<0$ sufficiently close to 0 that
%\begin{equation}
%    \left|\frac{\tilde{x}_0(u_*+\epsilon,v')-\tilde{x}_0(u_*-\epsilon,v')}{\tilde{x}_1(u_*+\epsilon,v')-\tilde{x}_1(u_*-\epsilon,v')}\right|<\frac{1}{2}.
%\end{equation}
By \eqref{eqn:monotonicityconventions} we can reparameterise $\tilde{\lambda}^{v}|_{[u_* - \epsilon, u_* + \epsilon]}$ by $\tilde{x}_1$:
\begin{equation*}
    [\tilde{\lambda}^{v}_1 (u_* + \epsilon), \tilde{\lambda}^{v}_1 (u_* - \epsilon)] \ni \mathfrak{s} \mapsto \overset{(3)}{\tilde{\lambda}^{v}}(\mathfrak{s}) \;.
\end{equation*}
By the mean value theorem, there exists $\mathfrak{s}'_v\in (\tilde{\lambda}^{v}_1 (u_* + \epsilon), \tilde{\lambda}^{v}_1 (u_* - \epsilon))$ such that
\begin{equation*}
    \frac{\overset{(3)}{\tilde{\lambda}^{v}}_0\big(\tilde{\lambda}^{v}_1 (u_* - \epsilon)\big) - \overset{(3)}{\tilde{\lambda}^{v}}_0\big(\tilde{\lambda}^{v}_1 (u_* + \epsilon)\big)}{\tilde{\lambda}^{v}_1 (u_* - \epsilon) - \tilde{\lambda}^{v}_1 (u_* + \epsilon)} = \frac{d}{d \mathfrak{s}}\overset{(3)}{\tilde{\lambda}^{v}}_0(\mathfrak{s}'_v)
\end{equation*}
The left-hand side goes to zero as $v\to0$ and hence so does the right-hand side. However the parameterisation implies $\frac{d}{d \mathfrak{s}}\overset{(3)}{\tilde{\lambda}^{v}}_1(\mathfrak{s}) = 1$. This contradicts \eqref{eqn:derivbound2} (after reversing the parameterisation of $\overset{(3)}{\tilde{\lambda}^{v}}$  so that the curve is future directed) since the curves $\tilde{\lambda}^{v}$ are null. This proves \eqref{EqAddEstimatePara}.

Returning to the parameterisation $\overset{(2)}{\tilde{\lambda}^v}$, for each $v\in[v_*,0)$, we have
\begin{equation}
    \begin{split}
        |\overset{(2)}{\tilde{\lambda}^v}_0(t')-\overset{(2)}{\tilde{\lambda}^v}_0(t)| =|s(t';v)-s(t;v)|
         = \frac{|t'-t|}{\beta_+-\beta_-}\left(\tilde{x}_0(u_*+\epsilon,v)-\tilde{x}_0(u_*-\epsilon,v)\right)
    \end{split}
\end{equation}
It now follows from \eqref{eqn:monotonicityconventions} together with \eqref{EqAddEstimatePara} that for $v<0$ close enough to $0$ we have
\begin{equation} \label{EqLipschitzAnti}
     \frac{1}{2}\underbrace{\frac{(\tilde{x}_0(u_*+\epsilon,0)-\tilde{x}_0(u_*-\epsilon,0))}{\beta_+ - \beta_-}}_{>0} |t' - t| \leq |\overset{(2)}{\tilde{\lambda}^v}_0(t')-\overset{(2)}{\tilde{\lambda}^v}_0(t)| \leq |t' - t| \;.
\end{equation}
 Combining this with \eqref{eqn:gradientbound1+1}, we have
\begin{equation}
    \begin{split}
        |\overset{(2)}{\lambda^v}_1(t')-\overset{(2)}{\lambda^v}_1(t)|&=|\overset{(1)}{\tilde{\lambda}^v}_1(s(t';v))-\overset{(1)}{\tilde{\lambda}^v}_1(s(t;v))|\\
        &<2\left|s(t';v)-s(t;v)\right|\\
        &\leq2|t'-t|.
    \end{split}
\end{equation}
We conclude that, for $v<0$ close enough to $0$, $\overset{(2)}{\tilde{\lambda}^v}:[\beta_-,\beta_+]\mapsto\tilde{U}$ is a family of uniformly Lipschitz future-directed causal curves. 
%The continuity of $\iota$ implies that $\overset{(2)}{\tilde{\lambda}^0}(t):=\iota(u(\overset{(2)}{\tilde{\lambda}^{v_*}}(t)),0)=\lim_{v\rightarrow0}\overset{(2)}{\tilde{\lambda}^v}(t)$ for all $t\in[\beta,0]$. 
Moreover, by \eqref{EqConvU} we have $\overset{(2)}{\tilde{\lambda}^v}(\beta_-)=\tilde{\lambda}^v(u_*-\epsilon) \to \tilde{\lambda}^0(u_* - \epsilon)$ and $\overset{(2)}{\tilde{\lambda}^v}(\beta_+)=\tilde{\lambda}^v(u_*+\epsilon) \to \tilde{\lambda}^0(u_* + \epsilon)$ as $v \to 0$. 
It now follows from Arzel\`a-Ascoli 
%We now use \cite[Proposition 2.6.1]{ElementsofCausality} to deduce 
that a subsequence, $\overset{(2)}{\tilde{\lambda}^{v_n}}$, accumulates at some Lipschitz continuous curve $\overset{(2)}{\tilde{\lambda}^0} : [\beta_-, \beta_+] \to \tilde{U}$ from $\tilde{\lambda}^0(u_* - \epsilon)$ to $\tilde{\lambda}^0(u_* + \epsilon)$.
%\footnote{Note that this Proposition assumes the metric to be $C^2$, however this is only required to show that the accumulation curve is locally Lipschitz and causal - a result which is proved for $C^0$ metrics in \cite[Theorem 1.6]{ChruscielGrant} (see \cite[Remark 2.6.2]{ElementsofCausality}).}
At points where $\overset{(2)}{\tilde{\lambda}^0}$ is differentiable, it follows from the left-hand side of \eqref{EqLipschitzAnti} that the derivative cannot vanish.  Hence, \cite[Lemma 2.24]{AspectsofC0causaltheory} implies that this accumulation curve is  causal.\footnote{Note that causal curves in \cite{AspectsofC0causaltheory} are defined to be locally Lipschitz.} %\cite[Proposition 2.6.1]{ElementsofCausality} together with \cite[Proposition 1.5]{ChruscielGrant} (see the proof of \cite[Theorem 1.6]{ChruscielGrant}) to infer the existence of a locally Lipschitz causal curve $\overset{(2)}{\tilde{\lambda}^0} : [\beta_-, \beta_+] \to \tilde{U}$ and a  sequence $v_n \to 0$ such that the $\overset{(2)}{\tilde{\lambda}^{v_n}}$ converge uniformly to $\overset{(2)}{\tilde{\lambda}^0}$. In particular, $\overset{(2)}{\tilde{\lambda}^0}$ is an accumulation curve of the $\overset{(2)}{\tilde{\lambda}^{v_n}}$.

We now show that the image of $\overset{(2)}{\tilde{\lambda}^0}$ is equal to the image of $\tilde{\lambda}^0\vert_{[u_* - \epsilon, u_* + \epsilon]}$. Since $\mathrm{Im}\big(\tilde{\lambda}^{v}\vert_{[u_* - \epsilon, u_* + \epsilon]}\big) = \mathrm{Im} \big(\overset{(2)}{\tilde{\lambda}^{v}} \big)$ for all $v\in[v_*,0)$, we have, for $t \in [\beta_-, \beta_+]$, that $$\overset{(2)}{\tilde{\lambda}^0}(t) = \lim_{n \to \infty}\overset{(2)}{\tilde{\lambda}^{v_n}}(t) = \lim_{n \to \infty}\tilde{\lambda}^{v_n}\vert_{[u_* - \epsilon, u_* + \epsilon]}(u_n)$$ for some sequence $(u_n)_{n=1}^\infty$ in $[u_*-\epsilon,u_*+\epsilon]$. Now, since the convergence in \eqref{EqConvU} is uniform, it follows that the right hand side lies in the image of $\tilde{\lambda}^0\vert_{[u_* - \epsilon, u_* + \epsilon]}$. Thus $\overset{(2)}{\tilde{\lambda}^0}$ maps into the image of $\tilde{\lambda}^0\vert_{[u_* - \epsilon, u_* + \epsilon]}$. Vice versa, consider $\tilde{\lambda}^0\vert_{[u_* - \epsilon, u_* + \epsilon]}(u) = \lim_{n \to \infty}\tilde{\lambda}^{v_n}(u) = \lim_{n \to \infty} \overset{(2)}{\tilde{\lambda}^{v_n}}(t_n)$ for some sequence $(t_n)_{n=1}^\infty$ in $[\beta_-,\beta_+]$. Again, the uniform convergence of the $\overset{(2)}{\tilde{\lambda}^{v_n}}$ implies that the right hand side is in the image of $\overset{(2)}{\tilde{\lambda}^0}$ and thus $\mathrm{Im}\big(\tilde{\lambda}^0\vert_{[u_* - \epsilon, u_* + \epsilon]}\big) = \mathrm{Im} \big(\overset{(2)}{\tilde{\lambda}^0} \big)$.
%By compactness, after taking a subsequence, we have $t_n \to t_\infty \in [\beta_-, \beta_+]$ and thus $\tilde{\lambda}^v\vert_{[u_* - \epsilon, u_* + \epsilon]}(u) = \lim_{n \to \infty} \overset{(2)}{\tilde{\lambda}^{v_n}}(t_\infty) = \overset{(2)}{\tilde{\lambda}^0}(t_\infty) $. 

%and $\lim_{v\rightarrow0}\overset{(2)}{\tilde{\lambda}^v}(t)\in\tilde{\lambda}^0\vert_{[u_*-\epsilon,u_*+\epsilon]}$ for all  $t\in[\beta_-,\beta_+]$. Hence $\overset{(2)}{\tilde{\lambda}^v}$ accumulate as $v\rightarrow0$ to a curve} $\overset{(2)}{\tilde{\lambda}^0}$ which is a reparameterisation}  of $\tilde{\lambda}^0\vert_{[u_*-\epsilon,u_*+\epsilon]}$. It follows from \cite[Theorem 1.6]{ChruscielGrant} that  $\overset{(2)}{\tilde{\lambda}^0}$ is a locally Lipschitz causal curve. 

For any $v\in[v_*,0)$, Corollary \ref{cor:reparam} (c)(i) implies $\overset{(2)}{\tilde{\lambda}^v}|_{[\beta_-,\beta_+]}$ is achronal with respect to smooth timelike curves in $\tilde{U}$. It follows from the lower semicontinuity of the timelike distance function (Lemma \ref{lemma:semicontinuity}) that $d_{\tilde{g}\vert_{\tilde{U}}}(\overset{(2)}{\tilde{\lambda}^0}(t),\overset{(2)}{\tilde{\lambda}^0}(t'))=0$ for all $t,t'\in[\beta_-,\beta_+]$. We conclude from Proposition \ref{prop:C1nullcurve} that $\overset{(2)}{\tilde{\lambda}^0}$ can be reparameterised as a $C^1$ null curve. Lemma \ref{lemma:reparamdifferentiability} then implies that parameterising $\overset{(2)}{\tilde{\lambda}^0}$ by $\tilde{x}_0$ (which is a smooth time function on $\tilde{U}$) gives a $C^1$ null curve, which we denote $\overset{(1)}{\tilde{\lambda}^0}(\tilde{x}_0)$. It follows from \eqref{eqn:monotonicityconventions} that $\frac{d\overset{(1)}{\tilde{\lambda}^0}_1}{d\tilde{x}_0} \leq 0$ and hence $\frac{d\overset{(1)}{\tilde{\lambda}^0}_1}{d\tilde{x}_0}$ satisfies \eqref{eqn:2dnullderivativeformula} with the negative root. We conclude that, for all $u\in[u_*-\epsilon,u_*+\epsilon]$, the extension of $\tilde{\gamma}^{u}$ intersects $\tilde{\lambda}^0$ transversely. 
%The same argument holds for any $u\in[u_*-\epsilon,u_*]$ by choosing near a Minkowski neighbourhood centred at $\iota(u,0)$. 
\end{proof}

\begin{remark}\label{remark:futureboundary}
In the setting of Proposition \ref{prop:boundarystructure}, if in part (b) we have the case $\tilde{\lambda}^0(u_*-\epsilon) \neq \tilde{\lambda}^0(u_*)$, then
     $\tilde{\lambda}^0(u_*)\in\partial^+\iota(M_t)$ and (after possibly reducing $\epsilon$, $|v_*|$ and $\tilde{U}$), $(\tilde{U},\tilde{\varphi})$ can be chosen to be a future boundary chart with $\iota([u_*-\epsilon,u_*+\epsilon]\times[v_*,0])\subset\tilde{U}_{\leq}$.
\end{remark}

\begin{proof}
It follows from \eqref{eqn:monotonicityatboundary} that the curve $\tilde{\lambda}^0|_{[u_* - \epsilon, u_*]}$ must lie in $\{\tilde{x}_0 \leq 0\} \cap \{\tilde{x}_1 \geq 0\}$. Furthermore, by assumption we have $\tilde{\lambda}^0(u_* - \epsilon) \in \{\tilde{x}_0 \leq 0\} \cap \{\tilde{x}_1 \geq 0\} \setminus (0,0)$. Combined with the monotonicity properties \eqref{eqn:monotonicityconventions} and the gradient bound \eqref{eqn:gradientbound1+1}, this implies that a segment of the line $\tilde{x}_0=2\tilde{x}_1$ directly below $\tilde{p}_{u_*}$ lies in the region bound by the curves $\tilde{\gamma}^{u_*}|_{[v_*,0)}$, $\tilde{\lambda}^{v_*}|_{[u_* - \epsilon, u_*]}$, $\tilde{\gamma}^{u_* - \epsilon}|_{[v_*,0)}$, and $\tilde{\lambda}^0|_{[u_* - \epsilon, u_*]}$ (see Figure \ref{fig:futureboundary}). By Proposition \ref{prop:boundarystructure} (a)(iii), this region is equal to $\iota\big((u_* - \epsilon, u_*) \times (v_*,0)\big)$. Since the curve $\tilde{x}_0=2\tilde{x}_1$ is timelike (by \eqref{eqn:gradientbound1+1}) we conclude that $\tilde{p}_{u_*}\in\partial^+\iota(M_t)$.

\begin{figure}[h]
\begin{minipage}{0.48\textwidth}
    \centering
    \includegraphics[scale=0.25]{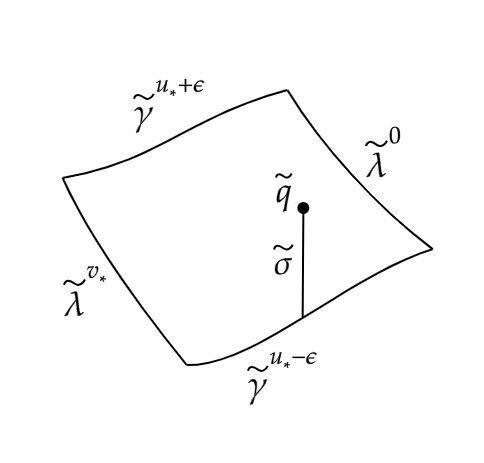}
    \caption{Illustrating the proof of Proposition \ref{prop:boundarystructure} (a)(iii). The figure here shows the case where  $\tilde{\lambda}^0(u_*-\epsilon)$ and $ \tilde{\lambda}^0(u_*+\epsilon)$ are distinct, although in general this may not be the case (see Theorem \ref{thm:C0structure}).\\
    \medskip\medskip}
    \label{fig:interior}
     \end{minipage}
    \minipage{0.04\textwidth}
  \includegraphics[scale=0.05]{blank.png}
 \endminipage
    \begin{minipage}{0.48\textwidth}
    \centering
    \includegraphics[scale=0.23]{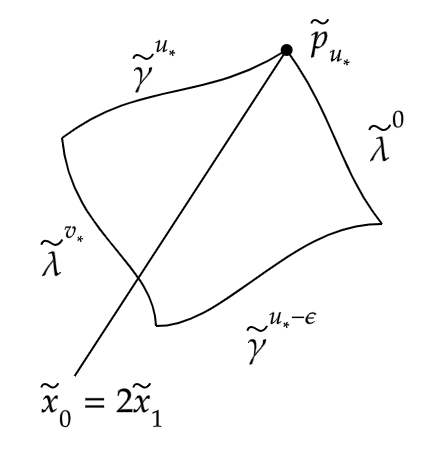}
    \caption{The interior of the region bound by the curves $\tilde{\gamma}^{u_*-\epsilon}|_{[v_*,0)}$, $\tilde{\gamma}^{u_*}|_{[v_*,0)}$, $\tilde{\lambda}^{v_*}|_{[u_* - \epsilon, u_*]}$,  and $\tilde{\lambda}^0|_{[u_* - \epsilon, u_*]}$  is equal to $\iota\big((u_* - \epsilon, u_*) \times (v_*,0)\big)$ and contains a segment of the line $\tilde{x}_0=2\tilde{x}_1$ directly below $\tilde{p}_{u_*}$. \\}
    \label{fig:futureboundary}
       \end{minipage}
\end{figure}
By Proposition \ref{prop:futureboundarychart}, we can therefore choose $\tilde{U}$, $\epsilon$ and $|v_*|$ smaller if necessary so that $(\tilde{U},\tilde{\varphi})$ is a future boundary chart centred at $\tilde{p}_{u_*}$ and the results of part (a)(i) still hold. Since $\text{graph}(f)$ is achronal (see Proposition \ref{prop:futureboundarychart}), it follows that points on the line $\tilde{x}_0=2\tilde{x}_1$ just below $\tilde{p}_{u_*}$ lie in $\tilde{U}_<$. The connectedness of $\iota([u_*-\epsilon,u_*+\epsilon]\times[v_*,0))$ then implies that $\iota([u_*-\epsilon,u_*+\epsilon]\times[v_*,0))\subset\tilde{U}_{<}$. 
We conclude that $\iota([u_*-\epsilon,u_*+\epsilon]\times[v_*,0])\subset\tilde{U}_{\leq}$.

\end{proof}

\section{The $C^0$-Structure of Extensions in 1+1-Dimensions}\label{The $C^0$ Structure of the Extension in 1+1-Dimensions}
\subsection{Dichotomy of the $C^0$-structure of local extensions: the corner and reference extensions }\label{Rigidity of the $C^0$-structure}

Having derived some basic structural results for the boundary in Proposition \ref{prop:boundarystructure}, we now investigate the $C^0$-structure of $C^0$-extensions. 
A natural first question to ask is whether this structure is fixed by the requirement that the metric extends continuously. In Example \ref{example:cornerextension} we show that this is not the case for $(M_t,g_t)$ by defining a $C^0$-extension across $\{v=0\}$ which is $C^0$-inequivalent to the reference extension $\overline{M_t}$. 
%This example describes one possible $C^0$ extension of this Lorentzian manifold. In this extension, all constant $u$ curves in $(M,g)$ are extended to distinct points on $\partial\iota(M)$. In the follow example we construct a $C^0$ extension in which all constant $u$ curves extend through the same point in $\partial\iota(M)$.

\begin{example}[The corner extension]\label{example:cornerextension}
Under the additional assumption that $\Omega$ extends as a continuous strictly positive function to $[-1,1] \times (-1,0]$,\footnote{I.e. $\Omega$ also extends continuously to $\{u = \pm 1\}$ as a strictly positive function.} we will construct a $C^0$-extension of $(M_t,g_t)$, denoted $\iota_{corner,\beta}$,  such that $\lim_{v\rightarrow0}(\iota_{corner,\beta}\circ\gamma^u)(v)$ exists for all $u\in(-1,1)$ and is independent of $u$. We proceed in two steps.

\textbf{Step 1:} Let $g'_0=-du\otimes dv-dv\otimes du$ be a Lorentzian metric on $M_t$. We begin by constructing a $C^0$-extension $\iota_{1,\beta} : M_t \hookrightarrow \tilde{M}_t$ of $(M_t,g'_0)$ such that $\lim_{v\rightarrow0}(\iota_{1,\beta}\circ\gamma^u)(v)$ exists for all $u\in(-1,1)$ and is independent of $u$ (see Figure \ref{fig:corner}).
\begin{figure}[h]
    \centering
    \includegraphics[scale=0.25]{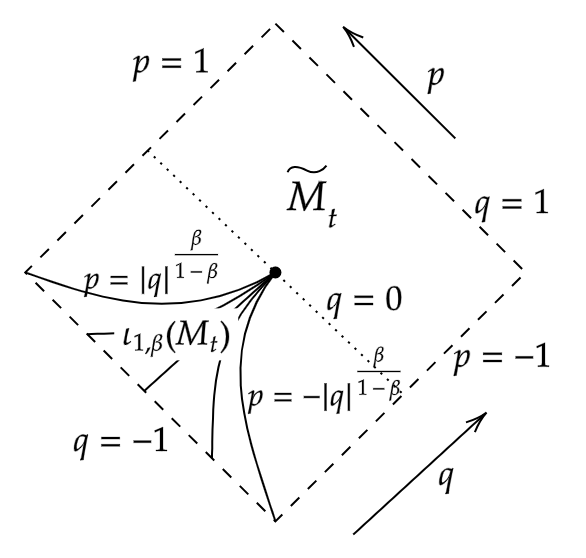}
    \caption{Figure showing $\tilde{M}_t$. The point $\lim_{v\rightarrow0}(\iota_{1,\beta}\circ\gamma^u)(v)$ exists for all $u\in(-1,1)$ and is independent of $u$.}
    \label{fig:corner}
\end{figure}

For $\beta\in\left(\frac{1}{2},1\right)$, let $(\tilde{M}_t,\tilde{g}_\beta)$ be the 1+1-dimensional Lorentzian manifold with global coordinates $(p,q)\in(-1,1)\times(-1,1)$, where $\tilde{g}_\beta$ is given in these coordinates by
    \begin{equation}
        \tilde{g}_\beta=\begin{cases}
            -\frac{1}{(1-\beta)}\left(dp\otimes dq + dq\otimes dp\right)+ \frac{2\beta p}{(1-\beta)^2q}dq\otimes dq & |p| < |q|^{\frac{\beta}{1-\beta}}, \text{ }-1<q<0;\\
            -\frac{1}{(1-\beta)}\left(dp\otimes dq + dq\otimes dp\right) +\frac{2\beta }{(1-\beta)^2}|q|^{\frac{2\beta-1}{1-\beta}}dq\otimes dq & \text{ }p\leq -|q|^{\frac{\beta}{1-\beta}}, -1<q<0;\\
            -\frac{1}{(1-\beta)}\left(dp\otimes dq + dq\otimes dp\right) - \frac{2\beta }{(1-\beta)^2}|q|^{\frac{2\beta-1}{1-\beta}}dq\otimes dq & \text{ }p\geq|q|^{\frac{\beta}{1-\beta}}, -1<q<0;\\
            -\frac{1}{(1-\beta)}\left(dp\otimes dq + dq\otimes dp\right)  & \text{ }0\leq q<1.            
        \end{cases}
    \end{equation}
    The restriction on $\beta$ ensures that $\tilde{g}_\beta$ is continuous. 

    Define the embedding $\iota_{1,\beta}:M_t\hookrightarrow\tilde{M}_t$ in terms of coordinates $(u,v)\in(-1,1)\times(-1,0)$ by
    \begin{equation}
        \begin{split}
            p(u,v)&=u|v|^\beta\\
            q(u,v)&=-|v|^{1-\beta}.
        \end{split}
    \end{equation}
    Inverting this, we have
    \begin{equation}
        \begin{split}
            u(p,q)&=p|q|^{-\frac{\beta}{1-\beta}}\\
            v(p,q)&=-|q|^{\frac{1}{1-\beta}}
        \end{split}
    \end{equation}
from which we see that $\iota_{1,\beta}$ isometrically embeds $(M_t,g'_0)$ into $(\tilde{M}_t,\tilde{g}_\beta)$, with 
\begin{equation}
    \begin{split}
        \iota_{1,\beta}(M)= \{(p,q)\in(-1,1)\times(-1,0):|p| < |q|^{\frac{\beta}{1-\beta}}\}. 
    \end{split}
\end{equation}
Note that $\lim_{v\rightarrow0}(\iota_{1,\beta}\circ\gamma^u)(v)=(0,0)\in\partial\iota_{1,\beta}(M_t)$ for all $u\in(-1,1)$. 

\textbf{Step 2:} We now generalise this construction to obtain a $C^0$-extension of $(M_t,g_t)$, denoted $\iota_{corner,\beta}$, such that $\lim_{v\rightarrow0}(\iota_{corner,\beta}\circ\gamma^u)(v)$ exists for all $u\in(-1,1)$ and is independent of $u$.

Define
\begin{equation}
    \begin{split}
        \iota_2:M_t&\hookrightarrow M_t\\
        (u,v)&\mapsto \left(\frac{2}{C}\int_{-1}^{u}\Omega(s,0)^2ds-1,v\right),
    \end{split}
\end{equation}
where $C=\int_{-1}^1\Omega(s,0)^2ds$.\footnote{Note that this map is not smooth in general.} This map defines an isometric embedding of $(M_t,g_t)$ into $(M_t,g_t')$, where $g_t'=-\Omega'(u,v)^2(du\otimes dv+dv\otimes du)$ and $\Omega'(u,v)=\sqrt{\frac{C}{2}}\frac{\Omega(\iota^{-1}_2(u,v))}{\Omega(\iota^{-1}_2(u,0))}$ (and we have used the strict positivity of $\Omega(u,v)^2$ to invert $\iota_2$). 

The conformal factor $\Omega'(u,v)^2$ extends continuously to $v=0$ as a strictly positive function with $\Omega'(u,0)^2=\frac{C}{2}$ independent of $u$. Hence we can define a continuous, strictly positive function $\tilde{\Omega}:\tilde{M}_t\rightarrow \R_{>0}$ in terms of coordinates $(p,q)$ by
\begin{equation}
    \begin{split}
       \tilde{\Omega}(p,q) = \begin{cases}
            \Omega'(p|q|^{-\frac{\beta}{1-\beta}},-|q|^{\frac{1}{1-\beta}}) & |p| < |q|^{\frac{\beta}{1-\beta}}, \text{ }-1<q<0;\\
            \Omega'(-1,-|q|^{\frac{1}{1-\beta}}) & \text{ }p\leq -|q|^{\frac{\beta}{1-\beta}}, -1<q<0;\\
           \Omega'(1,-|q|^{\frac{1}{1-\beta}}) & \text{ }p\geq|q|^{\frac{\beta}{1-\beta}}, -1<q<0;\\
          \Omega'(0,0) & \text{ }0\leq q<1.     
       \end{cases}
    \end{split}
\end{equation}
Note that the continuity of $\tilde{\Omega}$ relies on the extra assumption that $\Omega$ extends continuously to $\{u=\pm1\}$.

This allows us to use $\iota_{1,\beta}$ from Step 1 to construct a $C^0$-extension across $v=0$ of $(M_t,g_t')$ as follows. For $\beta\in(\frac{1}{2},1)$ and $(p,q)\in(-1,1)\times(-1,1)$, define $\tilde{g}'_\beta(p,q):= \tilde{\Omega}^2(p,q)\tilde{g}_\beta(p,q)$. Then $\iota_{1,\beta}$
defines a $C^0$-extension across $v=0$ of $(M_t,g_t')$ into $(\tilde{M}_t,\tilde{g}'_\beta)$ such that $\lim_{v\rightarrow0}(\iota_{1,\beta}\circ\gamma^u)(v)=(0,0)$ for all $u\in(-1,1)$.

We conclude that $\iota_{corner,\beta}:=\iota_{1,\beta}\circ\iota_2$ defines a $C^0$-extension across $v=0$ of $(M_t,g_t)$ into $(\tilde{M}_t,\tilde{g}'_\beta)$ such that $\lim_{v\rightarrow0}(\iota_{corner,\beta}\circ\gamma^u)(v)=(0,0)$ for all $u\in(-1,1)$. It follows that for any $u_*\in(-1,1)$ and any $\epsilon>0$ sufficiently small, the continuous extension of $\iota_{corner,\beta}|_{(u_* - \epsilon, u_* + \epsilon) \times (-1,0)}$ to $\iota_{corner,\beta}:(u_*-\epsilon,u_*+\epsilon)\times(-1,0]\rightarrow\tilde{M}_t$ is not a homeomorphism. We conclude that $\iota_{corner,\beta}$ is not locally $C^0$-equivalent to $\overline{M_t}$ at $(u,0)\in\overline{M_t}$, for any $u \in(-1,1)$. 

Observe that for $q\in(-1,0)$ and $\epsilon'\in(0,|q|^{\frac{\beta}{1-\beta}})$ we have
\begin{equation}
    \begin{split}
       \frac{(\tilde{g}_\beta)_{qq}(|q|^{\frac{\beta}{1-\beta}},q) -(\tilde{g}_\beta)_{qq} (|q|^{\frac{\beta}{1-\beta}}-\epsilon',q)}{\epsilon'}&=\frac{2\beta}{(1-\beta)^2q}
    \end{split}
\end{equation}
which is unbounded as $q\rightarrow0$. It follows that the metric $\tilde{g}_\beta$ (and similarly $\tilde{g}_\beta'$) is not locally Lipschitz. Furthermore, for each $u\in(-1,1)$ we have $\lim_{v\rightarrow0}(\iota_{corner,\beta}\circ\gamma^u)(v)=(0,0)\notin\partial^+\iota_{corner,\beta}(M_t)$. This cannot happen for locally Lipschitz extensions \cite[Proposition 3.5]{sbierski2022uniqueness}). Furthermore, in the language of \cite[Proposition 1.10 (vi)]{ChruscielGrant}, the past bubble set of the point $(0,0)\in\tilde{M}_t$ is given by $\iota_{1,\beta}(M_t)$, so in particular is open and non-empty. This cannot happen if the metric is locally Lipschitz \cite[Corollary 1.17]{ChruscielGrant}.

%This is in contrast to the reference extension $\iota_{ref}$, where the curves $\gamma^u$ extend through distinct points in $\partial\iota_{ref}(M_t)$. Consequently, these two extensions are not $C^0$-equivalent. More precisely, the identification map, given in terms of coordinates $(u',v')\in(-1,1)\times(-1,0)$ on $\iota_{ref}(M_t)$ by
%\begin{equation}
%\begin{split}
%     id:\iota_{ref}(M_t)&\rightarrow\iota_{1,\beta}(M_t)\\
 %    (u',v')&\mapsto(p(u',v'),q(u',v'))=(u'|v'|^\beta,-|v'|^{1-\beta}).
%\end{split}    
%\end{equation}
%does not extend continuously to a homeomorphism at the boundary $\partial\iota_{ref}(M_t)$ (which corresponds to $v'=0$, $u'\in(-1,1)$). 

\end{example}

We have so far seen two different local $C^0$-structures for $C^0$-extensions across $v=0$ of $(M_t,g_t)$. These are given by $\overline{M_t}$ from Section \ref{toymodel} and $\iota_{corner,\beta}$ from Example \ref{example:cornerextension}. In Theorem \ref{thm:C0structure} we show that these are the only two possible local $C^0$-structures for such extensions. The proof relies on Proposition \ref{prop:nomixedcornerextension}, which implies a rigidity property on the $C^0$-structure of the corner extension: if a $C^0$-extension of $(M_t,g_t)$ across $v=0$ locally exhibits the $C^0$-structure of the corner extension, then it must do so globally. This relies heavily on our assumption in Section \ref{toymodel} that the conformal factor $\Omega$ extends continuously as a positive function to $\overline{M_t}$, i.e. the existence of our reference extension.

\begin{prop}\label{prop:nomixedcornerextension}
%I'm going to comment out large sections of the previous version of this proposition (and edit some others)

Consider the setting of Proposition \ref{prop:boundarystructure} $(a)$ with near Minkowski neighbourhood $(\tilde{U},\tilde{\varphi}=(\tilde{x}_0,\tilde{x}_1))$ centred on $\tilde{p}$ and $v_*<0 < \epsilon$. Suppose there exist $u_-,u_+\in[u_* - \epsilon,u_* + \epsilon]$ with $u_-<u_+$ such that $\lim_{v \to 0}\tilde{\gamma}^{u_-}(v) = \lim_{v \to 0}\tilde{\gamma}^{u_+}(v) = \tilde{p}$. 

%(a) Then $\iota$ extends continuously to $\iota :(-1,1) \times (-1,0] \hookrightarrow \tilde{M}$ with $\iota(u,0) = \tilde{p}$ for all $u \in (-1,1)$. 

%(b) Reduce $|v_*|$ if necessary so that the straight line of gradient $-\frac{1}{2}$ (in the $\tilde{x}_1-\tilde{x}_0$ plane) through $\iota(u_*,v_*)$ intersects $\tilde{x}_0=\frac{1}{2}\tilde{x}_1$ at some point in $\tilde{U}$ and the straight line of gradient $-2$ through $\iota(u_*,v_*)$ intersects $\tilde{x}_0=2\tilde{x}_1$ at some point in $\tilde{U}$. Define $\tilde{B}\subset\tilde{U}$ to be the closure of the region bound by the curves $\tilde{x}_0=2\tilde{x}_1$, $\tilde{x}_0=\frac{1}{2}\tilde{x}_1$ and these two straight lines (see Figure \ref{fig:Btilde}). Then $\iota ((-1,1) \times [v_*,0]) \subset \tilde{B}$ and (10) holds on $\iota ((-1,1) \times [v_*,0))$.

Reducing $|v_*|$ if necessary, let $\tilde{B}\subset\tilde{U}$ be the closure of the region bound by the curves $\tilde{x}_0=2\tilde{x}_1$, $\tilde{x}_0=\frac{1}{2}\tilde{x}_1$ and the straight lines through $\iota(u_*,v_*)$ with gradient $-\frac{1}{2}$ and $-2$ (see Figure \ref{fig:Btilde}). Then $\iota$ extends continuously to $\iota :(-1,1) \times [v_*,0] \hookrightarrow \tilde{B}$ with $\iota(u,0) = \tilde{p}$ for all $u \in (-1,1)$. Moreover, \eqref{eqn:monotonicityconventions} holds on $\iota\big((-1,1) \times [v_*,0) \big) \subset \tilde{B}$.
%and \eqref{eqn:monotonicityconventions} holds on $\iota ((-1,1) \times [v_*,0))$.

\end{prop}

\begin{figure}[h]
    \centering
 %   \begin{minipage}{0.48\textwidth}
 %           \centering
\includegraphics[scale=0.25]{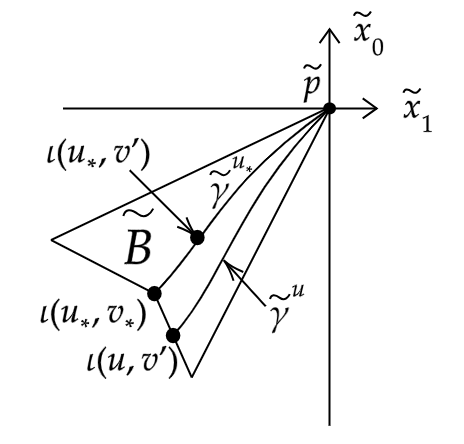}
   % \caption{\sout{For $u_*\in(-1,1)$, let $v_*<0$ be sufficiently small that the straight lines with gradient $-\frac{1}{2}$ and $-2$ through $\iota(u_*,v_*)$ and the lines $\tilde{x}_0=\frac{1}{2}\tilde{x}_1$ and $\tilde{x}_0=2\tilde{x}_1$ close off a region $\tilde{B}\subset\tilde{U}$}}
    \caption{Proof that $\tilde{\gamma}^u\vert_{[v',0)}\subset\tilde{B}$ for $u\in[u_-,u_+]$ (the figure shown corresponds to the case $u<u_*$) used to show the non-emptiness of $S^-_{\tilde{p}}$.}
    \label{fig:Btilde}
 %   \end{minipage}%
%     \minipage{0.04\textwidth}
%    \includegraphics[scale=0.05]{blank.png}
% \endminipage
 %   \begin{minipage}{0.48\textwidth}
  %      \centering
 %\includegraphics[scale=0.25]{Btildeproof.png}
 %   \caption{Proof of the fact that $\tilde{\lambda}^{v_0}\vert_{(-1,u_*)}\subset\tilde{B}$ for $v_0\in[v_*,0)$.\red{}\\ \\}
 %   \label{fig:Btildeproof}
 %   \end{minipage}
\end{figure}

Note that although this proposition only makes an assumption about the local behaviour of the extension near $\tilde{p}$, the result obtained places a restriction on its global (i.e. for all $u \in (-1,1)$) structure. 

\begin{proof} (a) Define 
%\begin{equation}
   % \begin{split}
    %    S^+_{\tilde{p}}:=\{u\in(u_-,1):\text{ }& \iota \text{ extends continuously to }\iota:[u_-,u]\times[v_*,0]\rightarrow\tilde{U} \text{ for some }v_*\in(-1,0)\text{ with }\\
     %   &\iota\big([u_-,u],0)=\{\tilde{p}\} \text{ and }\eqref{eqn:monotonicityconventions}\text{ holding on }\iota([u_-,u]\times[v_*,0)\big)\}\\
      %  S^-_{\tilde{p}}:=\{u\in(-1,u_+):\text{ }& \iota \text{ extends continuously to }\iota:[u,u_+]\times[v_*,0]\rightarrow\tilde{U} \text{ for some }v_*\in(-1,0)\text{ with }\\
       % &\iota\big([u,u_+],0)=\{\tilde{p}\} \text{ and }\eqref{eqn:monotonicityconventions}\text{ holding on }\iota([u,u_+]\times[v_*,0)\big)\}.
   % \end{split}
%\end{equation}
\begin{equation}
    \begin{split}
        S^+_{\tilde{p}}&:=\{u\in(u_-,1):\text{ } \iota \text{ extends continuously to }\iota:[u_-,u]\times[v_*,0]\rightarrow\tilde{B}\}\\
        S^-_{\tilde{p}}&:=\{u\in(-1,u_+):\text{ } \iota \text{ extends continuously to }\iota:[u,u_+]\times[v_*,0]\rightarrow\tilde{B} \}.
    \end{split}
\end{equation}

We will prove that $S^-_{\tilde{p}}=(-1,u_+)$ by showing that this set is non-empty as well as open and closed in $(-1,u_+)$. A similar argument shows that $S_{\tilde{p}}^+=(u_-,1)$. 
By Proposition \ref{prop:boundarystructure} (a)(i) it will then follow that \eqref{eqn:monotonicityconventions} holds on $\iota\big((-1,1) \times [v_*,0) \big) \subset \tilde{B}$. As in the proof of Proposition \ref{prop:boundarystructure} (a)(iii), we infer from this that $\tilde{x}_0(u,0)$ is non-decreasing in $u$ and $\tilde{x}_1(u,0)$ is non-increasing. Since $\iota(u,0)\in\tilde{B}$, we have $\tilde{x}_0(u,0),\tilde{x}_1(u,0)\leq0$ and moreover $\tilde{x}_0(u,0)=0\iff\tilde{x}_1(u,0)=0$. For $u\in(-1,u_+)$, it follows from the fact that $\tilde{x}_1(u,0)$ is non-increasing that $\tilde{x}_1(u,0)=0$ and hence $\iota(u,0)=\tilde{p}$. Similarly, for $u\in(u_-,1)$, since $\tilde{x}_0(u,0)$ is non-decreasing we have $\tilde{x}_0(u,0)=0$  and hence $\iota(u,0)=\tilde{p}$. We conclude that
$\iota(u,0)=\tilde{p}$ for all $u\in(-1,1)$ which concludes the proof. 

\textbf{Step 1 (non-emptiness):} 
%We prove the following useful result. 
%\textbf{Claim: }Suppose $u_1\in(-1,1)$ and $\iota$ extends continuously to $\iota:[u_1-\epsilon',u_1+\epsilon']\times[v_*,0)\hookrightarrow\tilde{U}$ for some $\epsilon'>0$ with \eqref{eqn:monotonicityconventions} holding on $\iota([u_1-\epsilon',u_1+\epsilon']\times[v_*,0))$ and $\iota(u,0)=\iota(u_1-\epsilon',0)=\tilde{p}$. Then $\iota([u_1-\epsilon',u_1]\times[v_*,0]\subset\tilde{B}$.
%
%\textbf{Proof of claim:} 
It follows from \eqref{eqn:monotonicityconventions} that $\tilde{x}_0(u,0)$ and $\tilde{x}_1(u,0)$ are monotonic functions on $[u_-,u_+]$ and hence $\iota(u,0)=\tilde{p}$ for all $u\in[u_-,u_+]$. Fix $u\in[u_-,u_+]$.  From \eqref{eqn:gradientbound1+1} and \eqref{eqn:monotonicityconventions} we have $\tilde{\gamma}^u\vert_{[v,0)}\subset\tilde{B}$ for $v<0$ sufficiently small. Suppose $\tilde{\gamma}^u\vert_{[v',0)}\subset\tilde{B}$ and $\tilde{\gamma}^u(v')\in\partial\tilde{B}$ for some $v'\in(v_*,0)$. Assume that $u<u_*$ (the case where $u>u_*$ follows similarly). Then \eqref{eqn:gradientbound1+1} and \eqref{eqn:monotonicityconventions} imply that $\tilde{\gamma}^u\vert_{[v',0)}$ lies below $\tilde{\gamma}^{u_*}\vert_{[v_*,0)}$ in $\tilde{B}$ and $\tilde{\gamma}^{u}(v')$ lies on the straight line of gradient $-2$ through $\iota(u_*,v_*)$ (see Figure \ref{fig:Btilde}). Consider the curve $\tilde{\lambda}^{v'}\vert_{[u,u_*]}$ originating at $\iota(u_*,v')$. This curve must exit $\tilde{B}$, however it follows from \eqref{eqn:gradientbound1+1} and \eqref{eqn:monotonicityconventions} that in order to do so it must first intersect $\tilde{\gamma}^u\vert_{(v',0)}$. This contradicts the injectivity of $\iota$, so we conclude that $\tilde{\gamma}^u\vert_{[v_*,0)}\subset\tilde{B}$ and hence $\iota([u_-,u_+]\times[v_*,0])\subset\tilde{B}$.
%The lightcone bounds \eqref{eqn:gradientbound1+1}, the monotonicity \eqref{eqn:monotonicityconventions} and the injectivity of $\iota$ imply that $\iota([u_*-\epsilon,u_*+\epsilon]\times[v_*,0))\subset\tilde{B}$.\todo{say more?} It follows from \eqref{eqn:monotonicityconventions} and the compactness of $\tilde{B}$ that, for all $u\in[u_*-\epsilon,u_++\epsilon]$, $\iota(u,0)=\lim_{v\rightarrow0}\iota(u,v)$ exists and lies in $\tilde{B}$. 
It follows that $u_-\in S_{\tilde{p}}^-$, so this set is non-empty. 

\textbf{Step 2 (closedness):} %We now show that $S_{\tilde{p}}^-$ is closed in $(-1,u_+)$. 
Let $(u'_n)_{n=1}^\infty$ be a sequence in $S_{\tilde{p}}^-$ converging to $u_1\in(-1,u_+)$. By restricting to a subsequence if necessary, we assume that $(u'_n)_{n=1}^\infty$ is strictly monotonic. It is clear that $[u_n',u_+)\subset S_{\tilde{p}}^-$, so if the sequence is strictly increasing then it follows immediately that $u_1\in S_{\tilde{p}}^-$. 

%\begin{figure}[h]
%    \centering
%    \includegraphics[scale=0.25]{closedness.png}
%    \caption{$\tilde{B}$ is the closure of the region bound by %the curves $\tilde{\gamma}^{u_+}\vert_{[v_*,0)}$, $\tilde{x}_1=0$ and the curve with gradient -2 in the $\tilde{x}_0-\tilde{x}_1$ plane through $\iota(u_+,v_*)$. For all $n\in\mathbbm{N}$, the curves $\tilde{\gamma}^{u_n'}\vert_{[v_*,0]},\tilde{\lambda}^{v_*}\vert_{[u_n',u_+]}$ lie in $\tilde{B}$ and hence $\iota([u_n',u_+]\times[v_*,0])\subset\tilde{B}$. \red{See JS comment}}
%    \label{fig:closedness}
%\end{figure}

%Reducing $|v_*|$ if necessary, let $\tilde{B}\subset\tilde{U}$ be the closure of the subset of $\tilde{U}$ bound by the curves $\tilde{\gamma}^{u_+}\vert_{[v_*,0)}$, $\tilde{x}_1=0$ and the curve with gradient -2 in the $\tilde{x}_0-\tilde{x}_1$ plane through $\iota(u_+,v_*)$. This is shown in Figure \ref{fig:closedness}. 
%It follows from \eqref{eqn:gradientbound1+1}, \eqref{eqn:monotonicityconventions} and the injectivity of $\iota$ that $\tilde{\gamma}^{u_n'}\vert_{[v_*,0]},\tilde{\lambda}^{v_*}\vert_{[u_n',u_+]}\subset\tilde{B}$ for all $n\in\mathbbm{N}$. Hence Proposition \ref{prop:boundarystructure} (a)(iii) implies that $\iota([u_n',u_+]\times[v_*,0])\subset\tilde{B}$\red{}. 
Now suppose the sequence is strictly decreasing.
%For any $n\in\mathbbm{N}$, $\iota$ extends continuously to $\iota([u_n',u_+]\times[v_*,0])\subset\tilde{B}$ with $\eqref{eqn:monotonicityconventions}$ holding on $\iota([u_n',u_+]\times[v_*,0))$. 
It follows that $\iota$ extends continuously to $\iota((u_1,u_+]\times[v_*,0])\subset\tilde{B}$ and Proposition \ref{prop:boundarystructure} (a)(i) implies $\eqref{eqn:monotonicityconventions}$ holds on $\iota((u_1,u_+]\times[v_*,0))$.
Since $\tilde{B}$ is compact, it follows from the monotonicity \eqref{eqn:monotonicityconventions} that $\tilde{\gamma}^{u_1}(v)=\lim_{n\rightarrow\infty}\tilde{\gamma}^{u_n'}(v)\in\tilde{B}$ for all $v\in[v_*,0)$. The connectedness of $\iota([u_1,u_+]\times[v_*,0))$ along with Proposition \ref{prop:boundarystructure} (a)(i) imply that \eqref{eqn:monotonicityconventions} holds along $\tilde{\gamma}^{u_1}\vert_{[v_*,0)}$ and hence we can take a limit $\tilde{\gamma}^{u_1}(0):=\lim_{v\rightarrow0}\tilde{\gamma}^{u_1}(v)\in\tilde{B}$. We conclude that $\iota([u_1,u_+]\times[v_*,0])\subset\tilde{B}$. The continuity of $\iota:[u_1,u_+]\times[v_*,0]\rightarrow\tilde{B}$ then follows from Proposition \ref{prop:boundarystructure} (a)(i).
%there exists $\epsilon'>0$ such that $\iota$ extends continuously to $\iota:[u_1-\epsilon',u_1+\epsilon']\times(-1,0]\rightarrow\red{\tilde{M}}$.
%with $\eqref{eqn:monotonicityconventions}$ holding on $\iota([u_*-\epsilon',u_*+\epsilon']\times[v_*,0))$. 
%For $n$ sufficiently large we have $u_n'\in(u_1,u_1+\epsilon']\cap S_{\tilde{p}}^-$ so we conclude that $\iota$ extends continuously to $\iota:[u_1,u_+]\times[v_*,0]\rightarrow\red{\tilde{B}}$. 
%with $\eqref{eqn:monotonicityconventions}$ holding on $\iota([u_*,u_+]\times[v_*,0))$. 
%Since $\iota(u_n',0)=\tilde{p}$, the continuity of the extension of $\iota$ implies that $\iota(u_*,0)=\tilde{p}$ and hence $u_*\in S_{\tilde{p}}^-$.

\textbf{Step 3 (openness):} %We claim that $S^-_{\tilde{p}}$ is open in $(-1,u_+)$. 
Let $u_1\in S_{\tilde{p}}^-$. We will show that $(u_1-\epsilon',u_1+\epsilon')\subset S_{\tilde{p}}^-$ for some $\epsilon'>0$. By Proposition \ref{prop:boundarystructure} (a)(i), there exists $-1<v_*'<0<\epsilon'$ such that $\iota$ extends continuously to $\iota:[u_1-\epsilon',u_1+\epsilon']\times[v_*',0]\rightarrow\tilde{U}$ with \eqref{eqn:monotonicityconventions} holding on $\iota([u_1-\epsilon',u_1+\epsilon']\times[v_*',0))$.
%We will show that $\iota(u,0)=\tilde{p}$ for all $u\in(u_1-\epsilon',u_1+\epsilon')$. It then follows from Claim 1 that $\iota((u_1-\epsilon',u_1+\epsilon')\times[v_*,0))\subset\tilde{B}$\todo{Same as previous comment - say a bit more?}.
Reducing $\epsilon'$ if necessary, it is clear that $[u_1,u_1+\epsilon']\subset S_{\tilde{p}}^-$. If $\iota(u',0)=\tilde{p}$ for some $u'\in[u_1-\epsilon',u_1)$ then it follows from the argument in Step 1 (replacing $u_-$ and $u_+$ by $u'$ and $u_1$ respectively) that $[u',u_1]\in S_{\tilde{p}}^-$ and we are done. Suppose this is not the case, i.e. suppose $\iota(u,0)\neq\tilde{p}$ for all $u\in[u_1-\epsilon',u_1)$. We will show this leads to a contradiction. 

First we recall that the spacetime volume of $A\subset M_t$ is defined by
\begin{equation}
        \text{vol}_{g_t} (A):=\int_A\sqrt{-\text{det}g_t} \,dudv 
                = \int_{A}\Omega(u,v)^2\,dudv\;.
       \end{equation}
For $v\in[v_*',0)$, we consider the following regions in $M_t$: 
\begin{equation}\label{eqn:AplusAminus}
    \begin{split}
A_+(v)&:=(u_1,u_1 + \epsilon')\times(v,0)\subset M_t\\
A_-(v)&:=(u_1 - \epsilon',u_1)\times(v,0)\subset M_t.
    \end{split}
\end{equation}
Since $\Omega(u,v)$ extends continuously to a strictly positive function on $(-1,1)\times(-1,0]$, and since $[u_1 - \epsilon', u_1 + \epsilon'] \times [v_*',0]$ is compactly contained in this region, there exist constants $0<C_-\leq C_+$ such that $C_-\leq\Omega(u,v)\leq C_+$ on $[u_1 - \epsilon', u_1 + \epsilon'] \times [v_*',0]$. Thus we can estimate
\begin{align*}
    &\text{vol}_{g_t} (A_+(v))\geq C_-^2\epsilon' |v|\\
    &\text{vol}_{g_t} (A_-(v)) \leq C_+^2\epsilon' |v| \;.
\end{align*}
Since $\iota$ is an isometry, it follows that
\begin{equation}
     \frac{\text{vol}_{\tilde{g}}(\iota(A_+(v)))}{\text{vol}_{\tilde{g}}(\iota(A_-(v)))} = \frac{\text{vol}_{g_t}(A_+(v))}{\text{vol}_{g_t}(A_-(v))}\geq\frac{C_-^2}{C_+^2}>0\;,\label{eqn:volumeinequality}
\end{equation}
where $\text{vol}_{\tilde{g}}$ denotes the spacetime volume on $(\tilde{U}, \tilde{g})$.
We now estimate the respective volumes using the geometry of $(\tilde{U}, \tilde{g})$ together with the assumption that $\iota(u,0) \neq \tilde{p}$ for $u \in [u_1 - \epsilon', u_1)$. This will contradict \eqref{eqn:volumeinequality} since the ratio of these volumes can be made arbitrarily small for, in particular, $|v|$ small enough.
We introduce new coordinates $p:=\frac{1}{\sqrt{2}}(\tilde{x}_0-\tilde{x}_1)$ and $q:=\frac{1}{\sqrt{2}}(\tilde{x}_0+\tilde{x}_1)$ on $\tilde{U}$. The spacetime volume of $\tilde{A}\subset\tilde{U}$ is computed in these coordinates by
$$\text{vol}_{\tilde{g}}(\Tilde{A}):= \int_{\Tilde{A}}\sqrt{-\text{det}\tilde{g}}\,dpdq.$$

In these coordinates, the metric $\Tilde{g}$ at $\tilde{p}$ is given by 
\begin{equation}\label{eqn:metricindoublenull}
    \Tilde{g}|_{\tilde{p}}=-dp\otimes dq-dq\otimes dp
\end{equation}
and hence $\partial_p$ and $\partial_q$ are null and linearly independent at $\tilde{p}$. 

We have 
\begin{equation}
      \frac{\partial\tilde{x}_1}{\partial p}<0<\frac{\partial\tilde{x}_0}{\partial p},\frac{\partial\tilde{x}_0}{\partial q},\frac{\partial\tilde{x}_1}{\partial q}
\end{equation} 
and hence comparing with the monotonicity properties \eqref{eqn:monotonicityconventions} we see that $\frac{\iota_*\partial_u}{\lVert\iota_*\partial_u\rVert_{\mathbbm{R}^2}}$ and $\frac{\iota_*\partial_v}{\lVert\iota_*\partial_v\rVert_{\mathbbm{R}^2}}$, defined on $\iota([u_1-\epsilon',u_1+\epsilon']\times[v_*',0))$, extend continuously to $\partial_p$ and $\partial_q$ respectively at $\tilde{p}$. Here, $\lVert.\rVert_{\R^2}$ denotes the norm defined with respect to the Riemannian metric $\delta=d\tilde{x}_0\otimes d\tilde{x}_0+d\tilde{x}_1\otimes d\tilde{x}_1$ on $\tilde{U}$ (note that $\lVert\partial_p\rVert_{\mathbbm{R}^2}=\lVert\partial_q\rVert_{\mathbbm{R}^2}=1$). 
Let $L,\underline{L}$ be two continuous, linearly independent null vector fields in $\tilde{U}$ such that $L\vert_{\tilde{p}}=\partial_p$ and $\underline{L}\vert_{\tilde{p}}=\partial_q$. We see that, in $\iota([u_1-\epsilon',u_1+\epsilon']\times[v_*',0))$, $L$ and $\underline{L}$ are proportional to $\iota_*\partial_u$ and $\iota_*\partial_v$ respectively. Since $\tilde{\lambda}^0\vert_{[u_1-\epsilon',u_1+\epsilon']}$ is null and transverse to the extension of $\tilde{\gamma}^u\vert_{[v_*',0)}$ for all $u\in[u_1-\epsilon',u_1+\epsilon']$ by Proposition \ref{prop:boundarystructure} (b), it follows that $L$ is tangent to $\tilde{\lambda}^0\vert_{[u_1-\epsilon',u_1+\epsilon']}$.

We now choose $\alpha\in(0,1)$ sufficiently small that $ \frac{2\alpha(1+\alpha)(1+\alpha^2)}{(1-\alpha)^5} < \frac{C_-^2}{C_+^2}$. By continuity, we can choose $\tilde{V}\subset\tilde{U}$ to be a sufficiently small neighbourhood of $\tilde{p}$ such that
\begin{align}
%\begin{split}
        \frac{\langle \partial_p,L\rangle_{\mathbbm{R}^2}}{\lVert L\rVert_{\mathbbm{R}^2}}, \frac{\langle \partial_q,\underline{L}\rangle_{\mathbbm{R}^2}}{\lVert\underline{L}\rVert_{\mathbbm{R}^2}}&>\frac{1}{\sqrt{1+\alpha^2}}\label{eqn:angles}\\
        \text{and } 1-\alpha\leq \sqrt{-\text{det}\tilde{g}}&\leq1+\alpha \label{eqn:determinant}
%\end{split}
\end{align} 
in $\tilde{V}$, where $\langle .,.\rangle_{\R^2}$ denotes the  inner product defined with respect to $\delta$.

Choosing $\epsilon' $ and $|v_*'|$ smaller if necessary, we have $\iota([u_1-\epsilon',u_1+\epsilon']\times[v_*',0])\subset \tilde{V}$. It follows from \eqref{eqn:angles} that, in $\iota([u_1-\epsilon',u_1+\epsilon']\times[v_*',0])$, the curves $\tilde{\gamma}^u$ and $\tilde{\lambda}^v$ (including $\tilde{\lambda}^0$) have gradients (when graphed in the $(q,p)$ plane) with absolute value less than $\alpha$ and greater than $1/\alpha$ respectively.

We use \eqref{eqn:angles} and \eqref{eqn:determinant} to obtain upper and lower bounds (expressed in terms of $\alpha$ and $q_1(v):= q(\tilde{\gamma}^{u_1}(v))$) for $\text{vol}_{\tilde{g}}(\iota(A_+(v))$ and $\text{vol}_{\tilde{g}}(\iota(A_-(v))$ respectively.

\textbf{Upper bound on $\text{vol}_{\tilde{g}}(\iota(A_+(v)))$:} For $v \in [v_*',0)$ sufficiently small,  the curves $p=\pm\alpha q$ and $p=\alpha q_1(v)+\frac{q_1(v)-q}{\alpha}$ enclose a region in $\tilde{U}$ which we call $T_+^{(\alpha)}(v)$ (see Figure \ref{fig:Tplus}).  By \eqref{eqn:monotonicityconventions} and \eqref{eqn:angles}, the curves $\tilde{\gamma}^{u_1}\vert_{[v,0]}$, $\tilde{\gamma}^{u_1 + \epsilon'}\vert_{[v,0]}$ and $\tilde{\lambda}^v\vert_{[u_1,u_1 + \epsilon']}$ are contained in $T_+^{(\alpha)}(v)$.

It follows from Proposition \ref{prop:boundarystructure} (a)(iii) that $\iota(A_+(v))\subset T_+^{(\alpha)}(v)$ and hence (see Figure \ref{fig:upperboundcalc}) 
\begin{equation}\label{eqn:volumeupperbound}
    \begin{split}
        \text{vol}_{\tilde{g}}(\iota(A_+(v)))&=\int_{\iota(A_+(v))}\sqrt{-\text{det}\tilde{g}}\,dpdq\\
        &\leq \int_{T^{(\alpha)}_+(v)}\sqrt{-\text{det}\tilde{g}}\,dpdq\\
        &\leq (1+\alpha)\int_{T^{(\alpha)}_+(v)}\,dpdq\\
       % &\leq \frac{5(1+2\alpha)}{8C(2-\alpha)}q_1^2
       &=\frac{\alpha(1+\alpha^2)}{(1-\alpha)}q_1(v)^2.
    \end{split}
\end{equation}

\begin{figure}[h]
    \centering
    \begin{minipage}{0.48\textwidth}
            \centering
    \includegraphics[scale=0.2]{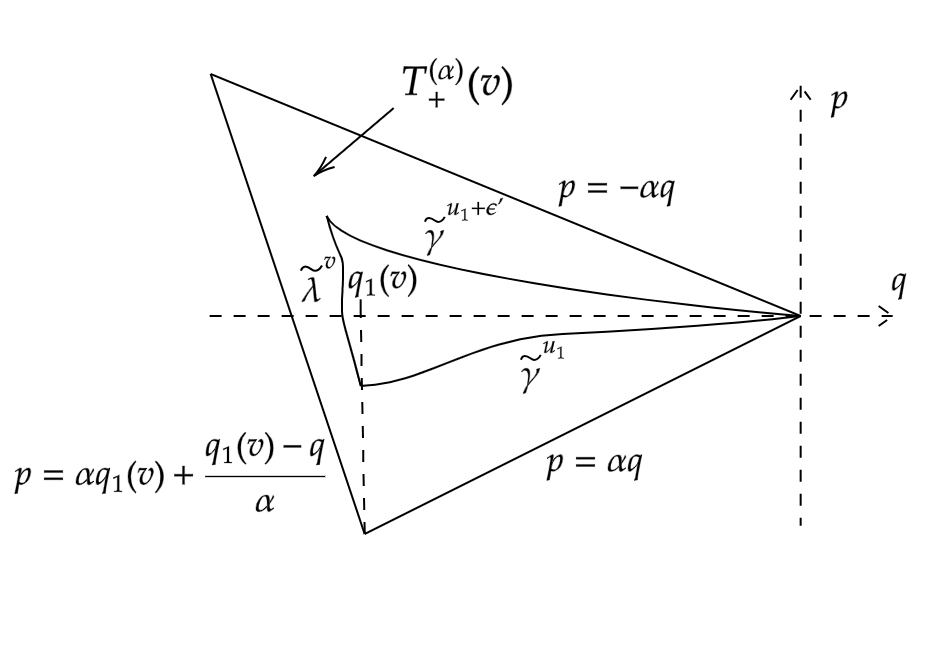}
    \caption{For $|v|$ sufficiently small $T^{(\alpha)}_+(v)\subset\tilde{U}$ is bound by the curves $p=\pm\alpha q$ and $p=\alpha q_1+\frac{q_1-q}{\alpha}$. The bounds \eqref{eqn:angles} combined with \eqref{eqn:monotonicityconventions} ensure that $\iota(A_+(v))\subset T^{(\alpha)}_+(v)$.}
    \label{fig:Tplus}
   \end{minipage}
    \minipage{0.04\textwidth}
  \includegraphics[scale=0.05]{blank.png}
 \endminipage
    \begin{minipage}{0.4\textwidth}
        \centering
    \includegraphics[scale=0.22]{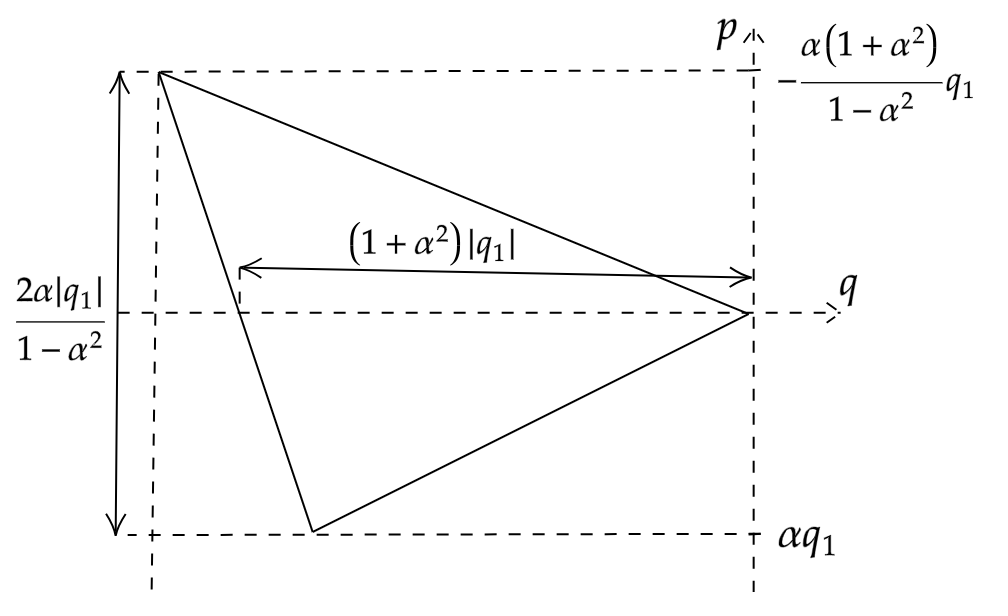}
    \caption{The area of the triangle $T^{(\alpha)}_+(v)$ in the $(q,p)$ plane can be computed using simple geometry.\\ \\}
   \label{fig:upperboundcalc}
    \end{minipage}
\end{figure}

\textbf{Lower bound on $\text{vol}(\iota(A_-(v)))$:} Now consider the region $\iota(A_-(v))$ for $v \in [v_*',0)$. For $|v|$ sufficiently small (i.e. $|q_1(v)|$ sufficiently small), the curves $p=\alpha q$, $p=\frac{q}{\alpha}$ and $p=\frac{q_1(v)-q}{\alpha}-\alpha q_1(v)$ enclose a region in $\tilde{U}$ which we call $P^{(\alpha)}(v)$.
From \eqref{eqn:monotonicityconventions} and \eqref{eqn:angles}, as well as from Proposition \ref{prop:boundarystructure} (b), it follows that the four curves $\tilde{\gamma}^{u_1 - \epsilon'}\vert_{[v,0)}$, $\tilde{\gamma}^{u_1}\vert_{[v,0)}$, $\tilde{\lambda}^{0}\vert_{[u_1 - \epsilon',u_1]}$ and $\tilde{\lambda}^{v}\vert_{[u_1 - \epsilon',u_1]}$ are contained in four wedges formed by straight lines with gradients in $\{\pm\alpha,\pm\frac{1}{\alpha}\}$ (see Figure \ref{fig:lowerbound}). It is here we use the fact that $\iota(u_1-\epsilon',0) \neq \tilde{p}$. The exact point at which $\tilde{\gamma}^{u_1 - \epsilon'}$ intersects $\tilde{\lambda}^0$ is not important -- all we need is that this is not $\tilde{p}$. This allows us to  choose $v$ sufficiently small (i.e. $q_1(v)$ sufficiently small) so that the lines $p=\frac{q}{\alpha}$ and $p=\frac{q_1(v)-q}{\alpha}-\alpha q_1(v)$ intersect above the wedge containing $\tilde{\gamma}^{u_1 - \epsilon'}\vert_{[v,0)}$. It follows that none of the four wedges in Figure \ref{fig:lowerbound} intersect $P^{(\alpha)}(v)$. Combining this with Proposition \ref{prop:boundarystructure} (a)(iii) we conclude that $P^{(\alpha)}(v)\subset \iota(A_-(v))\subset\tilde{U}$.

To simplify calculations, we consider $T^{(\alpha)}_-(v)\subset\tilde{U}$, where this is the region enclosed by the curves $p=\alpha q$, $p=\frac{q}{\alpha}$ and $p=-q+(1-\alpha)q_1(v)$ (see Figure \ref{fig:Tminus}). Since $\alpha\in(0,1)$, the line $p=-q+(1-\alpha)q_1(v)$ lies above the line $p=\frac{q_1(v)-q}{\alpha}-\alpha q_1(v)$ for $q>q_1(v)$ and hence $T^{(\alpha)}_-(v)\subset P^{(\alpha)}(v)\subset\iota(A_-(v))$. It follows that
(see Figure \ref{fig:lowerboundcalc}) 
\begin{equation}\label{eqn:volumelowerbound}
\begin{split}
    \text{vol}_{\tilde{g}}(\iota(A_-(v)))&=\int_{\iota(A_-)}\sqrt{-\det \Tilde{g}}\,dpdq\\
    &\geq \int_{T^{(\alpha)}_-(v)}\sqrt{-\det \Tilde{g}}\,dpdq\\
    &\geq (1-\alpha)\int_{T^{(\alpha)}_-(v)}\,dpdq\\
    &= \frac{(1-\alpha)^4}{2(1+\alpha)}q_1(v)^2.
\end{split}
\end{equation}
\begin{figure}[h]
    \centering
    \begin{minipage}{0.48\textwidth}
        \centering
    \includegraphics[scale=0.2]{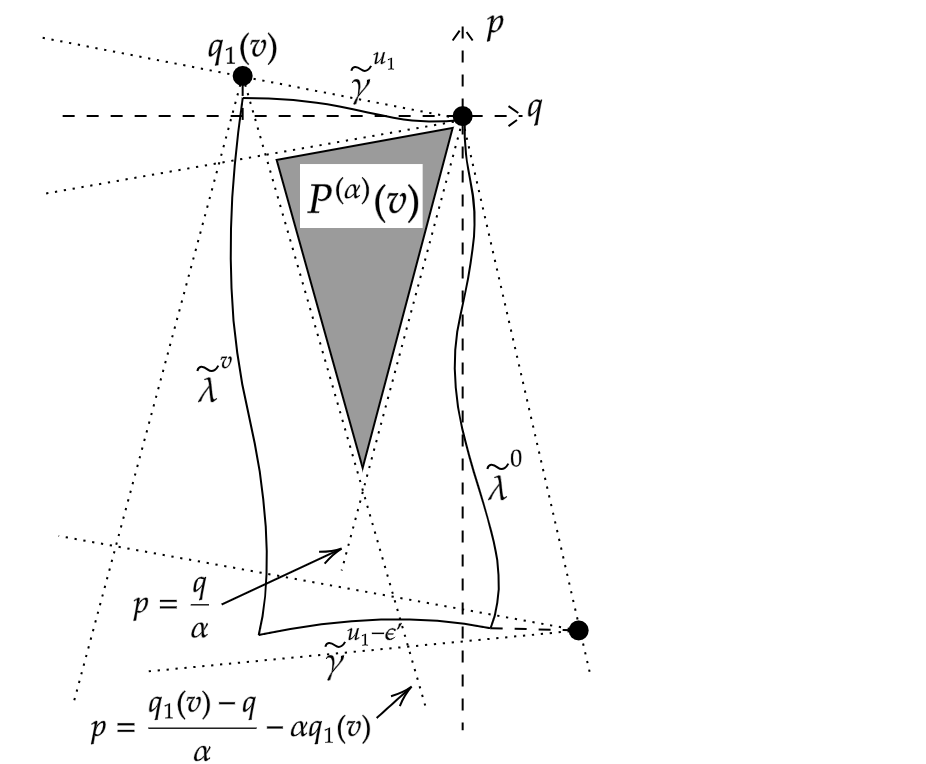}
    \caption{Dotted lines show four wedges (originating at the three black spots), formed by lines with gradients in $\{\pm\alpha,\pm\frac{1}{\alpha}\}$, which contain the four curves $\tilde{\gamma}^{u_1 - \epsilon'}\vert_{[v,0)}$, $\tilde{\gamma}^{u_1}\vert_{[v,0)}$, $\tilde{\lambda}^{0}\vert_{[u_1 - \epsilon',u_1]}$ and $\tilde{\lambda}^{v}\vert_{[u_1 - \epsilon',u_1]}$. Here, we use that $\iota(u_1 - \epsilon',0) \neq \tilde{p}$. For $\alpha$ fixed, choosing $v<0$ sufficiently small ensures 
 $P^{(\alpha)}(v)\subset\iota(A_-(v))$.}
    \label{fig:lowerbound}
    \end{minipage}%
     \minipage{0.02\textwidth}
    \includegraphics[scale=0.05]{blank.png}
 \endminipage
    \begin{minipage}{0.48\textwidth}
        \centering
    \includegraphics[scale=0.2]{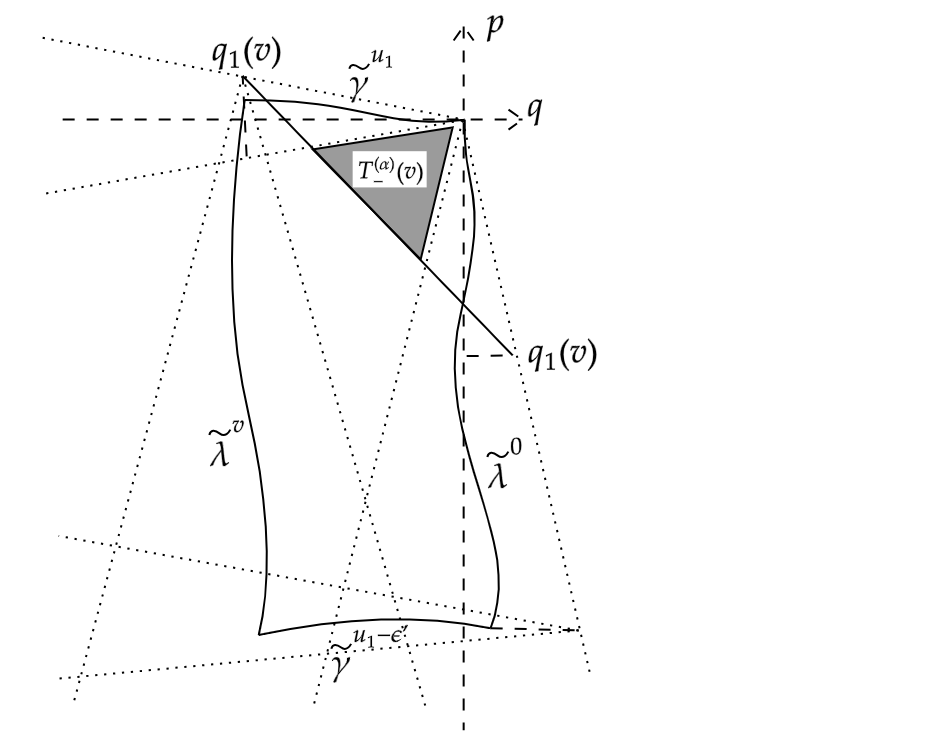}
    \caption{%For $v<0$ sufficiently small that $|q_1(v)|<|p_0(u)|$, 
    We use the region $T_-^{(\alpha)}(v)\subset P^{(\alpha)}(v)$ to obtain a lower bound for $\text{vol}_{\tilde{g}}(\iota(A_-(v)))$.\\ \\ \\ \\ \\}
    \label{fig:Tminus}
    \end{minipage}
\end{figure}
\begin{figure}[h]
    \centering
      \includegraphics[scale=0.3]{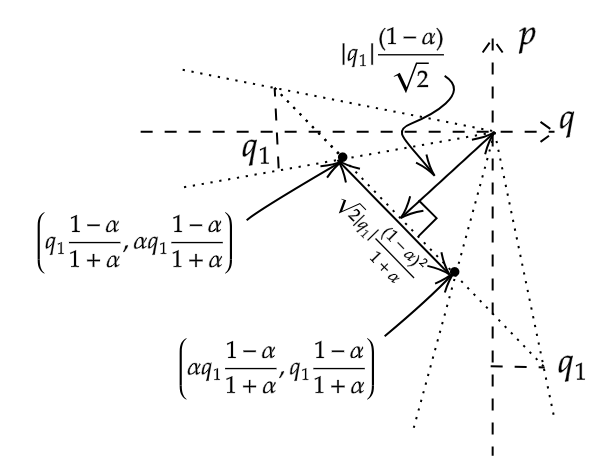}
    \caption{The area of $T^{(\alpha)}_-(v)$ in the $(q,p)$ plane can be computed using simple geometry.}
    \label{fig:lowerboundcalc}
\end{figure}

Combining this with \eqref{eqn:volumeinequality} and \eqref{eqn:volumeupperbound}, we obtain for $|v|$ sufficiently small
\begin{equation} 
0< \frac{C_-^2}{C_+^2}
\leq\frac{\text{vol}_{\tilde{g}}(\iota(A_+(v)))}{\text{vol}_{\tilde{g}}(\iota(A_-(v)))} \leq \frac{2\alpha(1+\alpha)(1+\alpha^2)}{(1-\alpha)^5}\;,
\end{equation}
which gives the desired contradiction by our choice of $\alpha$.

\end{proof}

\begin{thm}\label{thm:C0structure}
    Let $\iota : M_t \hookrightarrow \tilde{M}$ be a $C^0$-extension of $(M_t,g_t)$ such that $\tilde{p}_{u_*}:=\lim_{v \to 0} (\iota \circ \gamma^{u_*})(v) \in \partial \iota(M)$ exists for some $u_*\in(-1,1)$. Then one of the following holds:
    \begin{enumerate}
        \item[(a)] $\iota$ extends continuously to $\iota : (-1,1) \times (-1,0] \to 
        \tilde{M}$ with $\iota(u,0) = \tilde{p}_{u_*}$ for all $u \in (-1,1)$. 
        \item[(b)] 
        %$\tilde{p}_{u_*} \in \partial^+ \iota(M)$ and 
        $\tilde{p}_{u_*}\in\partial^+\iota(M_t)$ and there exist $v_*<0<\epsilon$
        %, and $\tilde{W} \subset \tilde{U}$  
        such that $\iota$ extends to $\iota : [u_*-\epsilon, u_*+\epsilon] \times [v_*,0] \to \tilde{U}_{\leq}$ as a homeomorphism onto its image, where $\big(\tilde{U}, \tilde{\varphi} = (\tilde{x}_0, \tilde{x}_1) \big)$ is a future boundary chart centred at $\tilde{p}_{u_*}$.
    \end{enumerate}
\end{thm}
In particular, if we are not in case $(a)$ then the $C^0$-extension $\iota : M_t \hookrightarrow \tilde{M}$ is locally $C^0$-equivalent to the reference extension $\overline{M_t}$ at $(u_*,0)\in\overline{M_t}$ and the conclusions of Proposition \ref{prop:boundarystructure} (b) apply.

\begin{proof} By Proposition \ref{prop:boundarystructure} (a)(i), we can choose a near Minkowski neighbourhood centred at $\tilde{p}_{u_*}$ as well as $v_*<0<\epsilon$ such that \eqref{eqn:monotonicityconventions} holds on $\iota([u_*-\epsilon,u_*+\epsilon]\times[v_*,0))$ and $\iota$ extends continuously to $\iota:[u_*-\epsilon,u_*+\epsilon]\times[v_*,0]\hookrightarrow\tilde{U}$. If $\iota(u_-,0)=\iota(u_+,0)$ for some $u_-,u_+\in[u_*-\epsilon,u_*+\epsilon]$ with $u_-<u_+$ then by Proposition \ref{prop:nomixedcornerextension}, $\iota$ extends continuously to $\iota : (-1,1) \times (-1,0] \to \tilde{M}$ with $\iota(u,0) = \tilde{p}_{u_*}$ for all $u \in (-1,1)$. Otherwise, $\iota :[u_*-\epsilon,u_*+\epsilon]\times[v_*,0]\hookrightarrow\tilde{U}$ is injective and thus a bijection onto its image. Since it is continuous and defined on a compact set, its inverse is also continuous \cite[Theorem 4.17]{Rudin}.
%Now suppose $\iota:[u_*-\epsilon,u_*+\epsilon]\times[v_*,0]\hookrightarrow\tilde{U}$ is a bijection.} In particular this map is a continuous bijection defined on a compact set, so its inverse is also continuous \cite[Theorem 4.17]{Rudin}.} 
The remaining properties follow from Remark \ref{remark:futureboundary}.
%By Proposition \ref{prop:boundarystructure} (b)(i), $\tilde{p}_{u_*}\in\partial^+\iota(M_t)$ and, reducing $\epsilon$ and $|v_*|$ if necessary, we can choose $(\tilde{U},\tilde{\varphi})$ to be a future boundary chart with
%$\iota([u_*-\epsilon,u_*+\epsilon]\times[v_*,0])\subset\tilde{U}_\leq$.
\end{proof}

Theorem \ref{thm:C0structure} constitutes the first main result of this paper. It classifies the possible local $C^0$-structures of $C^0$-extensions of $(M_t,g_t)$ across $\{v = 0\}$: either the $C^0$-structure is locally equivalent to that of the reference extension $\overline{M_t}$ or all future-directed null curves of constant $u$ extend to the same point. In this paper, we have only obtained this classification in 1+1-dimensions. Nevertheless, 
this is sufficient for exploring strongly spherically symmetric extensions (defined below) of black hole interiors to Cauchy horizons in 3+1-dimensions. This is addressed in the next section.

\subsection{Ruling out corner extensions from global structure}\label{cornerextensions}

In Example \ref{example:cornerextension} we constructed a \textit{corner extension} across $\{v=0\}$ of $(M_t,g_t)$ -- a $C^0$-extension such that $\iota(u,0)=\tilde{p}$ for all $u\in(-1,1)$. This construction required the extra assumption that $\Omega$ extends continuously to $\{u=\pm1\}$ as a strictly positive function.

In this section we leave our \emph{local} toy spacetime $(M_t, g_t)$ behind (recall this was a toy example describing a finite $u$-piece of a weak null singularity or of a Cauchy horizon).  We now consider a toy spacetime which models the entire weak null singularity or Cauchy horizon to the past. That is, we consider the  1+1-dimensional Lorentzian manifold $(Q,g_Q)$, where $Q=(-\infty,0)\times(-\infty,0)$ with global coordinates $(U,V)$,
    \begin{equation}
        g_Q=-\Omega(U,V)^2(dU\otimes dV+dV\otimes dU)
    \end{equation}
and $\Omega$ is a smooth, strictly positive function on $Q$ which extends  to $\overline{Q}:=(-\infty,0)\times(-\infty,0]$ as a continuous, strictly positive function.\footnote{The extension of the $V$-domain from $(-1,0)$ to $(-
\infty , 0)$ is of no importance for the result in this section.}
This is motivated by the study of \textit{strongly spherically symmetric} $C^0$-extensions of the subextremal Reissner Nordstr\"{o}m spacetime across the Cauchy horizon (see Section \ref{Implications for strongly spherically symmetric extensions across the Reissner-Nordstrom Cauchy horizon} for further discussion). In view of the rigidity of corner extensions exhibited in Proposition \ref{prop:nomixedcornerextension}, in this section we show that the \emph{global} structure of $(Q,g_Q)$ can rule out case $(a)$ in Theorem \ref{thm:C0structure} even for \emph{local} extensions.

%corner extensions across $\{V=0\}$ of the 1+1-dimensional Lorentzian manifold $(Q,g_Q)$, where $Q=(-\infty,0)\times(-\infty,0)$ with global coordinates $(U,V)$,
%    \begin{equation}
%        g_Q=-\Omega(U,V)^2(dU\otimes dV+dV\otimes dU)
%    \end{equation}
%and $\Omega$ is a smooth, strictly positive function on $Q$ which extends  to $\overline{Q}:=(-\infty,0)\times(-\infty,0]$ as a continuous, strictly positive function. This is motivated by the study of \textit{strongly spherically symmetric} $C^0$-extensions of the Reissner Nordstr\"{o}m spacetime across the Cauchy horizon (see Section \ref{Implications for strongly spherically symmetric extensions across the Reissner-Nordstrom Cauchy horizon} for further discussion).

Observe that for any $-\infty<U_1<U_2<0$, Example \ref{example:cornerextension} allows us to construct a corner extension across $\{V=0\}$ of the ``finite $U$-chunk'' $Q^*:=(U_1,U_2)\times(-\infty,0)\subset Q$. 
Indeed, consider $(Q^*, g_Q)$ with reference extension $\overline{Q^*}:=(U_1,U_2) \times (-\infty, 0]$. Rescaling $U$ and $V$ puts us in the setting of Example \ref{example:cornerextension} and we thus obtain a $C^0$-extension $\iota_{corner,\beta} : Q^* \hookrightarrow \tilde{Q}$ with $\iota_{corner,\beta}(U,0) = \tilde{p} \in \tilde{Q}$ for all $U \in (U_1,U_2)$.

We now ask whether this construction can be ``globalised'' to construct a corner extension of the full $(Q,g_Q)$ across $\{V=0\}$. The following example gives a condition which can be used to rule out the existence of such a global corner extension. It is based on the simple observation that, due to the local nature of the corner region, infinite spacetime volumes cannot be isometrically embedded into such regions.

\begin{prop}\label{prop:cornervolumecondition}
   Let $\iota$ be a $C^0$-extension of $(Q,g_Q)$ to $(\tilde{Q},\tilde{g}_{\tilde{Q}})$ such that $\tilde{p}:=\lim_{V\rightarrow0}\iota(U_*,V)\in\tilde{Q}$ exists for some $U_*\in(-\infty,0)$. Suppose $\mathrm{vol}_{g_Q}((-\infty,0)\times(V_*,0)):=\int_{-\infty}^0\int_{V_*}^0\Omega(U,V)^2dVdU=\infty$ for all $V_*<0$. Then there exists $\epsilon>0$ such that $\iota$ extends to $\iota:[U_*-\epsilon,U_*+\epsilon]\times(-\infty,0]\rightarrow\overline{Q}$ as a homeomorphism onto its image. In particular, $\iota$ is locally $C^0$-equivalent to $\overline{Q}$ at $(U_*,0) \in \overline{Q}$. 
\end{prop}

\begin{proof}
    For large $K>0$ consider the subregion $(U_* -K, -\frac{1}{K}) \times (-\infty, 0) \subset Q$. Clearly, restricting $\iota$ to this subregion gives a $C^0$-extension across $\{V=0\}$ to which, after rescaling of the coordinates, we can apply Theorem \ref{thm:C0structure}. If we are in case (b), the desired conclusion follows -- so we assume that case (a) holds and derive a contradiction.
%By applying Theorem \ref{thm:C0structure} to $\{(U,V)\in\left(U_*-K,\frac{U_*}{2}\right)\times(-1,0]\}$ for 
Since case (a) holds, we are in the setting of Proposition \ref{prop:nomixedcornerextension}. Let $(\tilde{U},\tilde{\varphi}=(\tilde{x}_0,\tilde{x}_1))$ be a near Minkowski neighbourhood of $\tilde{p}$. By  Proposition \ref{prop:nomixedcornerextension}, there exists $\tilde{B} \subset \tilde{U}$ and $V_*<0$ such that $\iota\big((U_* - K,-\frac{1}{K}) \times (V_*, 0) \big) \subset \tilde{B}$, cf.\ Figure \ref{fig:cornercondition}. Note that the size of $V_*<0$ depends only on $\tilde{U}$ and $\tilde{\gamma}^{U_*}$, so we can let $K \to \infty$ to infer that $\iota\big((- \infty,0)) \times (V_*, 0) \big) \subset \tilde{B}\subset \tilde{U}$. Since $\tilde{U}$ is a near Minkowski neighbourhood, working in the geometry of $(\tilde{Q},\tilde{g}_{\tilde{Q}})$ we have $\mathrm{vol}_{\tilde{g}_{\tilde{Q}}}\left(\iota\left((-\infty,0)\times(V_*,0)\right)\right):=\int_{\iota((-\infty,0)\times(V_*,0))}\sqrt{-\text{det} \tilde{g}_{\tilde{Q}}}<C$ for some constant $C$. Since $\iota$ is an isometry it follows that $\mathrm{vol}_{g_Q}((-\infty,0)\times(V_*,0))=\mathrm{vol}_{\tilde{g}_{\tilde{Q}}}\left(\iota\left((-\infty,0)\times(V_*,0)\right)\right)<C$ which is a contradiction. 
\end{proof}
\begin{figure}[h]
    \centering
    \includegraphics[scale=0.25]{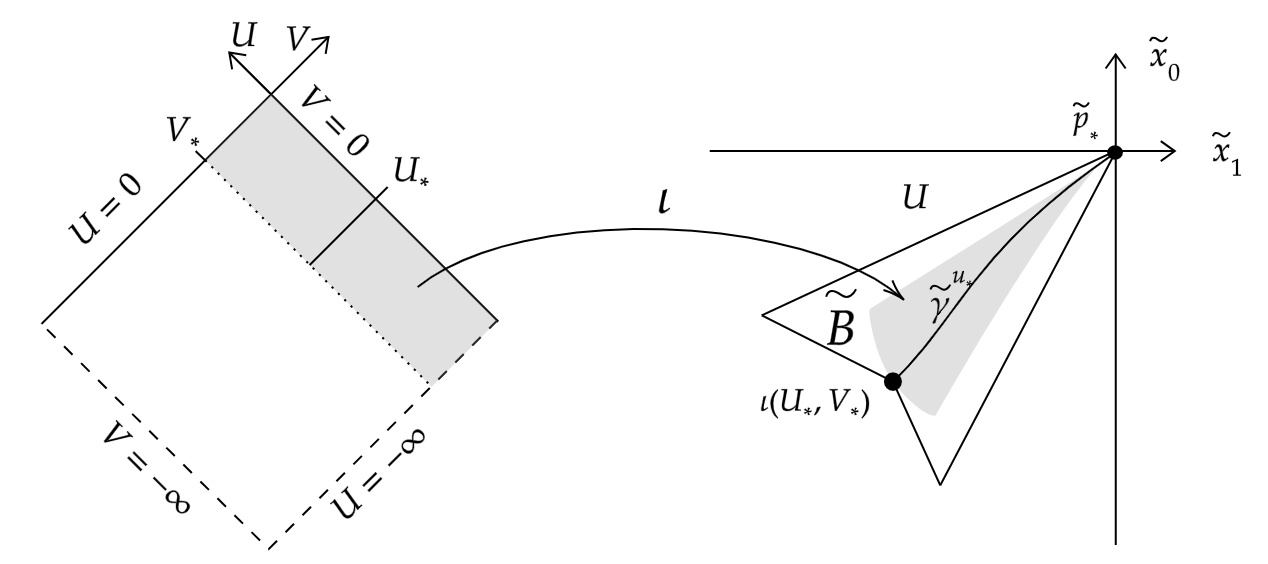}
    \caption{From Proposition \ref{prop:nomixedcornerextension} we can deduce that $\iota\left((-\infty,0)\times(V_*,0)\right)\subset\tilde{B}\subset\tilde{U}$.}
    \label{fig:cornercondition}
\end{figure}

%Let $\gamma^{U}$ denote the curve of constant $U$ in $Q$ parameterised by $V$ and define $\tilde{\gamma}^U:=\iota_Q\circ\gamma^U$. Let $(\tilde{U},\tilde{\varphi}=(\tilde{x}_0,\tilde{x}_1))$ be a near Minkowski neighbourhood centred on $\tilde{p}$ such that \eqref{eqn:monotonicityconventions} holds on $\tilde{\gamma}^{U_*}\vert_{[V_*,0]}$ for some $V_*<0$ with $|V_*|$ sufficiently small. Reducing $|V_*|$ if necessary, let $\tilde{B}\subset\tilde{U}$ be the closure of the subset of $\tilde{U}$ bound by the curves $\tilde{\gamma}^{U_*}\vert_{[V_*,0)}$, $\tilde{x}_1=0$ and the curve with gradient -2 in the $\tilde{x}_0-\tilde{x}_1$ plane through $\iota(U_*,V_*)$. This is shown in Figure \ref{fig:RNinteriorKruskal}. It follows from \eqref{eqn:gradientbound1+1}, \eqref{eqn:monotonicityconventions} and the injectivity of $\iota_Q$ that $\tilde{\gamma}^{U}\vert_{[V_*,0]}\subset\tilde{B}$ for all $U\in(-\infty,U_*]$. 

  % \begin{figure}[h]
   % \centering
%    \begin{minipage}{0.48\textwidth}
%            \centering
%            \includegraphics[scale=0.25]{RNinterior.png}
 %   \caption{For $U_*^{-1}<V_*<V_0<0$, we define $A(V_0)\subset M_{1+1}$ to be the region enclosed by the lines $UV=1$, $U=U_*$, $V=V_*$ and $V=V_0$. This figure depicts the region in $(u,v)$ space.}
  %  \label{RNinterior}
  %  \end{minipage}%
  %   \minipage{0.04\textwidth}
  %  \includegraphics[scale=0.05]{blank.png}
 %\endminipage
  %  \begin{minipage}{0.48\textwidth}
  %      \centering

   % \end{minipage}
%\end{figure}

\subsubsection{Implications for strongly spherically symmetric extensions across the Reissner-Nordstr\"{o}m Cauchy horizon}\label{Implications for strongly spherically symmetric extensions across the Reissner-Nordstrom Cauchy horizon}

As an example we consider the Reissner-Nordstr\"om black hole ``interior'' $(M_{RN},g_{RN})$. This spacetime describes the interior (i.e. the region inside the event horizon) of a spherically symmetric black hole of mass $M$ and charge $e$. We assume that $M^2>e^2>0$ (i.e. the black hole is ``sub-extremal''). Consider $\R^2\times\Sp^2$ with global coordinates $(u,v,\theta^A=(\theta,\phi))$ and
      \begin{equation}
        g_{RN}=-\frac{(r-r_+)(r-r_-)}{2r^2}(du\otimes dv+dv\otimes du)+r^2\cancel{g}_{\Sp^2}.
    \end{equation}
    where $r\in(r_-,r_+)$ is defined implicitly by 
\begin{align}\label{eqn:tortoise}
   \frac{v-u}{2}&=r+\frac{1}{2\kappa_+}\log\left(\frac{r_+-r}{r_+}\right)+\frac{1}{2\kappa_-}\log\left(\frac{r-r_-}{r_-}\right)\\
   &=:r_*(r)
\end{align}
and $\kappa_\pm=\frac{r_\pm-r_\mp}{2r_\pm^2}$, $r_\pm:=M\pm\sqrt{M^2-e^2}$. Note that $\frac{dr_*}{dr}=\frac{r^2}{(r_+-r)(r_--r)}<0$ and $r_*(r)\rightarrow\pm\infty$ as $r\rightarrow r_{\mp}$.
%Note that $\frac{dr_*}{dr}=\frac{r^2}{(r-r_+)(r-r_-)}$.
The boundary $r=r_-$ corresponds to a Cauchy horizon for the full Reissner-Nordstr\"{o}m spacetime \cite{HawkEllis}, while $r=r_+>r_-$ corresponds to the event horizon. 
In order to construct the reference extension across the Cauchy horizon we introduce Kruskal coordinates $(U,V)\in Q=(-\infty,0)\times(-\infty,0)$ defined by 
\begin{equation}\label{eqn:Kruskalcoords}
    U=-e^{-\kappa_-u}\quad V=-e^{\kappa_-v}.
\end{equation}
The coordinates $(U,V,\theta^A)$ are global coordinates on 
\begin{equation}\label{eqn:2dsplit}
M_{RN}:=Q\times \Sp^2
\end{equation}
and the metric in these coordinates is
\begin{equation}\label{eqn:RNKruskal}
    g_{RN}=\underbrace{-\underbrace{\frac{r_+r_-}{2\kappa_-^2r^2}e^{-2\kappa_-r}\left(\frac{r_+-r}{r_+}\right)^{1-\kappa_-/\kappa_+}}_{=: \Omega_{RN}^2}\left(dU\otimes dV+dV\otimes dU\right)}_{=: g_Q} +r^2\cancel{g}_{\Sp^2},
\end{equation}
where $r(U,V)$ is implicitly given by $UV = e^{2 \kappa_- r} \cdot \Big( \frac{r_+ - r}{r_+}\Big)^{\frac{\kappa_-}{\kappa_+}} \Big(\frac{r - r_-}{r_-}\Big)$ .
The reference extension across the `right' Cauchy horizon is given by $\overline{M_{RN}} := \overline{Q} \times \mathbb{S}^2:=(-\infty,0)\times(-\infty,0]\times \mathbb{S}^2$. Indeed, the metric \eqref{eqn:RNKruskal} extends analytically to $\{V=0\}$.

Note that $(Q, g_Q)$ is a 1+1-dimensional Lorentzian manifold. Consider a $C^0$-extension $\iota_Q : Q \hookrightarrow \tilde{Q}$ of $Q$ across $\{V=0\}$, where the continuous Lorentzian metric on $\tilde{Q}$ is denoted by $\tilde{g}_{\tilde{Q}}$. Assume that $r$ extends as a continuous strictly positive function to $\tilde{Q}$. We then set $\tilde{M} = \tilde{Q} \times \mathbb{S}^2$ and, with slight abuse of notation, define the spherically symmetric Lorentzian metric $\tilde{g} = \tilde{g}_{\tilde{Q}} + r^2 \cancel{g}_{\Sp^2}$ on $\tilde{M}$. Then $\iota := \iota_Q \times id_{\Sp^2} : M_{RN} \hookrightarrow \tilde{M}$ is a $C^0$-extension of $(M_{RN},g_{RN})$ across $\{V=0\}$. $C^0$-extensions arising in such a way are called  \textbf{strongly spherically symmetric}. 

We now discuss the implications of the results obtained in this section for such strongly spherically symmetric $C^0$-extensions. First let us note that if we only consider a subregion of $M_{RN}$ of finite $U$, i.e. 
\begin{equation*}
   M_{RN}^* = \underbrace{(U_1,U_2) \times (-\infty, 0)}_{=:Q^*} \times \mathbb{S}^2
\end{equation*}
for some $-\infty<U_1<U_2<0$, then it follows from our Example \ref{example:cornerextension} that \textbf{one can construct a strongly spherically symmetric $C^0$-extension of $M_{RN}^*$ across $\{V=0\}$ which terminates all outgoing null cones of constant $U$ in the same sphere}! This uses crucially that the area radius $r$ of the spheres has a $U$-independent limit for $V \to 0$.
\begin{example}\label{example:RNlocal}
Consider $(Q^*, g_Q)$ with reference extension $\overline{Q^*}:=(U_1,U_2) \times (-\infty, 0]$. Note, moreover, that $\Omega^2_{RN}$ extends continuously to $[U_1,U_2] \times (-\infty, 0]$. As discussed in Section \ref{cornerextensions}, there exists a $C^0$-extension $\iota_{corner,\beta} : Q^* \hookrightarrow \tilde{Q}$ with $\iota_{corner,\beta}(U,0) = \tilde{p} \in \tilde{Q}$ for all $U \in (U_1,U_2)$. Moreover, note that since $r$ extends continuously to $[U_1,U_2] \times (-\infty,0]$ with $r(U,0) = r_->0$ for all $U \in [U_1,U_2]$, it follows from the construction in Example \ref{example:cornerextension} that %$\left(\iota_{corner,\beta}\right)_*r$ 
$r$ also extends as a continuous strictly positive function to $\tilde{Q}$. Hence $\iota := \iota_{corner,\beta} \times id_{\Sp^2} : M_{RN}^* \hookrightarrow \tilde{M} := \tilde{Q} \times \Sp^2$ defines a strongly spherically symmetric $C^0$-extension which is not locally $C^0$-equivalent to the reference extension $\overline{M_{RN}^*} := (U_1,U_2) \times (-\infty, 0] \times \Sp^2$. 
\end{example}

The next result uses Proposition \ref{prop:cornervolumecondition} to show that one cannot globalise this corner construction to all of $(M_{RN},g_{RN})$ because of the infinite spacetime volume between the event and Cauchy horizons. In particular, any strongly spherically symmetric $C^0$-extension of $(M_{RN}, g_{RN})$ across $\{V=0\}$ is locally $C^0$-equivalent to the reference extension $\overline{M_{RN}}$.

\begin{cor}[Strongly spherically symmetric $C^0$-extensions of sub-extremal Reissner-Nordstr\"{o}m across the Cauchy horizon]\label{cor:eventhorizon}
Suppose $\iota=\iota_{Q}\times id_{\Sp^2}:M_{RN}\hookrightarrow\tilde{M}$ is a strongly spherically symmetric $C^0$-extension of $(M_{RN},g_{RN})$ across $\{V=0\}$ such that $\lim_{V\rightarrow 0}\iota_{Q}(U_*,V)=\tilde{p}$ for some $U_*\in(-\infty,0)$. Then there exists $\epsilon>0$  such that $\iota_{Q}$ extends to $[U_*-\epsilon,U_*+\epsilon]\times(-\epsilon,0]$ as a homeomorphism onto its image. In particular $\iota$ is locally $C^0$-equivalent to $\overline{M_{RN}}$ at $(U_*,0, \omega_0) \in \overline{M_{RN}}$ for any $\omega_0 \in \mathbb{S}^2$.
\end{cor}

%Since we are restricting to strongly spherically symmetric extensions, it suffices to show that the corresponding 1+1-dimensional $C^0$-extension of $Q$ is $C^0$-equivalent to the reference extension $\overline{Q}$. 

\begin{proof} %We will construct a region $A(V_0)\subset (-\infty, 0) \times (V_*, V_0)$ such that $\mathrm{vol}_{g_Q}(A(V_0))<\infty$ but $\mathrm{vol}_{g_Q}(A(V_0))\rightarrow\infty$ as $V_0\rightarrow0$. This implies that $\mathrm{vol}_{g_Q}\left((-\infty, 0) \times (V_*, 0)\right)=\infty$ for any $V_*<0$, so the result follows from Proposition \ref{prop:cornervolumecondition}.

We will show that  $\mathrm{vol}_{g_Q}\left((-\infty, 0) \times (V_*, 0)\right)=\infty$ for any $V_*<0$, so that the result follows from Proposition \ref{prop:cornervolumecondition}.
For $U_*^{-1}<V_*<V_0<0$, define (see Figure \ref{fig:RNUVspace}) $$A(V_0):=\{(U,V)\in Q:V^{-1}<U<U_*, V_*<V<V_0\}\;\subset(-\infty,0)\times(V_*,0).$$  From \eqref{eqn:tortoise} and \eqref{eqn:Kruskalcoords}, the condition $U>V^{-1}$ is equivalent to $r_*>0$.
  %, while the conditions $U<U_*$ and $V<V_0$ imply that $r_*<\frac{1}{2\kappa_-}\log(U_*V_0)$. 
  Since $r_*(r)$ is strictly decreasing, it follows that for points in $A(V_0)$ we have $r_-<r<r_0$, where $r_0$ is the unique solution to $r_*(r)=0$ (so in particular is independent of $V_0$). We conclude from \eqref{eqn:RNKruskal} that $\sqrt{-\mathrm{det} g_Q}=\Omega_{RN}^2> C_1 >0$ in $A(V_0)$, where $C_1$ is a constant independent of $V_0$. Using this, we have 
  \begin{equation}
      \begin{split}
          \mathrm{vol}_{g_Q}(A(V_0))&= \int_{A(V_0)} \sqrt{-\mathrm{det} g_Q}\, dUdV > C_1\int_{V_*}^{V_0}\left(U_*-\frac{1}{V}\right)dV = C_1U_*(V_0-V_*)-C_1\log\left(V_0V_*^{-1}\right)          
         % \\
        %\implies  C_1\int_{V_*}^{V_0}\left(U_*-\frac{1}{V}\right)dV<\mathrm{vol}_{g_Q}(A(V_0))&<C_2\int_{V_*}^{V_0}\left(U_*-\frac{1}{V}\right)dV\\
        %\implies
        %C_1U_*(V_0-V_*)-C_1\log\left(V_0V_*^{-1}\right)<\mathrm{vol}_{g_Q}(A(V_0))&<C_2U_*(V_0-V_*)-C_2\log\left(V_0V_*^{-1}\right)\\
        %&<\infty.
      \end{split}
  \end{equation}
 It follows that for any $V_*<0$ we have $\mathrm{vol}_{g_Q}(A(V_0))\rightarrow\infty$ as $V_0\rightarrow0$, which then implies the claim.
\end{proof}
    \begin{figure}[h]
    \centering
%    \begin{minipage}{0.48\textwidth}
 %           \centering
%\includegraphics[scale=0.25]{RNinteriorKruskal.png}
%    \caption{It follows from the lightcone bounds \eqref{eqn:gradientbound1+1}, the monotonicity \eqref{eqn:monotonicityconventions} and the injectivity of $\iota_Q$ that $\tilde{\gamma}^{U}\vert_{[V_1,0]}\subset\tilde{B}\subset\tilde{U}$ for all $U\in(-\infty,U_*]$.\\ \\}
%    \label{fig:RNinteriorKruskal}
%    \end{minipage}%
%     \minipage{0.04\textwidth}
%    \includegraphics[scale=0.05]{blank.png}
% \endminipage
%    \begin{minipage}{0.48\textwidth}
%        \centering
 \includegraphics[scale=0.25]{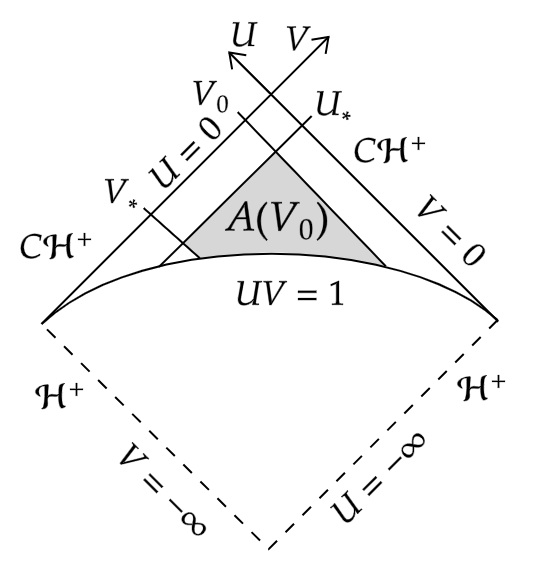}
    \caption{For $U_*^{-1}<V_*<V_0<0$, we define $A(V_0)\subset Q$ to be the region enclosed by the lines $UV=1$, $U=U_*$, $V=V_*$ and $V=V_0$. 
    %Recall that curves of constant $r$ correspond to curves $UV=constant$.
    }
    \label{fig:RNUVspace}
  %  \end{minipage}
\end{figure}

We remark that a similar argument also holds for the extremal case $e=M>0$.

To summarise, we have shown that a \emph{finite $U$-chunk} of the interior of sub-extremal Reissner-Nordstr\"om \emph{does} admit $C^0$-extensions with a different $C^0$-structure to the reference extension. In contrast, we have shown that the \emph{full} spacetime \emph{does not} admit such exotic $C^0$-extensions -- at least if we restrict to strongly spherically symmetric extensions. We also briefly remark that if one considers small and generic perturbations of exact sub-extremal Reissner-Nordstr\"om initial data under the spherically symmetric Einstein-Maxwell-scalar field system, then the Cauchy horizon turns into a spherically symmetric weak null singularity (see \cite{LukOh19I} and references therein). The spherically symmetric quotient $(Q,g_Q)$ again admits a continuous reference extension as in the case of exact sub-extremal Reissner-Nordstr\"om and, furthermore,  the area radius $r$ also extends continuously as a strictly positive function to the reference extension. However, in stark contrast to exact Reissner-Nordstr\"om, $r$ is generically no longer constant on the Cauchy horizon/weak null singularity, but is instead monotonically decreasing in $U$ \cite{Daf05a,LukOhShla23}. This even rules out our constructions of exotic $C^0$-extensions of finite $U$-chunks!\footnote{Recall that the construction in Example \ref{example:RNlocal} required $\lim_{V \to 0} r(U,V)$ to be independent of $U$.}

\section{The $C^1$-Structure of Extensions in 1+1-Dimensions}\label{SecC1Structure}

\subsection{Non-rigidity of the $C^1$-structure}

We have seen in Section \ref{Rigidity of the $C^0$-structure} how the requirement that the metric extends continuously imposes a certain rigidity on the $C^0$-structure of the extension. In particular, Theorem \ref{thm:C0structure} outlines the two possible local $C^0$-structures: that of the corner extension and that of the reference extension. In this section we investigate whether the local $C^0$-structure is sufficient to determine the local $C^1$-structure. Example \ref{ExC1Cor} shows that this is not the case if the $C^0$-structure is that of the corner. In Example \ref{example:C^1inequivalent2dimensions} we show that it is also not the case if the $C^0$-structure is that of the reference extension. We thus conclude that the assumption of the metric extending continuously is not strong enough to give us rigidity properties on the $C^1$-structure of the extension.

\begin{example}\label{ExC1Cor}
    For $\beta,\beta'\in(\frac{1}{2},1)$ with $\beta'>\beta$, let $\iota_{corner,\beta}:M_t\hookrightarrow \tilde{M}_t$ and $\iota_{corner,\beta'}:M_t\hookrightarrow \tilde{M}_t$ be two corner extensions of $(M_t,g_t)$ into $(\tilde{M}_t,\tilde{g}'_\beta)$ and $(\tilde{M}_t,\tilde{g}'_{\beta'})$ as described in Example \ref{example:cornerextension}.
    We claim that $\iota_{corner,\beta}$ and $\iota_{corner,\beta'}$ are $C^0$-equivalent but $C^1$-inequivalent extensions in the sense that the identification map $id:\iota_{corner,\beta}\rightarrow\iota_{corner,\beta'}$ defined by $id:=\iota_{corner,\beta'}\circ\iota_{corner,\beta}^{-1}$ extends to $\overline{\iota_{corner,\beta}(M_t)}\subset\tilde{M}_t$ as a homeomorphism but not as a diffeomorphism.

    The identification map is given in coordinates by
    \begin{equation}
        id(p,q)=\left(p|q|^{\frac{\beta'-\beta}{1-\beta}},-|q|^{\frac{1-\beta'}{1-\beta}}\right)=:(p'(p,q),q'(p,q)).
    \end{equation}
    This map extends as a homeomorphism to $\overline{\iota_{corner,\beta}(M_t)}=\{(p,q)\in(-1,1)\times(-1,0]:|p|\leq|q|^{\frac{\beta}{1 - \beta}}\}$. However, we have
    \begin{equation}
    \begin{split}
          \partial_qq'(p,q)&=\frac{1-\beta'}{1-\beta}|q|^{\frac{\beta-\beta'}{1-\beta}} \\
          &\rightarrow\infty \text{ as }q\rightarrow0
    \end{split}
    \end{equation}
so $id$ does not extend to $\overline{\iota_{corner,\beta}(M_t)}$ as a diffeomorphism.
\end{example}

In the following example we construct a $C^0$-extension of $(M_t,g_t)$ which is locally $C^0$-equivalent to $\overline{M_t}$ at $(0,0)\in\overline{M_t}$ but not locally $C^1$-equivalent there.
\begin{example} \label{example:C^1inequivalent2dimensions}
Let $D:=\left[-\frac{1}{2},\frac{1}{2}\right]\times\left[-\frac{1}{2},0\right]\subset\mathbbm{R}^2$ with standard $(p,q)$-coordinates.
%Throughout this argument we will use $(p,q)$ as global coordinates on $D$.
Define real-valued functions $u,v$ on $D$ by
\begin{equation}\label{eqn:uvdefinition}
    \begin{split}
        u(p,q)&=\begin{cases} -C\int_0^p\log(s^2+2q^2)ds \qquad &\textnormal{ for } (p,q) \neq (0,0) \\ 0 &\textnormal{ for } (p,q) = (0,0) \end{cases} \\
        v(p,q)&= \begin{cases} -C\int_0^q\frac{1}{\log(p^2+q^2+t^2)}dt \qquad \;\, &\textnormal{ for } (p,q) \neq (0,0) \\ 0 &\textnormal{ for } (p,q) = (0,0) \end{cases}
    \end{split}
\end{equation}
where $C>0$ is a constant to be determined. It is easy to show that $u$ and $v$ are continuous on $D$. For $q<0$ we have
\begin{equation*}
u(p,q)=-Cp\log(p^2+2q^2)+2Cp-2\sqrt{2}Cq\tan^{-1}(p/\sqrt{2}q) \;.   \end{equation*}
Furthermore, we note that
\begin{equation}\label{eqn:monotonicityuv}
    \begin{split}
         u(p,q)&\gtrless0 \text{ for }p\gtrless0\\
    u(p,q)&=0\iff p=0\\
    v(p,q)&\leq0 \\
        v(p,q)&=0\iff q=0
    \end{split}
\end{equation}
and 
\begin{equation}\label{eqn:oddeven}
\begin{split}
    (u(p,q),v(p,q))&=(-u(-p,q),v(-p,q)).
\end{split}
\end{equation}
Let $\epsilon'\in\left(0,\frac{1}{2}\right)$. The function $(s,q) \mapsto \log(s^2+2q^2)$ is smooth on $[-\frac{1}{2},\frac{1}{2}]\times\left[-\frac{1}{2},-\epsilon'\right]$. It follows that $u(p,q)$ is smooth on $\left[-\frac{1}{2},\frac{1}{2}\right]\times\left[-\frac{1}{2},-\epsilon'\right]$. Similarly, the function $(t,p,q) \mapsto \frac{1}{\log(p^2+q^2+t^2)}$ is smooth on the region $t\in[q,0]$, $(p,q)\in\left[-\frac{1}{2},\frac{1}{2}\right]\times\left[-\frac{1}{2},-\epsilon'\right]$. It follows that $v(p,q)$ is smooth on $\left[-\frac{1}{2},\frac{1}{2}\right]\times\left[-\frac{1}{2},-\epsilon'\right]$. Since $\epsilon'\in\left(0,\frac{1}{2}\right)$ was arbitrary, we conclude that $u(p,q)$ and $v(p,q)$ are smooth on $\left[-\frac{1}{2},\frac{1}{2}\right]\times\left[-\frac{1}{2},0\right)$.

On $\left[-\frac{1}{2},\frac{1}{2}\right]\times\left[-\frac{1}{2},0\right)$, we have
\begin{equation}\label{eqn:partialderivs}
    \begin{split}
        \frac{\partial u}{\partial p}&=-C\log(p^2+2q^2)\\
        \frac{\partial u}{\partial q}
        &=-2\sqrt{2}C\tan^{-1}(p/\sqrt{2}q) \\
        \frac{\partial v}{\partial p}&=2Cp\int^q_0\frac{1}{(\log(p^2+q^2+t^2))^2(p^2+q^2+t^2)}dt\\
        \frac{\partial v}{\partial q}&=-\frac{C}{\log(p^2+2q^2)}+2Cq\int^q_0\frac{1}{(\log(p^2+q^2+t^2))^2(p^2+q^2+t^2)}dt.
    \end{split}
\end{equation}

%\begin{equation}
%\begin{split}
 %       \left|\left.\frac{\partial v}{\partial p}\right|_{p=0}\right|&=\left|C\lim_{h\rightarrow0}\int_0^q\frac{\log(1+h^2/t^2)}{h\log t^2\log(h^2+t^2)}dt\right|\\
  %      &\leq \left|\frac{C}{\log q^2}\lim_{h\rightarrow0}\frac{1}{h\log(h^2+q^2)}\int^q_0\log(1+h^2/t^2)dt\right|\\
   %     &=\left|\frac{C}{\log q^2}\lim_{h\rightarrow0}\frac{1}{h\log(h^2+q^2)}\left[q\log(h^2/q^2+1)-2h\arctan(h/q)+h\pi\right]\right|
%\end{split}
%\end{equation}

For future use, we collect the following properties which hold on $\left[-\frac{1}{2},\frac{1}{2}\right]\times\left[-\frac{1}{2},0\right)$.
\begin{equation}\label{eqn:monotonicityderivatives}
    \begin{split}
        \frac{\partial u}{\partial p},\frac{\partial v}{\partial q}&>0\\
        \frac{\partial u}{\partial q}&\gtrless0 \text{ for } p\gtrless0\\
        \frac{\partial u}{\partial q}&=0\text{ iff }p=0 \text{ (i.e. $u=0$)}\\
         \frac{\partial v}{\partial p}&\lessgtr0\text{ for } p\gtrless0\\
        \frac{\partial v}{\partial p}&=0\text{ iff }p=0 \text{ (i.e. $u=0$)}
    \end{split}
\end{equation}

%and define the smooth map
%Note that $u$ and $v$ extend continuously to functions on $D$ (the closure of $D$). For simplicity we also denote these extended functions by $u$ and $v$ respectively. 

Define the map
\begin{equation}
\begin{split}
    \phi:D&\rightarrow \mathbbm{R}^2\\
    (p,q)&\mapsto(u(p,q),v(p,q)).
    \end{split}
\end{equation}
We will use the inverse of a restriction of this map to define a $C^0$-extension of $(M_t,g_t)$. 
%We begin by showing that, for a suitable choice of $C$, $\phi$ is injective and $(-1,1)\times (-1,0)\subset \phi(D)$. 

%Note that ensures that $\phi(D)\subset tildeD$is this the best order to have all this in? is tildeD actually needed?

\textbf{Step 1:} We first show that $\phi$ is injective. Suppose $\phi(p_1,q_1)=\phi(p_2,q_2)$, for some $(p_1,q_1)\neq (p_2,q_2)\in D$, where without loss of generality $p_1\leq p_2$. We will lead this to a contradiction. By \eqref{eqn:monotonicityuv}, we must have either $p_1\leq p_2<0$ or $0\leq p_1\leq p_2$. By \eqref{eqn:oddeven}, we may assume without loss of generality that $0\leq p_1\leq p_2$. 

Consider the path from $(p_1,q_1)$ to $(p_2,q_2)$ consisting of the curve $q=q_1$ composed with the curve $p=p_2$ (Figure \ref{fig:pathinD}).  Note that $p \mapsto u(p,0)$ is strictly increasing on $[0,\frac{1}{2}]$. Combining this with \eqref{eqn:monotonicityuv} and  \eqref{eqn:monotonicityderivatives} we see that 
\begin{equation}\label{eqn:uvinequality}
\begin{split}
    u(p_2,q_1)&\geq u(p_1,q_1) \qquad \textnormal{ with equality if, and only if } p_1 = p_2\\
   v(p_2,q_1)&\leq v(p_1,q_1).
\end{split}
\end{equation}
%Note that equality holds in the first line if and only if $p_1=p_2$.
\begin{figure}[h]
    \centering
    \includegraphics[scale=0.25]{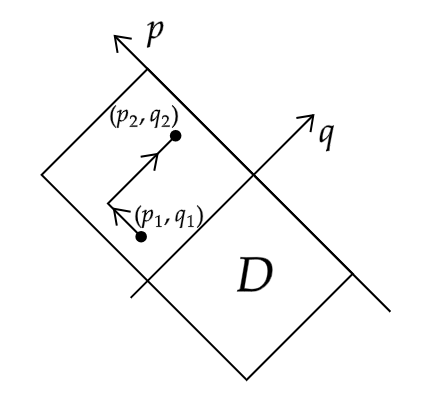}
    \caption{Path in $D$ used in the monotonicity argument to show that $\phi$ is injective (assuming $q_2>q_1$).}
    \label{fig:pathinD}
\end{figure}

We now distinguish the following three cases:
\begin{itemize}
    \item $q_1 = q_2$: we must then have $p_2 > p_1$. Now, the first line of \eqref{eqn:uvinequality} gives the contradiction $u(p_2,q_2) > u(p_1,q_1)$.
    \item $q_2 > q_1$: then, from \eqref{eqn:monotonicityderivatives} we have $u(p_2,q_2) > u(p_2,q_1)$, which, together with \eqref{eqn:uvinequality} gives again the same contradiction.
    \item $q_1 > q_2$: then, from \eqref{eqn:monotonicityderivatives} we have $v(p_2,q_2) < v(p_2,q_1)$ which gives a contradiction together with the second line of \eqref{eqn:uvinequality}.
\end{itemize}

%If $q_2<q_1$ then, once again using \eqref{eqn:monotonicityderivatives}, we have $v(p_2,q_2)< v(p_2,q_1)$. Combining this with \eqref{eqn:uvinequality} we have
%\begin{equation}
%    v(p_2,q_2)< v(p_2,q_1)\leq v(p_1,q_1)
%\end{equation}
% which is a contradiction. We conclude that $q_2\geq q_1$ and hence (again by \eqref{eqn:monotonicityderivatives}) we have $u(p_2,q_2)\geq u(p_2,q_1)$ with equality if and only if $p_2=0$ or $q_1=q_2$. Again using \eqref{eqn:uvinequality} we see that
%\begin{equation}
%\begin{split}
%        u(p_2,q_2)&\geq u(p_2,q_1)\geq u(p_1,q_1)
%\end{split}
%\end{equation}
%with equality between all three quantities if and only if $p_1=p_2=0$ or $(p_1,q_1)=(p_2,q_2)$. Suppose $p_1=p_2=0$. Then from \eqref{eqn:monotonicityderivatives} we have $v(p_2,q_2)\geq v(p_2,q_1)=v(p_1,q_1)$
%with equality if and only if $q_1=q_2$. We conclude that $\phi$ is injective. 

\textbf{Step 2:} Next we show that $[-1,1]\times[-1,0]\subset\phi\left(\left(-\frac{1}{2},\frac{1}{2}\right)\times\left(-\frac{1}{2},0\right]\right)$ for $C>0$ sufficiently large. 

%I moved this from earlier - maybe not needed?Along $\mu_1$ and $\mu_2$ we have $|u|>u':=u(\frac{1}{2},\frac{1}{2})$, while along $\mu_3$ we have $v\leq v':=v(\frac{1}{2},-\frac{1}{2})$. By choosing $C$ sufficiently large, we can ensure that $u',|v'|>0$ and hence $(-1,1)\times(-1,0)\subset tildeD$ explain why this is important. 

First observe that, for any $p\in \left[-\frac{1}{2},\frac{1}{2}\right]$ and any $q\in\left[-\frac{1}{2},0\right]$, we have from \eqref{eqn:monotonicityuv}, \eqref{eqn:oddeven} and \eqref{eqn:monotonicityderivatives} that
\begin{equation}
    \begin{split}
        \left|u\left(\pm\frac{1}{2},q\right)\right| &\geq\left|u\left(\frac{1}{2},-\frac{1}{2}\right)\right|\\
        &=-C\int_0^{\frac{1}{2}}\log\left(s^2+\frac{1}{2}\right)ds\\
        &>0\\
        \text{ and }v\left(p,-\frac{1}{2}\right)&\leq v\left(0,-\frac{1}{2}\right)\\
        &=C\int^{\frac{1}{2}}_0\frac{dt}{\log \left(t^2+\frac{1}{4}\right)}\\
       &<0.
       % &=\frac{C}{2\log\left(\frac{4}{3}\right)}
    \end{split}
\end{equation}

It follows that by choosing $C$ sufficiently large we have 
\begin{equation}\label{eqn:uvbounds}
    \begin{split}
\left|u\left(\pm\frac{1}{2},q\right)\right| >1
        \text{ and }v\left(p,-\frac{1}{2}\right)<-1.
    \end{split}
\end{equation}
%and hence $[0,1)\times(-1,0]\subset tildeD$
%It therefore suffices to show that $tildeD\subset\phi(D)$.
for all $p\in \left[-\frac{1}{2},\frac{1}{2}\right]$ and $q\in\left[-\frac{1}{2},0\right]$.

By \eqref{eqn:oddeven}, in order to show that $[-1,1]\times[-1,0]\subset\phi\left(\left(-\frac{1}{2},\frac{1}{2}\right)\times\left(-\frac{1}{2},0\right]\right)$, it suffices to show that $[0,1]\times[-1,0]\subset\phi\left(\left(-\frac{1}{2},\frac{1}{2}\right)\times\left(-\frac{1}{2},0\right]\right)$. 
Let $(u_1,v_1)\in [0,1]\times[-1,0]$. 

Suppose $v_1 = 0$. We have $u\left(0,0\right)=0\leq u_1\leq1< u\left(\frac{1}{2},0\right)$. Since $u\left(p,0\right)$ is a continuous function of $p$, it follows from the intermediate value theorem that there exists $p_*\in \left[0,\frac{1}{2}\right)$ such that $u\left(p_*,0\right)=u_1$. By \eqref{eqn:monotonicityuv} we also have $v(p_*,0)=0$. A similar argument shows that if $u_1=0$ then there exists $q_*\in\left(-\frac{1}{2},0\right]$ such that $u(0,q_*)=0$ and $v(0,q_*)=v_1$. 

It remains to consider $(u_1,v_1)\in(0,1]\times [-1,0)$. 
We consider the space of solutions to each of the two equations $u(p,q)=u_1$ and $v(p,q)=v_1$ separately. 
First, fix $q_*\in \left[-\frac{1}{2},0\right)$.
We have $u\left(0,q_*\right)=0< u_1\leq1< u\left(\frac{1}{2},q_*\right)$ and $u\left(p,q_*\right)$ a continuous function of $p$. It follows from the intermediate value theorem that there exists $p_*\in \left(0,\frac{1}{2}\right)$ such that $u\left(p_*,q_*\right)=u_1$. From \eqref{eqn:monotonicityderivatives} we have $\frac{\partial }{\partial p}u(p,q_*)>0$, so it follows that the value $p_*$ is unique. We can therefore define a function $p_1:\left[-\frac{1}{2},0\right)\rightarrow\left(0,\frac{1}{2}\right)$ such that $u(p_1(q),q))=u_1$ for all $q\in \left[-\frac{1}{2},0\right)$. We have $\frac{\partial u}{\partial p}>0$ for $(p,q)\in(0,\frac{1}{2})\times[-\frac{1}{2},0)$, so the implicit function theorem implies that $p_1$ is continuously differentiable on $\left[-\frac{1}{2},0\right)$. We have
\begin{equation}
\begin{split}
        \frac{dp_1}{dq}&=-\frac{\frac{\partial u}{\partial q}}{\frac{\partial u}{\partial p}}<0
\end{split}
\end{equation} 
and hence $p_1(q)$ is strictly decreasing. Since $p_1(q)$ is bounded, we can therefore define $p_1(0):=\lim_{q\rightarrow 0}p_1(q)\in\left[0,\frac{1}{2}\right)$. The continuity of $u(p,q)$ implies that $u(p_1(0),0)=u_1$. Since $u_1\in(0,1]$ it follows from \eqref{eqn:monotonicityuv} that $p_1(0)\in\left(0,\frac{1}{2}\right)$.

We can therefore extend $p_1$ to a continuous function $p_1:\left[-\frac{1}{2},0\right]\rightarrow\left(0,\frac{1}{2}\right)$ such that $u(p_1(q),q)=u_1$ for all $q\in \left[-\frac{1}{2},0\right]$. This defines a curve 
%$(p_1(q),q)\subset\left[0,\frac{1}{2}\right]\times\left[-\frac{1}{2},0\right]$ 
in $\left[0,\frac{1}{2}\right]\times\left[-\frac{1}{2},0\right]$
with endpoints on the lines $q=-\frac{1}{2}$ and $q=0$ (see Figure \ref{fig:surjectivity}).
%, where we use $(p,q)$ as global coordinates on $[0,\frac{1}{2}] \times \left[-\frac{1}{2},0\right]$.

Similarly, we can define a continuously differentiable function $q_1:\left[0,\frac{1}{2}\right]\rightarrow \left(-\frac{1}{2},0\right)$ such that $v(p,q_1(p))=v_1$ for all $p\in \left[0,\frac{1}{2}\right]$ (see Figure \ref{fig:surjectivity}). 
For $p\in\left(0,\frac{1}{2}\right)$, we have
\begin{equation}
\begin{split}
        \frac{dq_1}{dp}&=-\frac{\frac{\partial v}{\partial p}}{\frac{\partial v}{\partial q}}>0.
\end{split}
\end{equation}  
so it follows that $q_1$ can be inverted to define a continuous function $q_1^{-1}:\left[q_-,q_+\right]\rightarrow\left[0,\frac{1}{2}\right]$, where $q_-:=q_1(0)<q_1\left(\frac{1}{2}\right)=:q_+$. 

\begin{figure}[h]
    \centering
    \def\svgwidth{7.5cm}
    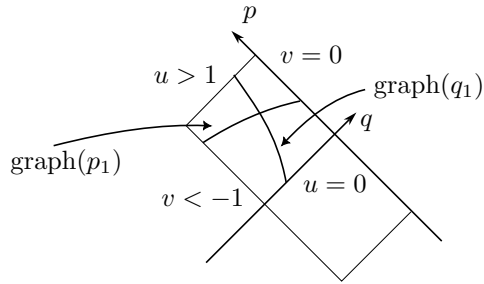
\caption{Proving that $(0,1] \times [-1,0) \subset \phi((-\frac{1}{2}, \frac{1}{2}) \times (-\frac{1}{2},0])$.}
    \label{fig:surjectivity}
\end{figure}

Consider the function
\begin{equation}
    \begin{split}
        f:\left[q_-,q_+\right]:&\rightarrow \left(-\frac{1}{2},\frac{1}{2}\right)\\
        q&\mapsto p_1(q)-q_1^{-1}(q)
    \end{split}
\end{equation}
We have
\begin{equation}
    \begin{split}
        f(q_-)&=p_1(q_-)>0\\
        f\left(q_+\right)&=p_1(q_+)-\frac{1}{2}<0.
    \end{split}
\end{equation}
Since $f$ is continuous, it follows from the intermediate value theorem that there exists $q_0\in(q_-,q_+)\subset\left(-\frac{1}{2},0\right)$ such that $f(q_0)=0$. Hence $p_1(q_0)=q_1^{-1}(q_0)\in\left(0,\frac{1}{2}\right)$, so $u(p_1(q_0),q_0)=u_1$ and $v(p_1(q_0),q_0)=v_1$. 
 
\textbf{Step 3:} The results of Steps 1 and 2 allow us to
make the following definition.
\begin{equation}\label{eqn:phiinverse}
    \left(\phi\vert_{\phi^{-1}([-1,1]\times[-1,0])}\right)^{-1}:[-1,1]\times[-1,0]\rightarrow \phi^{-1}([-1,1]\times[-1,0]).
    %\subset \left(-\frac{1}{2},\frac{1}{2}\right)\times\left(-\frac{1}{2},0\right).
\end{equation}
We will abuse notation and write $\phi^{-1}$ to denote this map. We claim that $\phi^{-1}$ is continuous. 
%To prove this, we begin by showing that $\phi^{-1}([-1,1]\times[-1,0])\subset\mathbbm{R}^2$ is compact. We have $\phi^{-1}([-1,1]\times[-1,0])\subset D$, so this set is bounded. To show that it is closed, let $\left(a_n\right)_{n=1}^\infty$ be a sequence in $\phi^{-1}([-1,1]\times[-1,0])$ such that $a_n\rightarrow a\in\mathbbm{R}^2$ as $n\rightarrow\infty$. Since $D$ is closed, we have $a\in D$. It follows from the continuity of $\phi$ that $\phi(a)=\lim_{n\rightarrow\infty}\phi(a_n)\in [-1,1]\times[-1,0]$ since this set is closed. Hence $a\in\phi^{-1}([-1,1]\times[-1,0])$. We conclude that $\phi^{-1}([-1,1]\times[-1,0])$ is closed and bounded, and hence compact.
Since $\phi$ is continuous, we have that $\phi^{-1}([-1,1] \times [-1,0]) \subset D \subset \R^2$ is closed and thus compact.
Hence $\phi\vert_{\phi^{-1}([-1,1]\times[-1,0])}:\phi^{-1}([-1,1]\times[-1,0])\rightarrow [-1,1]\times[-1,0]$ is a continuous bijection defined on a compact set. It follows that its inverse is also continuous \cite[Theorem 4.17]{Rudin}. 

\textbf{Step 4:} We now show that $\phi^{-1}$ is smooth on $(-1,1)\times(-1,0)$. 

The Jacobian of $\phi$ is
\begin{equation}
\begin{split}
        Jac(p,q):=&\frac{\partial u}{\partial p}\frac{\partial v}{\partial q}-\frac{\partial u}{\partial q}\frac{\partial v}{\partial p}\\
   >&0 \text{ for all }(p,q)\in \left(-\frac{1}{2},\frac{1}{2}\right)\times\left(-\frac{1}{2},0\right).\\
\end{split}
\end{equation}

Furthermore, we have
    $[-1,1]\times[-1,0]\subset \phi\left(\left(-\frac{1}{2},\frac{1}{2}\right)\times\left(-\frac{1}{2},0\right]\right)$ and $v(p,q)=0$ if and only if $q=0$. Hence $(-1,1)\times(-1,0)\subset \phi\left(\left(-\frac{1}{2},\frac{1}{2}\right)\times\left(-\frac{1}{2},0\right)\right)$.
Since $\phi$ is smooth on $ \left(-\frac{1}{2},\frac{1}{2}\right)\times\left(-\frac{1}{2},0\right)$, it follows from \cite[Theorem A.1.]{KassVos} that $\phi^{-1}$ is smooth on $(-1,1)\times(-1,0)$. 

\textbf{Step 5:} We now use $\phi^{-1}$ to define a continuous extension of $(M_t,g_t)$. 

We have shown that $\phi$ defines a continuous bijection on $\left[-\frac{1}{2},\frac{1}{2}\right]\times\left[-\frac{1}{2},0\right]$ and that $[-1,1]\times[-1,0]\subset\phi\left(\left[-\frac{1}{2},\frac{1}{2}\right]\times\left[-\frac{1}{2},0\right]\right)$. We extend the domain of the functions $u(p,q)$ and $v(p,q)$ to  $\left[-\frac{1}{2},\frac{1}{2}\right] \times \left[-\frac{1}{2},\frac{1}{2}\right]$ by
\begin{equation}\label{eqn:extendingphi}
    \begin{split}
        (u(p,q),v(p,q))&=(u(p,-q),-v(p,-q)).
    \end{split}
\end{equation}
Since $v(p,0) = 0$, it follows that $\phi$ now defines a continuous bijection on $\left[-\frac{1}{2},\frac{1}{2}\right]\times\left[-\frac{1}{2},\frac{1}{2}\right]$, with $[-1,1]\times[-1,1]\subset\phi\left(\left[-\frac{1}{2},\frac{1}{2}\right]\times\left[-\frac{1}{2},\frac{1}{2}\right]\right)$. Consequently $\phi^{-1}$ extends to a continuous map on $[-1,1]\times[-1,1]$. 

Let $\hat{M} :=\phi^{-1}\left(\left(-1,1\right)\times\left(-1,1\right)\right)\subset\mathbbm{R}^2 $ with global coordinates $(p,q)$. Define a map $\iota:M_t\hookrightarrow \hat{M}$ in terms of coordinates by $\iota(u,v)=\phi^{-1}(u,v)$. Since $\phi^{-1}$ is smooth on $(-1,1)\times(-1,0)$, it follows that $\iota$ defines a smooth embedding and $\iota(M_t)=\phi^{-1}\left((-1,1)\times(-1,0)\right)\subsetneq \hat{M}$. Furthermore, for $(p,q)\in\iota(M_t)$, the metric
%The map $\phi$ is smooth on $M\setminus\{(0,0)\}$, so it follows that $\hat{g}$ is smooth on $\hat{M}\setminus\{(0,0)\}$. We now show that $\hat{h}$ is continuous at $(0,0)$. 
\begin{equation}\label{eqn:pushfwdmetric}
\begin{split}
    \hat{g}(p,q)&:=\iota_*g_t(p,q)\\
   % &=-\Omega^2(u(p,q),v(p,q))\iota_*(du\otimes dv+dv\otimes du)(p,q)\\
    &= -\Omega^2(u(p,q),v(p,q))\left[2\frac{\partial u}{\partial p}\frac{\partial v}{\partial p}dp\otimes dp +\left(\frac{\partial u}{\partial p}\frac{\partial v}{\partial q}+\frac{\partial u}{\partial q}\frac{\partial v}{\partial p}\right)\left(dp\otimes dq+dq\otimes dp\right)\right.  \\
    & \hspace{5cm} \left. +2\frac{\partial u}{\partial q}\frac{\partial v}{\partial q}dq\otimes dq\right].
    \end{split}
\end{equation}
is also smooth on $\iota(M_t)$. 

We now show that $\hat{g}$ can be extended to a continuous metric on $\hat{M}$. For $q>0$, we define $\hat{g}(p,q):=\hat{g}(p,-q)$. This gives a metric which is continuous (in fact smooth) on $\hat{M}\setminus\{q=0\}$. We claim that this metric extends continuously to $q=0$. From \eqref{eqn:partialderivs}, we see that $\hat{g}$ extends continuously to $\hat{M}\setminus\{(0,0)\}$ with 
\begin{equation} \label{Eqgq0}
    \begin{split}
       \hat{g}(p,0):= %\lim_{q\rightarrow0}\hat{g}(p,q)=
       -C^2\Omega^2(u(p,0),0)\left[\left(dp\otimes dq+dq\otimes dp\right)-\frac{2\sqrt{2}\pi \cdot \mathrm{sgn} (p)}{\log p^2}dq\otimes dq\right]
    \end{split}
\end{equation}
 (recall that $\Omega^2(u,v)$ was assumed to extend continuously to $v=0$ as a strictly positive function). 
 Furthermore, we have \begin{equation}\label{eqn:00continuity}
     \lim_{(p,q) \to (0,0)} \hat{g}(p,q)=-C^2\Omega^2(0,0)\left(dp\otimes dq+dq\otimes dp\right)\;.
 \end{equation}
 This is implied by the following estimates: for $q=0$, the required estimates follow directly from \eqref{Eqgq0}.
%\begin{equation}
%    \begin{split}
%        \lim_{p\rightarrow0}\hat{g}(p,0)=-C^2\Omega^2(0,0)\left(dp\otimes dq+dq\otimes dp\right).
%    \end{split}
%\end{equation}
%It is clear from \eqref{eqn:partialderivs} that this metric extends continuously (in fact smoothly) to $\left(\left(-\frac{1}{2},\frac{1}{2}\right)\times\left(-\frac{1}{2},\frac{1}{2}\right)\right)\setminus\{(0,0)\}$. 
Let now $(p,q)\in\iota(M_t)\subset\left(-\frac{1}{2},\frac{1}{2}\right)\times\left(-\frac{1}{2},0\right)$. From \eqref{eqn:partialderivs} we have 
\begin{equation}
\begin{split}
    \left|\frac{\partial u}{\partial p}\frac{\partial v}{\partial q}-C^2\right|&=\left|2C^2q\log(p^2+2q^2)\int_0^q\frac{dt}{(\log(p^2+q^2+t^2))^2(p^2+q^2+t^2)}\right|\\
    &\leq \left|\frac{2C^2q^2}{\log(p^2+2q^2)(p^2+q^2)}\right|\\
    &\leq \frac{2C^2}{\left|\log(p^2+2q^2)\right|}\\
    &\rightarrow0\text{ as }(p,q)\rightarrow(0,0).
\end{split}
\end{equation}

Since $\frac{\partial u}{\partial p}\rightarrow\infty$ as $(p,q)\rightarrow(0,0)$, it follows from this that $\frac{\partial v}{\partial q}\rightarrow0 \text{ as }(p,q)\rightarrow(0,0).$
%\begin{equation}\label{eqn:dvdq}
%    \frac{\partial v}{\partial q}\rightarrow0 \text{ as }(p,q)\rightarrow(0,0).
%\end{equation}
We also have $\left|\frac{\partial u}{\partial q}\right|<C\sqrt{2}\pi$ and hence $\frac{\partial u}{\partial q}\frac{\partial v}{\partial q}\rightarrow0$ as $(p,q)\rightarrow(0,0)$. 
%Furthermore,
%\begin{equation}\label{eqn:dvdp}
%    \begin{split}
%        \left|\frac{\partial v}{\partial p}\right|&\leq\frac{2C|pq|}{(\log(p^2+2q^2))^2(p^2+q^2)}\\
%        &\leq\frac{C}{(\log(p^2+2q^2))^2}\\
%        &\rightarrow0\text{ as }(p,q)\rightarrow(0,0).
%    \end{split}
%\end{equation}
Furthermore,
\begin{equation}
\begin{split}
       \left|\frac{\partial u}{\partial p}\frac{\partial v}{\partial p}\right|&\leq \frac{2C^2|pq|}{(p^2+q^2)|\log(p^2+2q^2)|}\\
       &\leq\frac{C^2}{|\log(p^2+2q^2)|}\\
   &\rightarrow0\text{ as }(p,q)\rightarrow(0,0).
\end{split}
\end{equation}
Again since $\frac{\partial u}{\partial p}\rightarrow\infty$ as $(p,q)\rightarrow(0,0)$, it follows that $\left|\frac{\partial v}{\partial p}\right|\rightarrow0$ as $(p,q)\rightarrow(0,0)$ and hence $\frac{\partial u}{\partial q}\frac{\partial v}{\partial p}\rightarrow0$ as $(p,q)\rightarrow(0,0)$. This establishes \eqref{eqn:00continuity} and hence the continuity of $\hat{g}$ at $(0,0)$.
%It follows that
%\begin{equation}
%    \iota_*g(p,q)\rightarrow -C^2\Omega^2(0,0)\left(dp\otimes dq+dq\otimes dp\right) \text{ as }(p,q)\rightarrow(0,0)
%\end{equation}
%and hence we can extend $\hat{g}$ to define the following continuous metric on $\hat{M}$: }
%\begin{equation}
%    \hat{g}(p,q):=\begin{cases}
%        \iota_*g(p,-|q|) &\text{ for }q\neq0\\
%        \lim_{q'\rightarrow0^-}\iota_*g(p,q')&\text{ for }q=0.
%    \end{cases}
%\end{equation}
We conclude that $\iota$ defines a $C^0$-extension of $(M_t,g_t)$ into $(\hat{M},\hat{g})$.

\textbf{Step 6:} We claim that  $\iota$ is a $C^0$-extension of $M_t$ across $\{v=0\}$ which is locally $C^0$-equivalent to $\overline{M_t}$ at $(u,0)$ for all $u\in(-1,1)$.\footnote{Recall the reference extension $\overline{M_t}$ from Section \ref{toymodel}.} 

For each $u\in(-1,1)$, $\gamma^u$ is a future-directed null curve which is future inextendible in $M_t$. However, $\lim_{v\rightarrow0}(\iota\circ\gamma^u)(v)=\phi^{-1}(u,0)\in\hat{M}$, so $\iota : M_t \hookrightarrow \hat{M}$ defines a $C^0$-extension of $M_t$ across $v =0$. The map $\phi^{-1} : (-1,1) \times (-1,0] \to \hat{M}$ is a homeomorphism onto its image, so we conclude that $\iota$ and $\overline{M_t}$ are locally $C^0$-equivalent at $(u,0)$ for all $u\in(-1,1)$.

%, so it follows that 
%\begin{equation}
%    \partial^+\iota(M)=\{\phi^{-1}(u,0):u\in(-1,1)\}.
%\end{equation}

%For $v\in(-1,0)$, the injectivity and continuity of $\phi^{-1}$ imply
%\begin{equation}
%\lim_{u\rightarrow1}\iota(u,v)=\phi^{-1}(1,v)\notin\hat{M}
%\end{equation}
%and hence curves of constant $v$ are not future extended by $\iota$. However, the continuity of $\phi^{-1}$ implies that constant $u$ curves are future extended by $\iota$, so we conclude that

%From Example \ref{example:referenceextensionconformal}, we also have 
%\begin{equation}
%\begin{split}
%\iota(M)&=(-1,1)\times(-1,0)\subset\mathbbm{R}^2\\
%    \partial^+\iota_{ref}(M)&=(-1,1)\times\{0\}\subset\mathbbm{R}^2.
%    \end{split}
%\end{equation}
%Moreover, for each $u\in(-1,1)$, the point $(u,0)\in\partial^+\iota_{ref}(M)$ is anchored to the point $\phi^{-1}(u,0)\in\partial^+\iota(M)$ by the causal curve $\gamma^u$ which is future inextendible in $M$.
%We are therefore interested in how the identification map \eqref{eqn:identificationC1inequivalent} and its derivatives extend to $(-1,1)\times(-1,0]$.

%As noted above, $\phi^{-1}$ is continuous on $[-1,1]\times[-1,1]$, so in particular $id$ extends continuously to $(-1,1)\times(-1,0]$.
%We see from (\ref{ref:uvdefinition}) that as $v'\rightarrow0$ we have $q(u',v')\rightarrow0$.  and $u'(p,q)\rightarrow u'(p,0)=Cp(2-\log p^2)$. By the inverse function theorem, this function is invertible for $p\in(-1,1)$ with smooth inverse $p(u)$ defined on the inverse image of $(-1,1)$. 
%It follows that the two extensions $\iota_{ref}$ and $\iota$ are $C^0$-equivalent. 

\textbf{Step 7:} Finally we show that $\iota$ is not locally $C^1$-equivalent to $\overline{M_t}$ at the point $(0,0)\in\overline{M_t}$.

Inverting the matrix of partial derivatives $\frac{\partial(u,v)}{\partial(p,q)}$, we obtain
\begin{equation}
\begin{split}
   \frac{\partial q}{\partial v}=\frac{1}{Jac(p,q)}\frac{\partial u}{\partial p}
\end{split}
\end{equation}

for $(u,v)\in(-1,1)\times(-1,0)$, where $Jac(p,q)=\frac{\partial u}{\partial p}\frac{\partial v}{\partial q}-\frac{\partial u}{\partial q}\frac{\partial v}{\partial p}\rightarrow C^2>0$ as $(p,q)\rightarrow (0,0)$ (i.e. as $(u,v)\rightarrow(0,0)$). We also have $\frac{\partial u}{\partial p}=-C\log(p^2+2q^2)\rightarrow\infty$ as $(p,q)\rightarrow(0,0)$. It follows that $\frac{\partial q}{\partial v}\rightarrow\infty$ as $(p,q)\rightarrow(0,0)$, and hence $\iota$ does not extend as a $C^1$-diffeomorphism to $(-\epsilon,\epsilon)\times(-\epsilon,0]$ for any $\epsilon>0$. 
\end{example}

\subsection{Application to weak null singularities in 3+1-dimensions}

Finally we return to our motivating example: weak null singularities $(\Mw, \gw)$ in 3+1-dimensions (see the beginning of the introduction).
Example \ref{example:C^1inequivalent2dimensions} can be used to construct a $C^0$-extension of $(\Mw,\gw)$ which is locally $C^0$-equivalent to $\overline{\Mw}$ at $(0,0,\theta^A)$ for any $\theta^A\in\Sp^2$ but not locally $C^1$-equivalent at these points.

\begin{example} \label{ExWNS}
   Let $\hat{M}$ and  $\iota:M_t\hookrightarrow\hat{M}$ be as in Example \ref{example:C^1inequivalent2dimensions}. Let $\tilde{M}_{wns}=\hat{M}\times \Sp^2$ and define a map $\tilde{\iota}:\Mw\hookrightarrow\tilde{M}_{wns}$ in terms of coordinates $(u,v,\theta^A)$ on $\Mw$ and $(p,q,\tilde{\theta}^A)$ on $\tilde{M}_{wns}$ by
    \begin{equation}
        \tilde{\iota}(u,v,\theta^A)=(\iota(u,v),\theta^A).
    \end{equation} 
       Following Example \ref{example:C^1inequivalent2dimensions}, we have
\begin{equation}
    \begin{split}
&\tilde{\iota}_*\gw(p,q,\tilde{\theta})\\
&=\hat{g}_{wns}(p,q, \tilde{\theta})+\gamma_{AB}(\phi(p,q), \tilde{\theta})\left(d\tilde{\theta}^A-b^A(\phi(p,q), \tilde{\theta})\left(\frac{\partial v}{\partial p}dp+\frac{\partial v}{\partial q}dq\right)\right) \\
&\qquad \qquad \qquad \qquad \qquad \qquad \otimes \left(d\tilde{\theta}^B-b^B(\phi(p,q), \tilde{\theta})\left(\frac{\partial v}{\partial p}dp+\frac{\partial v}{\partial q}dq\right)\right)
%&\underset{q\rightarrow0}{\longrightarrow}\begin{cases}
%    \hat{g}_{wns}(p,0, \tilde{\theta})+\gamma_{AB}(\phi(p,0), \tilde{\theta})\left( d\tilde{\theta}^A+\frac{C}{\log p^2}b^A(\phi(p,0), \tilde{\theta})dq\right)\otimes \left( d\tilde{\theta}^B+\frac{C}{\log p^2}b^B(\phi(p,0), \tilde{\theta})dq\right)  &\text{for }p\neq0\\
%    \hat{g}_{wns}(0,0, \tilde{\theta})+\gamma_{AB}(\phi(p,0), \tilde{\theta})d\tilde{\theta}^A\otimes d\tilde{\theta}^B &\text{for }p=0.
%\end{cases} 
\end{split}
\end{equation}
where $\hat{g}_{wns}(p,q, \tilde{\theta})$ is of the form \eqref{eqn:pushfwdmetric} with $\Omega^2$ now also depending on $\tilde{\theta}$. Recall from Example \ref{example:C^1inequivalent2dimensions} that $\phi(p,q) = \big(u(p,q), v(p,q) \big)$.
Recall also that $\Omega^2$, $\gamma_{AB}$ and $b^A$ were assumed to extend continuously to $\{v=0\}$ as a strictly positive function, a Riemannian metric on $\Sp^2$, and a vector field on $\Sp^2$ respectively. Combining this with Example \ref{example:C^1inequivalent2dimensions}, it follows that $\hat{g}_{wns}$ extends continuously to $q=0$. Furthermore, it follows from \eqref{eqn:partialderivs} (see also Step 5 in Example \ref{example:C^1inequivalent2dimensions}) that $\frac{\partial v}{\partial p}(p,q)$ extends continuously to $q=0$ as $\frac{\partial v}{\partial p}(p,0) = 0$ and $\frac{\partial v}{\partial q}(p,q)$ extends continuously to $q=0$ as $\frac{\partial v}{\partial q}(p,0) = - \frac{C}{\log(p^2)}$.
 Hence we can define a continuous metric on $\tilde{M}_{wns}$ by
    \begin{equation}
                \tilde{g}_{wns}(p,q,\tilde{\theta}^A):=
                \begin{cases}
        \tilde{\iota}_*\gw(p,-|q|,\tilde{\theta}^A) &\text{ for }q\neq0\\
        \lim_{q'\rightarrow0^-}\tilde{\iota}_*\gw(p,q',\tilde{\theta}^A)&\text{ for }q=0.
        \end{cases}
        \end{equation}
It follows that $\tilde{\iota}$ defines a $C^0$-extension of $(\Mw,\gw)$ into $(\tilde{M}_{wns},\tilde{g}_{wns})$. Moreover, for $C>0$ sufficiently large, the vector field $C \partial_u + \partial_v$ is future-directed timelike in a small neighbourhood of $(u, 0,\theta^A)\in\overline{M_{wns}}$. Its integral curve through $(u, 0,\theta^A)$ defines a future-directed causal curve $\tau : [-\delta, 0) \to \Mw$ with $\lim_{s \to 0} \tau^{v}(s) = 0$, $\lim_{s \to 0} \tau^u(s) <1$ and such that $\lim_{s \to 0}(\iota \circ \tau)(s) \in \rd \tilde{\iota}(\Mw) \subset \tilde{M}_{wns}$ exists. Hence $\tilde{\iota}$ defines a $C^0$-extension of $\Mw$ across $v =0$.

We now compare this extension to the reference extension, $\overline{\Mw}$, defined in Section \ref{Introduction}. 
Consider a fixed point $(u, 0,\theta^A)\in \rd\overline{\Mw}$.  Recall from Example \ref{example:C^1inequivalent2dimensions} that $\iota$ is locally $C^0$-equivalent to $\overline{M_t}$ at $(u,0)\in\overline{M_t}$ for all $u\in(-1,1)$ but locally $C^1$-inequivalent at $(0,0)$. It follows immediately that, for any $(u,\theta^A)\in(-1,1)\times\Sp^2$, $\tilde{\iota}$ is locally $C^0$-equivalent to $\overline{\Mw}$ at $(u,0,\theta^A)\in\overline{\Mw}$ but locally $C^1$-inequivalent at $(0,0,\theta^A)$.
\end{example}

\bibliographystyle{amsplain}
\bibliography{bibliography.bib}
\end{document}

%% file: Surjectivity.pdf_tex
%% Creator: Inkscape 1.3.2 (091e20e, 2023-11-25), www.inkscape.org
%% Accompanies image file 'Surjectivity.pdf' (pdf, eps, ps)
\begingroup%
  \makeatletter%
  \providecommand\color[2][]{%
    \errmessage{(Inkscape) Color is used for the text in Inkscape, but the package 'color.sty' is not loaded}%
    \renewcommand\color[2][]{}%
  }%
  \providecommand\transparent[1]{%
    \errmessage{(Inkscape) Transparency is used (non-zero) for the text in Inkscape, but the package 'transparent.sty' is not loaded}%
    \renewcommand\transparent[1]{}%
  }%
  \providecommand\rotatebox[2]{#2}%
  \newcommand*\fsize{\dimexpr\f@size pt\relax}%
  \newcommand*\lineheight[1]{\fontsize{\fsize}{#1\fsize}\selectfont}%
  \ifx\svgwidth\undefined%
    \setlength{\unitlength}{462.00394758bp}%
    \ifx\svgscale\undefined%
      \relax%
    \else%
      \setlength{\unitlength}{\unitlength * \real{\svgscale}}%
    \fi%
  \else%
    \setlength{\unitlength}{\svgwidth}%
  \fi%
  \global\let\svgwidth\undefined%
  \global\let\svgscale\undefined%
  \makeatother%
  \begin{picture}(1,0.49104862)%
    \lineheight{1}%
    \setlength\tabcolsep{0pt}%
    \put(0,0){\includegraphics[width=\unitlength,page=1]{Surjectivity.pdf}}%
    \put(0.41151734,0.46796882){\color[rgb]{0,0,0}\makebox(0,0)[lt]{\lineheight{1.25}\smash{\begin{tabular}[t]{l}$p$\end{tabular}}}}%
    \put(0.61899256,0.27687329){\color[rgb]{0,0,0}\makebox(0,0)[lt]{\lineheight{1.25}\smash{\begin{tabular}[t]{l}$q$\end{tabular}}}}%
    \put(0.253181,0.35331144){\color[rgb]{0,0,0}\makebox(0,0)[lt]{\lineheight{1.25}\smash{\begin{tabular}[t]{l}$u>1$\end{tabular}}}}%
    \put(0.47703582,0.38880063){\color[rgb]{0,0,0}\makebox(0,0)[lt]{\lineheight{1.25}\smash{\begin{tabular}[t]{l}$v=0$\end{tabular}}}}%
    \put(0.51798494,0.15948583){\color[rgb]{0,0,0}\makebox(0,0)[lt]{\lineheight{1.25}\smash{\begin{tabular}[t]{l}$u=0$\end{tabular}}}}%
    \put(0.26683066,0.1349166){\color[rgb]{0,0,0}\makebox(0,0)[lt]{\lineheight{1.25}\smash{\begin{tabular}[t]{l}$v<-1$\end{tabular}}}}%
    \put(-0.00070309,0.20043504){\color[rgb]{0,0,0}\makebox(0,0)[lt]{\lineheight{1.25}\smash{\begin{tabular}[t]{l}$\mathrm{graph}(p_1)$\end{tabular}}}}%
    \put(0.64356198,0.33420195){\color[rgb]{0,0,0}\makebox(0,0)[lt]{\lineheight{1.25}\smash{\begin{tabular}[t]{l}$\mathrm{graph}(q_1)$\end{tabular}}}}%
    \put(0,0){\includegraphics[width=\unitlength,page=2]{Surjectivity.pdf}}%
  \end{picture}%
\endgroup%